\documentclass[acmsmall,screen,dvipsnames,x11names
,nonacm
]{acmart}

\usepackage{mathtools}
\usepackage{amsthm}

\usepackage{cleveref}

\usepackage{graphicx} 
\usepackage[disable]{todonotes}

\def\extended{0}
\usepackage{olmacros}

\usepackage{mathpartir}
\usepackage{enumerate}
\usepackage{enumitem}

\usepackage[T1]{fontenc}
\usepackage{graphicx}
\usepackage[font=small,skip=0pt]{caption}
\usepackage{xspace}
\newcommand\TOL{\textsf{TOL}\xspace}
\usepackage{xr}
\usepackage{thmtools} 

\makeatletter
\def\@acmplainindent{0pt}
\def\@acmdefinitionindent{0pt}
\def\@proofindent{\noindent}
\makeatother

\ifx\extended\undefined
\fi

\usetikzlibrary{patterns,decorations.pathreplacing,arrows,arrows.meta,decorations.pathmorphing,positioning,fit,trees,shapes,shadows,automata,calc,calligraphy,tikzmark}

\usetikzlibrary{cd}
\tikzcdset{3d cd/.style={/tikz/every odd row/.append style={xshift={#1}}}}

\tikzset{
	contrapos/.style={
		draw=mygreen,
		thick,
		dotted,
		line width=.75pt
	},
	implication/.style={
		-implies,double equal sign distance, thick,shorten <=1pt,shorten >=1pt
	},
}

\newtheorem{remark}{Remark}

\title{Total Outcome Logic: Unified Reasoning for a Taxonomy of Program Logics}

\author{James Li}
\email{jl3594@cornell.edu}
\affiliation{%
  \institution{Cornell University}
  \country{USA}
}

\author{Noam Zilberstein}
\email{noamz@cs.cornell.edu}
\orcid{0000-0001-6388-063X}
\affiliation{%
  \institution{Cornell University}
  \country{USA}
}

\author{Alexandra Silva}
\email{alexandra.silva@cornell.edu}
\orcid{0000-0001-5014-9784}
\affiliation{%
  \institution{Cornell University}
  \country{USA}
}

\definecolor{mygreen}{HTML}{3C8031}
\definecolor{myorange}{HTML}{F58137}
\definecolor{rot}{RGB}{200,34,84}

\newcommand{\mynode}[2]{
\scriptsize
\def\arraystretch{.75}
\arraycolsep=0pt
\begin{array}{cc}
  #1
  \\
  {\color{NavyBlue}\tiny\textsf{#2}}
\end{array}
}
\newcommand{\taxnode}[3]{
\scriptsize
\def\arraystretch{.9}
\arraycolsep=0pt
\begin{array}{cc}
  #1
  \\
  \scalebox{.75}{$ #2 $}
  \\
  {\color{NavyBlue}\tiny\textsf{#3}}
\end{array}
}

\begin{document}

\begin{abstract}
  While there is a long tradition of reasoning about (non)termination in program analysis, specialized logics are typically needed to give different termination criteria. This includes partial correctness, where termination is not guaranteed, and total correctness, where it is guaranteed. We present \emph{Total Outcome Logic} (\TOL), a single logic which can express the full spectrum of termination conditions and program properties offered by the aforementioned logics.
\TOL extends (non)termination and (in)correctness reasoning across different kinds of branching effects, so that a single metatheory powers this reasoning in different kinds of programs, including  nondeterministic and probabilistic. 
We also show that \TOL subsumes several recently created taxonomies of (in)correctness logics, so that many different kinds of properties can be proven with a single unified theory.
\end{abstract}

\maketitle

\section{Introduction}
\label{sec:intro}

The study of program logics is heavily intertwined with the \emph{computational effects} that can occur in the underlying programs. Prime examples of effects include nondeterminism and nontermination.
In terms of (non)termination, traditional program logics such as Hoare Logic initially revolved around \emph{partial correctness}---properties that hold \emph{if} the program terminates \cite{hoarelogic,FLOYD67}. Partial correctness was attractive, as it offered a simple way to describe loop behaviors via invariants while still providing \emph{safety}---the property that the program never displays undesirable behaviors \cite{lamport1977proving}.
Soon after, the notion of \emph{total correctness} was introduced, to additionally guarantee termination \cite{manna1974axiomatic}. More recently, other program logics for proving the possibility of nontermination have been explored too \cite{raad2024nontermproving}.

Orthogonally, different logics have been proposed to deal with nondeterminism in various ways. Traditional \emph{correctness} logics use a demonic interpretation of nondeterminism where the postcondition applies to all reachable states \cite{hoarelogic}. However, interest in \emph{possible correctness} \cite{wpp} and \emph{incorrectness} \cite{il,moller2021algebra,outcome} has emerged, which demand angelic nondeterminism (the postcondition applies to \emph{some} reachable state).

To make sense of these different logics, various recent efforts have been made to build taxonomies highlighting the similarities and differences between them \cite{verscht2025taxonomy,cousot2024calculational,ascari2023sufficient,zhang:quantitative-sp}.
In particular, \citet{verscht2025taxonomy} identified the following three dimensions for program logics:
\begin{enumerate}
\item \emph{Correctness} (being able to reach) vs. \emph{incorrectness} (reachability)
\item \emph{Totality} (termination) vs. \emph{partiality} (possible nontermination)
\item \emph{Angelic} vs. \emph{demonic} nondeterminism
\end{enumerate}
Though extensive, this characterization is incomplete. As \citet{verscht2025taxonomy} already acknowledge: after presenting a taxonomy of 16 logics, they identify two more \emph{in-between} ones! But even more dimensions are missing. For one, nondeterminism need not be treated in a strictly angelic or demonic fashion, but rather in a gradient that can be captured by considering multiple \emph{reachable outcomes}. As demonstrated by Outcome Logic \cite{outcome,zilberstein:completeOL,zilberstein2024outcome,zilberstein2024unified}, this more expressive treatment of nondeterminism is needed in order to build logics and static analyses that unify correctness \emph{and} incorrectness reasoning.
In addition, branching may not arise from nondeterminism at all, but rather from other computational effects such as probabilistic choice. So the natural question arises: \emph{is there a taxonomy that captures all these additional dimensions?}

 In this paper we take a different approach than \citet{verscht2025taxonomy}, \citet{cousot2024calculational}, and \citet{ascari2023sufficient}. We present a single expressive logic, which subsumes the aforementioned taxonomies, while also providing a higher degree of expressivity to deal with \emph{in-between} cases. As such, this paper is not only about understanding the similarities and differences between different logics, but also about providing a common framework for reasoning about those properties.

Compositionality is the key to developing these unified reasoning principles. We start with a compositional denotational semantics where sequential composition is defined \emph{monadically} and loops are defined as the least fixed point of a Scott continuous operation. However, our desire to reason about nontermination in conjunction with nondeterminism brings about well-known barriers to compositionality: unbounded nondeterminism makes loop semantics non-compositional \cite{plotkin:countable-nondeterminism}.
As an illustration, consider the command $x \coloneqq\bigstar$, which nondeterministically assigns a natural number to the variable $x$.
We now ask, does the following program terminate? 
\[
  x \coloneqq\bigstar \fatsemi \whl{x > 0}{x \coloneqq x-1}
\]
Based on naive intuitions, the answer appears to be yes; for any natural number $n$ assigned to $x$, the loop will terminate after precisely $n$ iterations. The problem is that we cannot make this argument \emph{compositionally}. Indeed, a nonterminating trace exists in every finite approximation of the loop. This means that standard techniques for proving termination such as loop variants and ranking functions will not suffice.


Curiously, the issues with unbounded choice do not arise in conjunction with other computational effects such as probabilistic choice. Termination of probabilistic programs has been studied extensively \cite{mciver2005abstraction,kaminski,mciver2018new,feng2023lower,majumdar2025sound}, and there are many examples of probabilistic programs with infinitely many outcomes, which \emph{almost surely terminate}---they terminate with probability 1. For example, one can write a probabilistic  program that computes a geometric distribution, where the probability that $x = n$ is $\frac1{2^n}$ for all $n\ge 1$. The probability of termination after $n$ iterations is $1 - \frac1{2^n}$, which clearly converges to 1 using a simple compositional proof.

In this paper, we identify the key difference between these types of choice. Probabilistic programs are \emph{conservative}---probability mass is treated like a resource, with the total amount of mass (1) being split between branches of computation and nontermination, so the weight of infinite traces can converge to 0 in the limit. Nondeterminism is \emph{indicative}---there is an infinite amount of total \emph{weight}, and spawning new branches does not consume anything.

Our formalism follows the weighted programming paradigm, where the possible program end states are assigned weights \cite{batz2022weighted,zilberstein:completeOL}. Varying the representations of the weights gives interpretations of different effects. For example, Boolean weights indicate whether or not a state is a possible outcome of a nondeterministic program whereas real-valued weights are used to quantify the probabilities of outcomes in a probabilistic program.
Weights are elements of a \emph{semiring}, meaning that they can be added and multiplied to provide semantics to branching and sequential composition, respectively. In conservative weightings---like probabilistic programs---it is possible to have compositional inference rules for termination with unbounded choice, whereas in indicative weightings it is not possible.
Our full list of contributions is as follows:
\begin{itemize}[leftmargin=*]
  \item Building on Outcome Logic \cite{zilberstein:completeOL}, we define a generic program semantics, which can also quantify the weights of infinite traces (\Cref{sec:semantics}). We show how the ability of each model to make unbounded choices corresponds to particular properties of the underlying semiring.
  \item We extend Outcome Logic with the ability to reason about nontermination to obtain Total Outcome Logic (\Cref{sec:logic}). Unlike prior work, this allows us to prove that programs always terminate, or sometimes do not terminate. 
  \item We show that \cites{verscht2025taxonomy} entire taxonomy of \emph{Hoare-like} logics can be expressed in Total Outcome Logic (\Cref{sec:vk-taxonomy}).
  \item We show that a second taxonomy---\emph{the \citet{cousot2024calculational} cube}---can also be expressed in Total Outcome Logic and compare it to the aforementioned taxonomy of \citeauthor{verscht2025taxonomy} (\Cref{sec:cousot}).
  \item We present a new taxonomy of correctness and incorrectness logics, showing that Total Outcome Logic gives more flexibility to express reusable specifications, making it a good foundation for static analysis (\Cref{sec:new-taxonomy}).
  \item We derive proof rules to reason more easily about special cases (\Cref{sec:derived}).
  \item We present a variety of case studies in \Cref{sec:examples} to demonstrate how to reason about termination and nontermination with Total Outcome Logic.
\end{itemize}
Omitted proofs are included in the appendix.
\section{Weighted Program Semantics}
\label{sec:semantics}

We consider a simple {\em weighted} programming language in the style of \citet{batz2022weighted} and \citet{zilberstein:completeOL}. We will present the syntax and semantics in this section, while highlighting our new addition: the ability to quantify weights of infinite traces in a denotational model.

\begin{figure}
\begin{align*}
  \textsf{Commands} \ni  C &\Coloneqq \:\: \skp 
    \:\:|\:\: \assume e
    \:\:|\:\: a \in \mathsf{Act}
    \:\:|\:\: C_1 \seq C_2
    \:\:|\:\: C_1 + C_2 
    \:\:|\:\: \iter{C}{e_1}{e_2}
    \\
  \textsf{Guards} \ni      e \ &\Coloneqq \:\: b \mid u\in U
    \\
  \textsf{BExp} \ni      b \ &\Coloneqq \:\: \tru \mid \fls \mid b_1\lor b_2 \mid b_1\land b_2 \mid \lnot b \mid t\in\mathsf{Test}
\end{align*}
\caption{Syntax of a simple weighted imperative language, parametric on a set of basic program actions $\textsf{Act}$ and a set of weights $U$.}\label{fig:lang}
\end{figure}

\subsection{Syntax}\label{sec:syntax}
The syntax is presented in~\Cref{fig:lang}. The language includes the usual syntax for the program that does nothing $\skp$, assertions $\assume e$, basic (uninterpreted) program actions $a\in \textsf{Act}$, and sequential composition $\seq$. We also include in the language a nondeterministic choice $C_1+C_2$, which is a program that nondeterministically chooses one of the branches to execute. One difference worth highlighting is that the $\assume e$ command takes a guard $e$ that is not just a usual Boolean expression. In this language guards can be Boolean---the expressions in $\textsf{BExp}$ are built from a basic set of tests $\mathsf{Test}$ using the usual Boolean operations---but they can also be {\em weights} $u\in U$. The intuition behind $\assume u$ is that the computation trace in which this is executed is weighed by $u$. This will become clearer below when we discuss the formal semantics.

Using Boolean tests, we can define syntactic sugar for if statements, shown below.
The syntax for a loop $\iter{C}{e_1}{e_2}$, due to \citet{zilberstein:completeOL}, is slightly unusual as it has two guards: $e_1$ is the guard for the body of the loop $C$ to continue executing, whereas $e_2$ is the guard for the exit of the loop. In a typical syntax, loops usually only have Boolean guard $b$ and  the exit guard is simply the negation of $b$. This sort of guarded iteration can be encoded in our language as follows:
\[
  \iftf b{C_1}{C_2}
  \triangleq
  (\assume b\fatsemi C_1) + (\assume{\lnot b}\fatsemi C_2)
  \qquad
  \whl bC \triangleq \iter Cb{\lnot b}
\]

\subsection{Semantics as Weighting Functions}\label{sec:sem-weight}
In order to provide the semantics of our language we first need to introduce a few algebraic definitions. To make sense of the weights in the context of computation traces, we will need to have structure on the set $U$ that enables us to combine different weights along a computation trace and across computation traces. In other words, how do we give semantics to expressions like: $(\assume 3) \seq a \seq (\assume 5)$ or $(\assume 3) \seq a_1 + (\assume 5) \seq a_2$?  
We will require the set of weights $U$ to have the structure of a {\em semiring}, which is an algebraic structure that combines two {\em monoids}. 
\begin{definition}[Monoid]
    A \textit{monoid} $\langle U, +, \zero \rangle$ consists of a carrier set $U$, an associative binary operation $+: U \times U \rightarrow U$, and an identity element $\zero$ ($u + \zero = \zero + u = u$). If $+: U \rightharpoonup U$ is partial, then $\langle U, +, \zero \rangle$ is a \textit{partial} monoid. If $+$ is commutative ($u + v = v + u$), then the monoid is \textit{commutative}. 
\end{definition}

\begin{definition}[Semiring]
    A (partial) \textit{semiring} $\langle U, +, \cdot, \0, \1 \rangle$ is an algebraic structure such that:
    \begin{enumerate}
        \item $\langle U, +, \0 \rangle$ is a (partial) commutative monoid and $\langle U, \cdot, \1 \rangle$ is a monoid.
        \item Distributivity: $u \cdot (v + w) = uv + uw$ and $(u + v) \cdot w = uw + vw$
        \item Annihilation: $\0 \cdot x = x \cdot \0 = \0$
    \end{enumerate}
\end{definition}

\begin{example}[Semirings]
The following semirings encode different kinds of computation.
\begin{enumerate}[leftmargin=*]

\item The Boolean semiring $\langle \mathbb{B}, \lor, \land, 0, 1\rangle$ uses Boolean $\mathbb{B} = \{0, 1\}$ weights, which indicate whether states are possible outcomes. Addition is disjunction and multiplication is conjunction.

\item The probabilistic semiring $\langle [0,1], +, \cdot, 0, 1\rangle$ uses real-valued probabilities as weights to quantify the likelihoods of outcomes. Addition is standard arithmetic addition, but it is partial since it is undefined if the sum is above 1. Multiplication is standard arithmetic multiplication.

\item The natural number semiring $\langle \mathbb{N}^\infty, +, \cdot, 0, 1\rangle$ uses natural numbers (plus $\infty$) as weights, with addition and multiplication given by standard arithmetic operations. This encodes nondeterminism, where we count the traces leading to each outcome, \ie as multisets of outcomes. 

\end{enumerate}
For more examples, refer to \citet{batz2022weighted} and \citet{zilberstein:completeOL}.
\end{example}

In order to define the semantics we will assume a set of program states $\Sigma$ and then assign to each command a map from program states to weighted, possibly nonterminating, traces. More precisely, we will build a denotational object in two steps. First, to encode information about nontermination, we expose divergence as a possible outcome of running a program. We will represent this with the element $\nonterm$, where for any set $S$, we write $S_{\nonterm} \triangleq S \cup \{\nonterm\}$. Second, we define weighting functions, which assign a weight to all states $\sigma\in\Sigma$, and to nontermination.
\begin{definition}[Weighting Function]
    Given a set of program states $\Sigma$ and a (partial) semiring $\mathcal{A} = \langle U, +, \cdot, \zero, \one \rangle$, the set of weighting functions is 
    \[
        \WA^\nonterm(\Sigma) \triangleq \bigg\{m: \Sigma_{\nonterm} \rightarrow U ~\bigg|~ |m| \text{ is defined and $\supp(m)$ is countable} \bigg\}
    \]
    The support $\supp(m) \triangleq \{ \sigma \mid m(\sigma) \neq \zero \}$ is the set of states to which $m$ assigns nonzero weight and the $\textit{mass}$ of $m \in \WA^\nonterm(\Sigma)$ written as $|m| \triangleq \sum_{\sigma \in \supp(m)} m(\sigma)$ is the cumulative weight in $m$. 
\end{definition}
We take a weighting function to represent a configuration in a computation, which may have branched into many different outcomes, including the nontermination outcome $\nonterm$. Captured here are two distinct computational effects: weighted branching and nontermination. We can now assign to each program command $C$ a denotation $\de{C}: \Sigma \rightarrow \WA^\infty(\Sigma)$, as shown in \autoref{fig:semantics}\footnote{The distinction between $\WA^\nonterm(\Sigma)$ and $\WA^\infty(\Sigma)$ will be explained when we discuss the semantics of loops.}.
\begin{figure}[t]
\begin{align*}
    \de{\skp}(\sigma) &\triangleq \eta(\sigma) &  \de{a}(\sigma) &\triangleq \de{a}_{\Act}(\sigma) \\
    \de{C_1 \seq C_2}(\sigma) &\triangleq \dem{C_2}(\de{C_1}(\sigma)) &
    \de{C_1 + C_2}(\sigma) &\triangleq \de{C_1}(\sigma) + \de{C_2}(\sigma) \\
    \de{\assume e}(\sigma) &\triangleq \de{e}(\sigma) \cdot \eta(\sigma) &
    \de{\iter{C}{e_1}{e_2}}(\sigma) &\triangleq \lfp f.\:\left(\Phi_{\iter{C}{e_1}{e_2}}\right)(f)(\sigma)
\end{align*}
\[
   \text{where}\quad \Phi_{\iter{C}{e_1}{e_2}}(f)(\sigma) \triangleq \de{e_1}(\sigma) \cdot f^{\dagger}(\de{C}(\sigma)) + \de{e_2}(\sigma)\cdot \eta(\sigma)
\]
\caption{Denotational semantics $\de{C}: \Sigma \rightarrow \WA^\infty(\Sigma_{\nonterm})$ of programs, parametric on a set of program states $\Sigma$; an interpretation $\de{a}_{\Act}\colon \Sigma \to \WA^\infty(\Sigma)$ for atomic actions $a\in \mathsf{Act}$; and $\Test \subseteq \bb{2}^{\Sigma}$, the set of primitive tests. }
\label{fig:semantics}
\end{figure}

We left several operations underspecified in \autoref{fig:semantics}, which we will now explain in detail, from simpler to more complex: 1. semantics of guards and branching; 2. monadic structure $\eta$ and $(-)^\dagger$ that provides semantics of $\skp$ and sequential composition; 3. fixpoint semantics of loops. 

\subsubsection*{Guards and Branching.} Recall that guards $e$ can be Boolean expressions $b$ or weights $u \in U$. We define the semantics of guard expressions $e$ as $\de{e} \colon \Sigma \rightarrow U$, where $\de{b}(\sigma) \triangleq \detest{b}(\sigma)$ and $\de{u}(\sigma) \triangleq u$.
The semantics of Boolean guards $\detest{b} \colon \Sigma \rightarrow \{\0, \1\}$ is a simple extension of the set-element predicate (where $\0, \1$ are the semiring identities), since primitive tests are sets of program states $\Test \subseteq \bb{2}^{\Sigma}$. The full definition is available in \Aref{app:defs}.

The semantics of $\assume e$ uses a product operation $\cdot$ on weighting functions whereas a   $+$ operation is used in the definition of $\de{C_1+C_2}$. These are operations on weighting functions that  are lifted from the semiring $+$ and $\cdot$ pointwise: $(m_1 + m_2)(\sigma) \triangleq m_1(\sigma) + m_2(\sigma)$ and $(u\cdot m)(\sigma) \triangleq u\cdot m(\sigma)$.

\subsubsection*{Monadic structure of weighting functions.} We next turn to the semantics of sequential composition and $\skp$, which is achieved using a monad \cite{manes1976algebraic,moggi91}.

\begin{definition}[Monads]
    A \textit{Kleisli triple} $\langle T, \eta, (-)^{\dagger} \rangle$ in the category \textbf{Set} consists of a functor $T: \textbf{Set} \rightarrow \textbf{Set}$, a natural transformation $\eta: \Id \Rightarrow T$, and, for any sets $X$ and $Y$, a map $(-)^{\dagger}: (X \rightarrow T(Y)) \rightarrow T(X) \rightarrow T(Y)$ such that $\eta_X^{\dagger} = \id_X$, $f^{\dagger} \circ \eta = f$, and $f^{\dagger} \circ g^{\dagger} = (f^{\dagger} \circ g)^{\dagger}$. If such a triple exists, then the functor $T$ is a monad.
\end{definition}

Let the \textit{characteristic} weighting function for an element $x \in X_\nonterm$ be $I_x: X_\nonterm \rightarrow U$ such that $I_x(y) = \one$ if $x = y$ and $\zero$ otherwise. 
\begin{definition}
    For any semiring $\mathcal{A}$, $\langle \WA^\nonterm, \eta, (-)^{\dagger} \rangle$ is a Kleisli triple, where 
    \begin{align*}
        \eta_X(x)(y) &\triangleq I_x(y) = 
        \begin{cases}
          \one &\text{if } x = y \\
          \zero &\text{if } x \neq y
        \end{cases}
        &
        f^{\dagger}(m)(y) &\triangleq \smashoperator{\sum_{x \in \supp(m) \cap \Sigma}} m(x) \cdot f(x)(y) + m(\nonterm) \cdot I_\nonterm(y) 
    \end{align*}
\end{definition}

Note how the Kleisli extension $f^{\dagger}(m)$ discriminates between weight allocated to program states (in $\Sigma$) and weight allocated to $\nonterm$ by $m$. This anticipates how we model sequencing programs $C_1$ and $C_2$ (see \cref{fig:semantics}); intuitively, the terminal outcomes of $C_1$ are fed as input to $C_2$ whereas information about nontermination of $C_1$ must be ``carried through'' the second command. Although we define the monad operations directly here, $\WA^\nonterm$ can also be defined via a distributive law \cite{beck1969distributive} between the weighting function monad of Outcome Logic \cite{zilberstein:completeOL} and the $+1$ or \emph{error} monad, which is known to compose with all other monads \cite{composingmonads}.

\subsubsection*{Loops.}  Finally, we discuss the semantics of iterated computations, which requires $\WA$ to be a directed complete partial order (dcpo).
We use a \emph{fusion order}: $m_1 \sqsubseteq m_2$ if and only if the weight of each program state increases and the weight of nontermination decreases. This is similar to the fixpoint maximal trace semantics of \citet{cousot2002constructive} in which finite traces are compared covariantly and infinite traces are compared contravariantly.
This order can also be viewed as an \emph{approximation order} \cite{abramsky1995domain}---$m_1 \sqsubseteq m_2$ means that $m_2$ contains more information about the program's behavior than $m_1$. As such, we initially consider the behavior to be nontermination (bottom), and the approximation gets larger as more and more terminating executions are discovered.

\begin{definition}[Fusion Order] \label{def:fusion-order}
    We define the \textit{fusion order} $\sqsubseteq$ on $\WA^\nonterm(\Sigma)$ as:
        \[
        m_1 \sqsubseteq m_2 \quad\text{iff}\quad
            m_1(\sigma) \leq m_2(\sigma)\ \text{for all $\sigma \in \Sigma$, and}\
            m_1(\nonterm) \geq m_2(\nonterm)
    \]
\end{definition}
Above, $\leq$ is the \emph{natural order} on $\mathcal{A}$: $u \leq v$ iff $u + w = v$ for some $w$. If $\langle U, \leq \rangle$ is a partial order, then the semiring is \textit{naturally-ordered}. The semiring is \textit{bounded} if there exists a top element $\top$ such that $u \leq \top$ for all $u \in U$. This top element is sometimes required to be an infinity, as defined below.
   
\begin{definition}[Infinity \cite{Golan1999SemiringsAT}]
    An element $w$ in semiring $\langle U, +, \cdot, \0, \1 \rangle$ is an \textit{infinity} if it is additively absorbing: $u + w = w + u = w$ for all $u\in U$. In addition, $w$ is a \textit{strong infinity} if it is infinite as well as multiplicatively absorbing: $x \cdot w = w \cdot x = w$ for all $x \neq \zero$.
\end{definition}

We also need a notion of {\em bounded weighting}, which generalizes {bounded nondeterminism} and is crucial for ensuring the standard continuity results for our semantics.
\begin{definition}[Bounded Weighting]
    We define the following classes of weighting functions:
    \begin{itemize}
        \item $\WA^+(\Sigma) \triangleq \{m \in \WA^\nonterm(\Sigma) \mid \supp(m) \text{ is finite}\}$
        \item $\WA^\omega(\Sigma) \triangleq \{m \in \WA^\nonterm(\Sigma) \mid m(\nonterm) = \sup_{m' \in E(m)} m'(\nonterm)\}$. 
        \item $\WA^\infty(\Sigma) \triangleq \WA^+(\Sigma) \cup \WA^\omega(\Sigma)$ 
    \end{itemize}
    For any particular $m \in \WA^\nonterm(\Sigma)$, let $E(m)$ be the set of weighting functions that assign equal weight to all program states in $\Sigma$ (\ie that can only differ in the weight on nontermination $\nonterm$):
    \[
        E(m) \triangleq \{m' \in \WA^\nonterm(\Sigma) \mid \forall \sigma \in \Sigma.\: m'(\sigma) = m(\sigma) \}
    \]
    A mapping $\Sigma \rightarrow \WA^\infty(\Sigma)$ is said to exhibit \textbf{bounded weighting}.
\end{definition}

Our notion approximates the operational behavior captured by bounded nondeterminism, which stipulates that a nondeterministic program can only select between finitely many branches at each step of execution. Drawing from the definition of \citet{back-unbounded} for a powerdomain semantics, this means that running a nondeterministic program must produce a set of outcomes that is either (i) finite or (ii) exhibits nontermination (such that infinite sets of outcomes must be generated by some nonterminating trace). 
The sets $\WA^+(\Sigma)$ and $\WA^\omega(\Sigma)$ are the analogue of Back's two cases, respectively; weightings produced by programs must either:
\begin{enumerate}
    \item Describe a finite set of outcomes (finite support) (that is, be in $\WA^+(\Sigma)$); or
    \item Maximize the weight on nontermination $\nonterm$  (that is, be in $\WA^\omega(\Sigma)$).
\end{enumerate}
This can also be seen as a generalization of the Plotkin powerdomain, which includes all finite sets of states and infinite states that contain $\nonterm$ \cite{plotkin1976powerdomain,s_ondergaard1992non}.
Here, the choice to \textit{maximize} nontermination lies in our approach to tracking the presence of nontermination in weighting functions. In particular, we support two semantic paradigms and correspondingly require the semiring of weights $\mathcal{A}$ to belong to one of two classes.
\begin{enumerate}[leftmargin=*]
    \item \textbf{\textit{Conservative Weighting}}: $\mathcal{A}$ is partial and bounded with \textit{non-idempotent} top element $\top$ such that, for all $x \in X\setminus\{\zero\}$, $x + \top$ and $\top + x$ are not defined. So, only $\top + \0 = \0 + \top = \top$. In this case, we choose our weighting functions to be drawn from $\mathcal{D} \triangleq \{m \in \WA^\infty(\Sigma) \mid |m| = \top \}$. Thus, programs under a \textit{conservative weighting} scheme always produce a weighting function of fixed mass $\top$, which is conserved across sequences of programs. Examples include probabilistic and deterministic programs, for which the mass of weighting functions is fixed at 1. 

    \item \textbf{\textit{Indicative Weighting}}: $\mathcal{A}$ is total and bounded with a \textit{strongly infinite} top element $\top$.
    Since the semiring is total, we can rewrite 
    \[
        \WA^\omega(\Sigma) \triangleq \{m \in \WA^\nonterm(\Sigma) \mid m(\nonterm) = \top\}
    \]
    and limit weighting functions to $\mathcal{D} \triangleq \WA^\infty(\Sigma)$. We refer to this as an \textit{indicative weighting} scheme in that the weight on $\nonterm$ serves as an indicator for nontermination; we assign a weight of $\0$ if the computation does not produce a diverging trace, and assign infinite weight otherwise. 
\end{enumerate}
At first, the introduction of these restrictions may seem to fetter the generalizability of our approach. However, we note that any semiring of weights $\mathcal{A} = \langle U, +, \cdot, \0, \1 \rangle$ that does not belong to either of the above classes can be extended with a strongly infinite element \cite[p. 166]{Golan1999SemiringsAT}. In particular, we adjoin to $\mathcal{A}$ an element $\infty \not\in U$ to accommodate an \textit{indicative} weighting scheme. Addition and multiplication can be defined such that 
 $\infty + u = u + \infty = \infty$ for all $u \in U$;
 $\infty \cdot u = u \cdot \infty = \infty$ for all $u \in U\setminus\{\0\}$; and
 $\infty \cdot \0 = \0 \cdot \infty = \0$.
Doing so limits our (indicative) semantics to a coarser representation of program behavior, foregoing attempts at tracking finer weights with which a program can diverge, and instead only maintaining whether nontermination is possible or not. Nevertheless, this is sufficient for most purposes, and in the case for some semirings, necessary to preserve continuity, which is defined below.

\begin{definition}[Scott-Continuity]
    A semiring $\langle U, +, \cdot, \0, \1 \rangle$ with order $\le$ is \textit{Scott-continuous} if for any directed set $D \subseteq U$ (where all pairs of elements in $D$ have a supremum):
    \[
        \sup_{u \in D} (u + v) = (\sup D) + v \quad 
        \sup_{u \in D} (u \cdot v) = (\sup D) \cdot v \quad 
        \sup_{u \in D} (v \cdot u) = v \cdot (\sup D)
    \]
    The semiring is \textit{lower Scott-continuous} if the same operations preserve infima.
\end{definition}
For Scott-continuous semirings, we can define an infinite sum operation over an index set $I$ as the supremum of sums over all finite subsets of $I$.
    \[
        \sum_{i \in I} u_i = \sup\:\left\{\sum_{i \in J} u_i \mid J \subseteq I \text{ finite}\right\}
    \]
This construction yields a \emph{complete} semiring, meaning that the sum operator obeys the following desirable laws \cite{kuich2011algebraic}:
    \begin{enumerate}
        \item For $I = \{i_1, \ldots, i_n\}$ finite, $\sum_{i \in I} u_i = u_{i_1} + \ldots + u_{i_n}$.
        \item If $I = \bigcup\limits_{k \in K} J_k$ is a disjoint union of nonempty sets, then $\sum\limits_{k \in K} \sum\limits_{j \in J_k} u_j = \sum\limits_{i \in I} u_i$. 
        \item If $\sum\limits_{i \in I} u_i$ is defined, then $v \cdot \sum\limits_{i \in I} u_i = \sum\limits_{i \in I} v \cdot u_i$ and $\bigg(\sum\limits_{i \in I} u_i\bigg) \cdot v = \sum\limits_{i \in I} u_i \cdot v$. 
    \end{enumerate}
The above definitions are what we need to fully define the denotational semantics for our language, and in particular the semantics of loops. In particular, we require $\mathcal{A} = \langle U, +, \cdot, \0, \1 \rangle$ to be a naturally ordered, finitary, Scott-continuous, partial semiring bounded by top element $\top$.

\section{Total Outcome Logic}
\label{sec:logic}

We now leverage our semantics to define \textit{Total} Outcome Logic (\TOL), building atop the proof system given by \citet{zilberstein:completeOL} to also reason about the possible divergence and certain termination of programs.
We will later see that \TOL is strictly more expressive than Outcome Logic; by exposing nontermination as a program outcome, \TOL is capable of expressing additional properties with termination and nontermination guarantees, including the classical notion of \textit{total correctness.}

\subsection{Outcome Assertions}

First, we define \textit{outcome assertions}, which serve as pre- and postconditions in \TOL. In many logics, pre- and postconditions are represented extensionally as \textit{sets of program states}. Outcome assertions, on the other hand, must operate on a finer level, specifying not only a collection of states but also their weights at the respective point in the computation. We give extensional definitions of outcome assertions $\varphi, \psi \in \bb{2}^{\WA^\infty(\Sigma)}$ as sets of weighting functions. For any $m \in \WA^\infty(\Sigma)$, we write $m \vDash \varphi$ (read as \textit{$m$ satisfies $\varphi$}) to mean $m \in \varphi$. 
Below, we introduce  notation for  assertions used later in constructing the logic. We start with basic propositional constructs, which are defined as usual:
\begin{align*}
    \top &\triangleq \WA^\infty(\Sigma) &
    \bot &\triangleq \emptyset &
    \neg \varphi &\triangleq \WA^\infty(\Sigma) \setminus \varphi \\
    \varphi \lor \psi &\triangleq \varphi \cup \psi &
    \varphi \land \psi &\triangleq \varphi \cap \psi &
    \varphi \Rightarrow \psi &\triangleq \lnot \varphi \vee \psi
\end{align*}
Given a predicate on program states $P \subseteq \Sigma$, we define liftings to indicate that $P$ over-approximates, under-approximates, or exactly characterizes the support of a weighting function with weight $u$. Similarly, $\Uparrow^{(u)}$ denotes that the nontermination outcome $\nonterm$ occurs with weight $u\cdot \top$. 
%
\begin{align*}
    \ov{P}^{(u)} &\triangleq \{m \mid |m| = u, \supp(m) \subseteq P\}
    &
    \un{P}^{(u)} &\triangleq \{m  \mid |m| = u, \supp(m) \supseteq P\}
    \\
    \ex{P}^{(u)} &\triangleq \{m  \mid |m| = u, \supp(m) = P\}
    &
    \Uparrow^{(u)} &\triangleq \{ u \cdot \top \cdot \eta(\nonterm) \}
\end{align*}
When the superscript is omitted, then the weight is exactly $\one$, \ie $\ov P = {\ov P}^{(\one)}$.
Operations $u \odot \varphi$ and $\varphi \odot u$ scale the outcome assertion $\varphi$ with weight $u$ on the left or right.
Existential quantification with respect to predicate $\phi: T \rightarrow \bb{2}^{\WA^\infty(\Sigma)}$ is defined as the union of all $\phi(t)$ for inputs $t \in T$. Then, $m \vDash \exists x : T.\: \phi(x)$ precisely when $m$ satisfies $\phi(t)$ for \textit{some} input $t$.
\begin{align*}
    u \odot \varphi &\triangleq \{u \cdot m \mid m \in \varphi\}  &
    \varphi \odot u &\triangleq \{m \cdot u \mid m \in \varphi\}  &
    \exists x : T.\: \phi(x) &\triangleq \textstyle\bigcup_{t \in T} \phi(t) 
\end{align*}
We write the \textit{outcome conjunction} $\varphi \oplus \psi$ to consist of weighting functions $m$ that can be split into components $m = m_1 + m_2$, with $m_1$ satisfying $\varphi$ and $m_2$ satisfying $\psi$. We define this more generally for a potentially infinite index set $T$:
\begin{align*}
    \bigoplus_{x \in T} \phi(x) &\triangleq \left\{\sum_{t \in T} m_t \mid \forall t \in T.\: m_t \in \phi(t)\right\}
\end{align*}
Finally, we assert \textit{identity} to $m$ with $\bone(m)$, which consists of the singleton $\{m\}$. 

\subsection{Total Outcome Logic}
The judgments of Total Outcome Logic are \textit{outcome triples}, which take the form $\triple{\varphi}{C}{\psi}$ for a precondition $\varphi$, command $C$, and postcondition $\psi$. 
\begin{definition}[Outcome Triples] The validity of outcome triples is defined as follows:
\[
    \vDash \triple{\varphi}{C}{\psi} \tiff \forall m \in \WA^\infty(\Sigma).\quad m\vDash \varphi \implies \dem{C}(m) \vDash \psi
\]
\end{definition}
The main reasoning apparatus in \TOL consists of inference rules for deriving triples, shown in \Cref{fig:inference-structural,fig:inference-command}. Some of these rules carry over from Outcome Logic \cite{zilberstein:completeOL}, while others must be adapted to account for nontermination. We introduce the rules as belonging to three main categories: structural rules, standard command rules, and iteration.

\begin{figure}[t]
\begin{mathpar}
    \inferrule*[right=False]{\:}{\triple{\bot}{C}{\varphi}}

    \inferrule*[right=True]{\:}{\triple{\varphi}{C}{\top}}

    \inferrule*[right=Div]{\:}{\triple{\Uparrow^{(u)}}{C}{\Uparrow^{(u)}}}

    \inferrule*[right=Scale]
        {\triple{\varphi}{C}{\psi}}
        {\triple{u \odot \varphi}{C}{u \odot \psi}}

    \inferrule*[right=Disj]
        {\triple{\varphi_1}{C}{\psi_1} \\ \triple{\varphi_2}{C}{\psi_2}}
        {\triple{\varphi_1 \lor \varphi_2}{C}{\psi_1 \lor \psi_2}}

    \inferrule*[right=Conj]
        {\triple{\varphi_1}{C}{\psi_1} \\ \triple{\varphi_2}{C}{\psi_2}}
        {\triple{\varphi_1 \land \varphi_2}{C}{\psi_1 \land \psi_2}}

    \inferrule*[right=Choice]
        {\forall t \in T.\: \triple{\phi(t)}{C}{\phi'(t)}}
        {\textstyle\triple{\bigoplus_{x \in T} \phi(x)}{C}{\bigoplus_{x \in T} \phi'(x)}}

    \inferrule*[right=Exists]
        {\forall t \in T.\: \triple{\phi(t)}{C}{\phi'(t)}}
        {\triple{\exists x : T.\: \phi(x)}{C}{\exists x : T.\: \phi'(x)}}

    \inferrule*[right=Consequence]
        {\varphi' \implies \varphi \quad \triple{\varphi}{C}{\psi} \quad \psi \implies \psi'}
        {\triple{\varphi'}{C}{\psi'}}
\end{mathpar}
\caption{Structural rules.}
\label{fig:inference-structural}
\end{figure}
\noindent\textit{Structural Rules} are provided in \autoref{fig:inference-structural} and hold for arbitrary program commands $C$, as they capture logical rules that apply to pre- and post-conditions. These include the trivial rules, \sruleref{False} and \sruleref{True}, which feature the vacuous precondition $\bot$ and the weakest postcondition $\top$, respectively. \sruleref{Div} handles nontermination, stating that assertions $\diverge{u}$ are preserved across programs. 

A given triple can be left-scaled by any weight $u \in U$ using \sruleref{Scale}. The rules \sruleref{Disj}, \sruleref{Conj}, and \sruleref{Choice} allow multiple triples to be conjoined via the connectives $\land$, $\lor$, and $\oplus$, respectively. Introduction of existential quantification is accomplished through \sruleref{Exists}. Finally, the standard \sruleref{Consequence} rule allows for strengthening preconditions and weakening postconditions. 
   
\noindent\textit{Standard Command Rules} are provided in \autoref{fig:inference-command} to specify how non-iterative program commands transform outcome assertions. \cruleref{Skip} and \cruleref{Seq} are standard for Hoare-style logics. The \cruleref{Assume} rule has the premise $\varphi \vDash e = u$, which is defined below:
    
\begin{definition}[Assertion Entailment] \label{def:assertion-entailment}
    Given an assertion $\varphi$, expression $e$, and weight $u \in U$, we write $\varphi \vDash e = u$ ($\varphi$ \textit{entails} $e = u$) iff for all $m\vDash\varphi$, $\nonterm \notin \supp(m)$ and $\de{e}(\sigma) = u$ for all $\sigma\in\supp(m)$.
\end{definition}

That is, for every weighting function $m\vDash\varphi$ and any $\sigma\in\supp(m)$, $e$ evaluates to $u$ in state $\sigma$.
For tests $b \in \Test$, we introduce the shorthand $\varphi \vDash b$ to mean $\varphi \vDash b = \1$. 

In \cruleref{Plus}, the side condition $\varphi \vDash \tru$ is equivalent to a sure-termination guarantee: for $m \in \varphi$, $\nonterm \:\not\in \supp(m)$. If this holds, the triples proved for both branches are combined with outcome conjunction. Sure-termination must be enforced in these cases as the additional effect of nontermination is sequenced differently than for regular weighted states. Similar issues arise in instances of OL with exceptions \cite{outcome,zilberstein2024outcome} and probabilistic nondeterminism \cite{zilberstein2025demonic}.

\subsubsection*{Iteration.} Our nontermination semantics culminates in the \textsc{Iter} rule for iteration commands $\iter{C}{e}{e'}$ (\autoref{fig:inference-command}), which also serves to introduce any assertions pertaining to nontermination. The form of this rule corresponds roughly to an unrolling of the loop and describes sequences of assertions that hold between iterations. These assertions come in three sorts:
\begin{itemize}[leftmargin=*]
    \item $(\varphi_n)_{n \in \N}$ describe ongoing lines of computation, \ie those that have not yet terminated. These include program configurations that are passed to the next iteration of the loop and subsequently transformed.  Note that we have as premise that $\varphi_n \vDash \tru$ for every $n$, so every $m \in \varphi_n$ describes only weights on program states $\Sigma$, but not the weight on $\nonterm$---we call computations restricted to the states $\Sigma$-computations.
    
    \item $(\psi_n)_{n \in \N}$ describe the terminating outcomes that are accrued in iteration $n$. In particular, observe that $\psi_n$ is obtained by filtering $\varphi_n$ through command $\assume e'$: $\triple{\varphi_n}{\assume e'}{\psi_n}$. 
    
    \item $(\zeta_n)_{n \in \N}$ describe nontermination encountered in the $n^\text{th}$ iteration (\ie via nested loops). We write $\zeta_n \vDash \div$ as a premise to enforce this characterization.
\end{itemize}

\begin{definition}[Nonterminating Assertions]
    For outcome assertion $\varphi$, we write $\varphi \vDash \div$ (read as: $\varphi$ is nonterminating) iff for all $m \in \varphi$, $\supp(m) \subseteq \{\nonterm\}$. 
\end{definition}

The conclusion of the rule has as postcondition two \textit{limiting} assertions $\psi_\infty$ and $\zeta_\infty$. Intuitively, the former describes the aggregate of all terminating outcomes accrued by each $\psi_n$. We represent this formally by writing $(\psi_n)_{n \in \N} \rightsquigarrow \psi_\infty$, which is defined below:

\begin{figure}[t]
\begin{mathpar}
    \inferrule*[right=Skip]{\:}{\triple{\varphi}{\skp}{\varphi}}

    \inferrule*[right=Seq]
        {\triple{\varphi}{C_1}{\zeta} \\ \triple{\zeta}{C_2}{\psi}}
        {\triple{\varphi}{C_1 \seq C_2}{\psi}}

    \inferrule*[right=Plus]
        {\varphi \vDash \tru \\ \triple{\varphi}{C_1}{\psi_1} \\ \triple{\varphi}{C_2}{\psi_2}}
        {\triple{\varphi}{C_1 + C_2}{\psi_1 \oplus \psi_2}}

    \inferrule*[right=Assume]
        {\varphi \vDash e = u}
        {\triple{\varphi}{\assume e}{\varphi \odot u}}
        
    \inferrule*[right=Iter]
        {(\psi_n)_{n \in \N} \rightsquigarrow \psi_{\infty} \\
        (\varphi_n \oplus \zeta_n)_{n \in \N} \Uparrow \zeta_\infty \\
        \forall n \in \N.\\
        \varphi_n \vDash \tru \\ 
        \zeta_n \vDash \div \\
        \triple{\varphi_n \oplus \zeta_n}{\assume e \seq C}{\varphi_{n+1} \oplus \zeta_{n+1}} \\
        \triple{\varphi_n}{\assume e'}{\psi_n}
        }
        {\triple{\varphi_0 \oplus \zeta_0}{\iter{C}{e}{e'}}{\psi_{\infty} \oplus \zeta_\infty}
        }
        
%

\end{mathpar}
\caption{Inference rules for program commands.}
\label{fig:inference-command}
\end{figure}

\begin{definition}[Converging Assertions]
    A family $(\psi_n)_{n \in \N}$ of assertions converges to $\psi_{\infty}$ (written $(\psi_n)_{n \in \N} \rightsquigarrow \psi_{\infty}$) iff for any $(m_n)_{n \in \N}$, if $m_n \vDash \psi_n$ for each $n \in \N$, then $\sum_{n \in \N} m_n \vDash \psi_{\infty}$. 
\end{definition}

On the other hand, $\zeta_\infty$ describes the degree of nontermination of the loop in the limit. This notion is captured by the semantic assertion $(\varphi_n \oplus \zeta_n)_{n \in \N} \Uparrow \zeta_\infty$, which we take to mean: 
\begin{definition}[Diverging Assertions]
    A family $(\psi_n)_{n \in \N}$ of outcome assertions diverges to $\psi_\infty$ (written $(\psi_n)_{n \in \N} \Uparrow \psi_\infty)$ iff for all $(m_n)_{n \in \mathbb{N}}$, if $m_n \vDash \psi_n$ for all $n \in \N$, then $(\inf_{n \in \N} |m_n| \cdot \top)\cdot\eta(\nonterm) \vDash\psi_\infty$.
\end{definition}
In other words, the total mass in $\varphi_n \oplus \zeta_n$ in the limit is transfered to the nontermination outcome $\nonterm$. Intuitively, nontermination in the loop has two possible sources: nonterminating weight from nested loops, which is accounted for by $\zeta_n$'s, and weight leftover from $\Sigma$-computations that never end, accounted for by the $\varphi_n$'s.
This rule allows us to explain the paradox that we introduced in \Cref{sec:intro}, in which terminating nondeterministic programs cannot make unboundedly many choices, but almost surely terminating probabilistic programs can. Consider the two programs below, both of which increment $x$ in a loop, but the program on the left is interpreted in the Boolean semiring, whereas the one on the right is interpreted in the probabilistic semiring.
\[
    x \coloneqq 0 \fatsemi \iter{(x \coloneqq x+1)}11
\qquad\qquad
    x \coloneqq 0 \fatsemi \iter{(x \coloneqq x+1)}p{1-p}
\]
For the nondeterministic program, we set $\varphi_n \triangleq \ov{x=n}$, $\psi_n \triangleq \ov{x=n}$, and $\zeta_n \triangleq \wg{\ov\tru}0$ for all $n\in\N$. Letting $\psi_\infty\triangleq \bigoplus_{k\in \N} \ov{x=k}$, clearly $\ov{x=n}_{n\in \N} \rightsquigarrow \bigoplus_{k\in\N} \ov{x=k}$, indicating that there is a terminating outcome where $x=k$ for any natural number $k$. Letting $\zeta_\infty \triangleq\mathord{\Uparrow}$, we also have $\ov{x=n}_{n\in\N} \Uparrow \zeta_\infty$; the weight of continuing to iterate remains to be 1 for all $n\in\N$, so there is a nonterminating trace of weight 1, giving us the final specification.
\[
    \triple{\ov\tru}{x \coloneqq 0 \fatsemi \iter{(x \coloneqq x+1)}11}{(\bigoplus_{k\in\N} \ov{x=k})\oplus \mathord\Uparrow}
\]
By contrast, in the probabilistic program we instead set $\varphi_n \triangleq \wg{\ov{x=n}}{p^n}$ and $\psi_n\triangleq \wg{\ov{x=n}}{p^n \cdot (1-p)}$. The terminating outcomes then form a geometric distribution $\bigoplus_{k\in\N} \wg{\ov{x=k}}{p^k\cdot(1-p)}$. Since the weight captured by each $\varphi_n$ decreases towards 0, the infimum is 0, giving us $\zeta_\infty \triangleq \wg{\ov\tru}0$, so the infinite execution occurs with zero probability, and it is therefore omitted from the specification.
\[
    \triple{\ov\tru}{x \coloneqq 0 \fatsemi \iter{(x \coloneqq x+1)}p{1-p}}{\bigoplus_{k\in\N} \wg{\ov{x=k}}{p^k \cdot (1-p)}}
\]

\begin{remark}[Accommodating Nontermination] In general, applying rules with some assertion entailment $\varphi \vDash e = u$ as a premise involves decomposing the precondition into components that describe terminating vs. nonterminating outcomes -- the latter of which are handled by \sruleref{Div} or some derivative of it.
As an example, take $\varphi = \varphi_t \oplus \varphi_{\nonterm}$ where $\varphi_t \vDash \tru$ and $\varphi_{\nonterm} \vDash \div$. We show a derivation involving \cruleref{Plus} across program $C_1 + C_2$, in which we are given premises $\triple{\varphi_t}{C_1}{\psi_{t_1}}$ and $\triple{\varphi_t}{C_2}{\psi_{t_2}}$:
\begin{mathpar}
    \inferrule*[Right=Choice]
        {\inferrule*[Right=Plus]
            {\triple{\varphi_t}{C_1}{\psi_{t_1}} \\
            \triple{\varphi_t}{C_2}{\psi_{t_2}}
            }
            {\triple{\varphi_t}{C_1 + C_2}{\psi_{t_1} \oplus \psi_{t_2}}} \\
        \qquad
        \inferrule*[Right=\textsc{Div*}]
            {\varphi_{\nonterm} \vDash \div}
            {\triple{\varphi_{\nonterm}}{C_1 + C_2}{\varphi_{\nonterm}}}
        }
        {\triple{\varphi_t \oplus \varphi_{\nonterm}}{C_1 + C_2}{\psi_{t_1} \oplus \psi_{t_2} \oplus \varphi_{\nonterm}}}
\end{mathpar}

\end{remark}

\subsubsection*{Soundness and Relative Completeness}

The proof system that we defined for Total Outcome Logic is sound and complete. Soundness means that every triple that is derivable via the inference rules is semantically valid. We state this theorem below, and prove it in the appendix.

\begin{theorem}[Soundness]\quad
$
     \vdash \triple{\varphi}{C}{\psi} \quad\implies\quad \vDash \triple{\varphi}{C}{\psi}.
$
\end{theorem}
Relative completeness means that any valid triple can be derived in our proof system, provided an oracle that decides semantic properties such as the implications used in the rule of \sruleref{Consequence} and the assertion entailments used in the \cruleref{Assume} rule.

\begin{theorem}[Relative Completeness]
\quad
$
    \vDash \triple{\varphi}{C}{\psi} \implies \vdash \triple{\varphi}{C}{\psi}
$
\end{theorem}

\section{The Verscht and Kaminski Taxonomy}
\label{sec:vk-taxonomy}

Starting with Hoare Logic, there is an expansive body of work in developing new program logics to reason across different computational domains and effects. More recently, an opposite impulse to unify these theories has taken shape in work on taxonomies of logic, which seek to abstract from these specific frameworks and uncover underlying formalisms that relate them.

One such example is \cites{verscht2025taxonomy} taxonomy of \emph{Hoare-like} logics, a collection of 16 logics for expressing different properties about nondeterministic, looping programs. These logics are roughly organized along three axes: \emph{correctness vs incorrectness}, \emph{partiality vs totality}, and \emph{angelic vs demonic} nondeterminism.
The taxonomy shows how 16 logics relate to each other in terms of weakening, contraposition, and Galois connections, while also proposing as-of-yet unexplored logics, including certain \emph{in-between} ones. We have adapted the taxonomy into a double cube structure, shown in \Cref{fig:vk-taxonomy}. Our taxonomy contains only 14 logics; two of \citeauthor{verscht2025taxonomy}'s logics are redundant due to Galois connections, as we explain in \Cref{sec:vk-sp}. The in-between logics, which appear in the interior of the upper cube, are added in \Cref{sec:vk-in-between}.

The taxonomy formulates logics in terms of predicate transformers---variants of \emph{weakest preconditions} and \emph{strongest postconditions}. Some of these transformers are well-known, while others are entirely new, neatly completing the picture.
In this section, we show that \TOL subsumes every logic in this taxonomy.
Interestingly, \TOL being lifted to reason directly on the sets of outcomes eliminates the need for rigid distinctions between logics, and instead allows us to specify many different kinds of properties just by altering the postcondition.
This is why, as we see in \Cref{sec:vk-in-between}, \TOL is readily suited for the in-between cases, and provides robust proof theory for a shared analysis across different modes of reasoning (\Cref{sec:new-taxonomy}). Throughout the remaining sections, we will work in a restricted domain that excludes the empty weighting function, meaning that running a program must always result in some outcome, whether it be a reachable state or nontermination.

\subsection{Weakest Precondition Logics: The Upper Cube} \label{sec:modality}

The upper cube in \Cref{fig:vk-taxonomy} is comprised of logics defined by different flavors of \emph{weakest precondition} transformers. Originally due to \citet{Dijkstra76}, these transformers initially had two variants: the weakest precondition of a program $C$ relative to postcondition $Q$ is the set of states that always terminate in $Q$ after $C$ is run, whereas the weakest \emph{liberal} precondition is the set of start states that end in $Q$ \emph{if} the program terminates. Both of these transformers treat nondeterminism demonically, so \citet{verscht2025taxonomy} refer to them as the demonic weakest precondition ($\dwp$) and demonic weakest liberal precondition ($\dwlp$), respectively.
\begin{align*}
    & \dwp(C,Q) \triangleq \{\sigma \mid \supp(\de{C}(\sigma)) \subseteq Q\} 
    & \dwlp(C,Q) \triangleq \{\sigma \mid \supp(\de{C}(\sigma)) \subseteq Q_{\nonterm}\} 
\end{align*}
It is well known that total and partial correctness Hoare Logic \cite{hoarelogic,manna1974axiomatic} can be framed in terms of these transformers. That is, the total correctness Hoare Triple $\hoare PCQ$ is equivalent to $P \subseteq \dwp(C, Q)$ and the partial correctness triple $\hoare PCQ$ is equivalent to $P\subseteq \dwlp(C, Q)$.
Of course, the above transformers can also be formulated angelically.
\begin{align*}
    & \awp(C,Q) \triangleq \{\sigma \mid \supp(\de{C}(\sigma)) \cap Q \neq \emptyset\} 
    & \awlp(C,Q) \triangleq \{\sigma \mid \supp(\de{C}(\sigma)) \cap Q_{\nonterm} \neq \emptyset\}
\end{align*}
The angelic weakest precondition ($\awp$), originally due to \citet{wpp} under the name \emph{weakest possible precondition}, is the set of states that reach $Q$ via \emph{some} trace in $C$. The resulting logic $P\subseteq \awp(C, Q)$---known as Lisbon Logic \cite{outcome}, Backwards Under-Approximation \cite{moller2021algebra}, and Sufficient Incorrectness Logic \cite{ascari2023sufficient}---is useful for proving incorrectness, \ie that certain bugs are reachable. The angelic weakest liberal precondition is less well studied. It is the set of states that either reach the postcondition or diverge, and was introduced by \citet{verscht2025taxonomy} in order to complete the front face of the upper cube in \Cref{fig:vk-taxonomy}.
It is related to the logic UNTer, a logic for proving possible nontermination \cite{raad2024nontermproving}. Indeed, $P \subseteq \awlp(C, \fls)$ means that for each state in $P$, there exists a diverging trace.
We encode these logics in \TOL using the following modalities; recall that here we assume that $\supp(m) \neq \emptyset$.
\begin{align*}
    \always P &\triangleq \exists (u, v) : U^2\setminus\{(\zero,\zero)\}.\: \ov{P}^{(u)} \oplus \Uparrow^{(v)}
    &&= \{m \mid \supp(m) \subseteq P_{\nonterm}\}
    \\
    \sometimes P &\triangleq \exists u : U\setminus\{\zero\}.\: \ov{P}^{(u)} \oplus \top
    &&= \{m \mid \supp(m) \cap P \neq \emptyset\}
    \\
    \alwaystot P &\triangleq \exists u : U\setminus\{\zero\}.\: \ov{P}^{(u)}
    && = \{m \mid \supp(m) \subseteq P\}
    \\
    \sometimespart P &\triangleq \exists (u, v) : U^2\setminus\{(\zero, \zero)\}.\: \ov{P}^{(u)} \oplus \Uparrow^{(v)} \oplus \top
    &&= \{m \mid \supp(m) \cap P_{\nonterm} \neq \emptyset \}
\end{align*}
The $\always$ and $\sometimes$ modalities are inspired by Dynamic Logic \cite{pratt1976semantical,harel2001dynamic} and correspond to $\dwlp$ and $\awp$, respectively. That is, $\always Q$ states that $Q$ covers all the terminating states in some weighting function $m$, and $\sometimes Q$ states that $m$ contains some state in $Q$. As was shown by \citet{zilberstein:completeOL}, $\always$ and $\sometimes$ can be used to encode partial correctness Hoare Logic and Lisbon Logic:
\begin{figure}[t]
\begin{tikzcd}[
  3d cd=7em, arrows=dash, row sep=.5em, column sep=-3em,
  dedo/.style={dash pattern=on \pgflinewidth off 2\pgflinewidth},
  cells={nodes={
      inner sep=+.2em, 
      align=center, text width=8em, text depth=+6pt, text height=1.3em,
      t/.style={text={####1}}}}]
\taxnode{\triple{\ov{\lnot P}}C{\sometimes{\lnot Q}}}{P\supseteq \dwlp(C, Q)}{} \ar[dd, leftarrow] \ar[r, rightarrow]
&
\taxnode{\triple{\ov{\lnot P}}C{\sometimespart{\lnot Q}}}{P\supseteq \dwp(C, Q)}{} \ar[dd, leftarrow] 
\\
\taxnode{\triple{\ov P}C{\sometimes Q}}{P \subseteq \awp(C, Q)}{Lisbon Logic} \ar[r, rightarrow, crossing over] \ar[u, dotted, leftrightarrow, mygreen]
&
\taxnode{\triple{\ov P}C{\sometimespart Q}}{P \subseteq \awlp(C, Q)}{Angelic Partial Correctness} \ar[u, dotted, leftrightarrow, mygreen]
\\
\taxnode{\triple{\ov{\lnot P}}C{\alwaystot{\lnot Q}}}{P \supseteq \awlp(C, Q)}{Total NC} \ar[r, rightarrow] 
&
\taxnode{\triple{\ov{\lnot P}}C{\always{\lnot Q}}}{P \supseteq \awp(C, Q) \Leftrightarrow \dslp(C, P) \supseteq Q}{Necessary Conditions} \ar[dd, rightarrow]
&
\taxnode{(1)}{\dsp(C, P) \supseteq Q}{Demonic Incorrectness} \ar[l, rightarrow] \ar[dd, rightarrow]
\\
\taxnode{\triple{\ov P}C{\alwaystot Q}}{P \subseteq \dwp(P, Q)}{Total Hoare Logic} \ar[uu, rightarrow] \ar[r, rightarrow] \ar[u, dotted, leftrightarrow, mygreen]
&
\taxnode{\triple{\ov P}C{\always Q}}{P\subseteq \dwlp(C, Q) \Leftrightarrow \asp(C, P) \subseteq Q}{Partial Hoare Logic}
  \ar[uu, rightarrow, crossing over] \ar[u, dotted, leftrightarrow, mygreen]
  \ar[dd, rightarrow]
&
\taxnode{(3)}{\aslp(C, P) \subseteq Q}{} \ar[l, rightarrow, crossing over] \ar[u, dotted, leftrightarrow, mygreen] 
\\ &
\taxnode{(2)}{\aslp(C, P) \supseteq Q}{Angelic Partial Incorrectness}
&
\taxnode{\triple{\un P}C{\un Q}}{\asp(C,P) \supseteq Q}{Incorrectness Logic} \ar[l, rightarrow]
\\ &
\taxnode{(4)}{\dsp(C,P)\subseteq Q}{}  \ar[u, dotted, leftrightarrow, mygreen]
&
\taxnode{\triple{\un{\lnot P}}C{\un{\lnot Q}}}{\dslp(C,P) \subseteq Q}{Necessary Incorrectness} \ar[l, rightarrow] \ar[uu, leftarrow, crossing over] \ar[u, dotted, leftrightarrow, mygreen]
\end{tikzcd}

%
%
\caption{Adaptation of the \citet{verscht2025taxonomy} Taxonomy.
Standard arrows represent implications and green dotted arrows represent contrapositives. The numbered logics are represented as two triples: (1) $\triple{\ov{\lnot P}}C{\always{\lnot Q}} \land \triple{\un P}C{\un Q}$, (2) $\triple{\ov{\lnot P}}C{\always{\lnot Q}} \lor \triple{\un P}C{\un Q}$, (3) $\triple{\ov P}C{\always Q} \land \triple{\un{\lnot P}}C{\un{\lnot Q}}$, and (4) $\triple{\ov P}C{\always Q} \lor \triple{\un{\lnot P}}C{\un{\lnot Q}}$.}
\label{fig:vk-taxonomy}
\end{figure}

\begin{theorem}[Subsumption of Hoare and Lisbon Logic] \label{thm:subsume-dwlp-awp}
\[
    P \subseteq \dwlp(C, Q) \iff \vDash \triple{\ov{P}}{C}{\always Q}
    \qquad\text{and}\qquad
    P \subseteq \awp(C, Q) \iff \vDash \triple{\ov{P}}{C}{\sometimes Q}
\]
\end{theorem}
The $\alwaystot$ and $\sometimespart$ modalities witness a particular (zero or nonzero) amount of termination mass, and were therefore not expressible in standard Outcome Logic \cite{outcome,zilberstein:completeOL,zilberstein2024outcome}. More precisely $\alwaystot Q$ states not only that $Q$ covers all the reachable states, but also that $\nonterm$ has weight zero. Finally, $\sometimespart Q$ states that either $Q$ or $\nonterm$ occur with nonzero weight. These new modalities can be used to express the final logics:
\begin{theorem}[Subsumption of Total Hoare Logic and Angelic Partial Correctness] \label{thm:subsume-dwp-awlp}
\[
    P \subseteq \dwp(C, Q) \iff \vDash \triple{\ov{P}}{C}{\alwaystot Q}
    \qquad\text{and}\qquad    
    P \subseteq \awlp(C, Q) \iff \triple{\ov{P}}{C}{\sometimespart Q}
\]
\end{theorem}
Just as the logics in \Cref{fig:vk-taxonomy} form a commuting square, so do the modalities:
\[
\begin{tikzcd}[
  3d cd=0em, row sep=.75em, column sep=3em]
\mynode{P \subseteq \awp(C, Q)}{Lisbon Logic} \ar[r]
&
\mynode{P \subseteq \awlp(C, Q)}{Angelic Partial Correctness}
\\
\mynode{P \subseteq \dwp(C, Q)}{Total Hoare Logic} \ar[r] \ar[u]
&
\mynode{P\subseteq \dwlp(C, Q)}{Partial Hoare Logic} \ar[u]
\end{tikzcd}
\qquad
\begin{tikzcd}[
  3d cd=0em, row sep=.75em, column sep=1em]
\sometimes Q \ar[r, Rightarrow] 
&
\sometimespart Q
\\
\alwaystot Q \ar[r, Rightarrow]\ar[u, Rightarrow]
&
\always Q \ar[u, Rightarrow]
\end{tikzcd}
\]
The back face of the upper cube is obtained by taking contrapositives of the four aforementioned logics, \ie by negating the pre- and postconditions. Because of the de Morgan style dualities between the modalities and predicate transformers (\eg $\always P$ iff $\lnot \sometimes \lnot P$ and $\alwaystot P$ iff $\lnot\sometimespart \lnot P$), these logics can also be expressed by flipping the $\subseteq$ to be a $\supseteq$.

These contrapositive logics have not been studied, with the exception of $P \supseteq \awp(C, Q)$, which was named \emph{Necessary Conditions} (NC) by \citet{cousot2013automatic}. Indeed, in an NC triple $\hoare PCQ$, $P$ is a necessary condition to reach $Q$, which means that all states outside of $P$ must not reach $Q$. It was observed by \citet{zhang:quantitative-sp} and \citet{ascari2023sufficient} that NC can be defined in terms of the angelic weakest precondition.

\subsection{Strongest Postcondition Logics: The Lower Cube}
\label{sec:vk-sp}

The lower cube in \Cref{fig:vk-taxonomy} is framed in terms of strongest postconditions \cite{dijkstra1990strongest}, which give the set of states reachable from the precondition. \citeauthor{verscht2025taxonomy} refer to the classical strongest postcondition as the \emph{angelic} strongest postcondition, as the states must only be reachable from \emph{some} start state, not \emph{all} start states.
There is a well known Galois connection between $\dwlp$ and $\asp$, meaning that partial Hoare Logic can be framed in terms of both:
\[
  P \subseteq \dwlp(C, Q)
  \iff
  \asp(C, P) \subseteq Q
  \qquad
    \asp(C,P) \triangleq \{\tau \mid \exists \sigma \in P.\ \tau \in \supp(\de{C}(\sigma)) \}
\]
Cousot and \citet{zhang:quantitative-sp} observed another Galois connection between $\awp$ and the demonic strongest \emph{liberal} postcondition ($\dslp$), meaning that NC also has two characterizations. Here, liberality corresponds to unreachability of the postcondition rather than nontermination from the precondition, so $\dslp(C, P)$ contains states that are unreachable from outside of $P$.
\[
  P \supseteq \awp(C, Q) \iff \dslp(C, P) \supseteq Q
  \qquad
    \dslp(C,P) \triangleq \{\tau \mid \forall \sigma\notin P.\ \tau \notin \supp(\de{C}(\sigma))\}
\]
For completeness, we include the remaining two transformers, although they do not have obvious intuitions or use cases, and as \citet[\S3.2.4]{verscht2025taxonomy} point out, they cannot be defined inductively. As such, we do not define them directly, but rather represent them equivalently as combinations of other transformers that we have seen.
\begin{align*}
    & \aslp(C,P) = \asp(C,P) \cup \dslp(C,P) 
    & \dsp(C,P) = \asp(C,P) \cap \dslp(C,P)
\end{align*}
We can now describe the logics in the lower cube of \Cref{fig:vk-taxonomy}. Due to the aforementioned Galois connections, the top left edge overlaps with the upper cube, and describes partial Hoare Logic and NC. The only other logic in the cube that has been studied is Incorrectness Logic (IL) (originally introduced under the name Reverse Hoare Logic \cite{reversehoare}), in the back right, which is obtained by swapping the $\subseteq$ for a $\supseteq$ in the $\asp$ definition of Hoare Logic, \ie an IL postcondition under-approximates the reachable states whereas a Hoare Logic postcondition over-approximates \cite{il}. As was shown by \citet[Proposition 6]{hyperhoare}, IL can be encoded using under-approximate assertions $\un P$, \ie $m \vDash \un P$ iff $\supp(m) \supseteq P$.

The remaining logics in the lower cube are not particularly well-understood. Indeed, \citet{verscht2025taxonomy} showed that these logics have no inductive calculi, and accordingly we encode them in \TOL using two triples.
For example, Demonic Incorrectness states both that $P$ is a necessary condition to reach $Q$ ($\triple{\ov{\lnot P}}C{\always\lnot Q}$), and also that every state in $Q$ is reachable from $P$ ($\triple{\un P}C{\un Q}$), \ie it is similar to standard incorrectness logic, but with the additional stipulation that $P$ covers \emph{all} the states that could reach the bug.
No concrete applications have been found for these logics; as of now, they exist only to neatly complete the cube.

\subsection{In-Between Logics}
\label{sec:vk-in-between}

Although \Cref{fig:vk-taxonomy} appears comprehensive, more logics can be formed by taking unions and intersections of the aforementioned predicate transformers. Focusing on the front face of the upper cube, \citet{verscht2025taxonomy} form new logics by taking the conjunction and disjunction of Lisbon and Partial Hoare Logic. These \emph{in-between} logics can be encoded in \TOL as follows:
\[
\begin{tikzcd}[
  3d cd=0em, row sep=0em, column sep=2em]
\mynode{\triple{\ov{P}}C{\sometimes{Q}}}{Lisbon Logic} \ar[rrr, rightarrow] \ar[drr, rightarrow]
&&&
\mynode{\triple{\ov{P}}C{\sometimespart{Q}}}{Angelic Partial Correctness}
\\
&&
\scalebox{.7}{$\triple{\ov{P}}C{\sometimes Q \vee \always \fls}$} \ar[ur, rightarrow]
\\
&
\scalebox{.7}{$\triple{\ov{P}}C{\ov{Q} \oplus \always\fls}$} \ar[ur, rightarrow] \ar[uul, rightarrow] \ar[drr, rightarrow]
\\
\mynode{\triple{\ov{P}}C{\alwaystot Q}}{Total Hoare Logic} \ar[rrr, rightarrow]\ar[uuu, rightarrow] \ar[ur, rightarrow]
&&&
\mynode{\triple{\ov{P}}C{\always Q}}{Partial Hoare Logic} \ar[uuu, rightarrow] \ar[uul, rightarrow]
\end{tikzcd}
\]
The first new logic offers triples of the form $\triple{\ov P}C{\ov{Q}\oplus \always\fls}$---$Q$ is reachable and is the only terminating outcome, but nontermination is possible too. This is the precise meaning of $\triple{\ov P}C{\ov Q}$ in standard Outcome Logic \cite{zilberstein:completeOL}, as its semantic model could not identify the existence of nonterminating traces.
This logic has a clear application in that it unifies partial correctness and incorrectness reasoning within a single logic \cite{outcome}.

The second new logic has the form $\triple{\ov P}C{\sometimes Q \vee \always\fls}$, meaning that either $Q$ is reachable  or $C$ never terminates. The applications for this second logic are less clear, being weaker than both Hoare Logic and Lisbon Logic. Still, expressing it in \TOL at least gives us compositional rules to derive specifications, whereas \citeauthor{verscht2025taxonomy} must derive it using two separate calculi.
More broadly, the existence of these in-between logics show that \Cref{fig:vk-taxonomy} does not tell the complete story; \emph{how many other logics are in-between?} In \Cref{sec:new-taxonomy}, we explore this question further and see how \TOL is more expressive than any one of the predicate transformers that we have seen thus far.

\section{The Cousot Cube}
\label{sec:cousot}

A second effort to build taxonomies of program logics was made by \citet{cousot2024calculational}. In this framework, 12 logics are arranged along a cube, as shown in \Cref{fig:cousot-cube}. The goals of constructing this cube were orthogonal to the taxonomy of \citet{verscht2025taxonomy}. Whereas the latter demonstrates weakening and contrapositive relationships, Cousot's cube is intended to demonstrate how different logics can be constructed by composing abstractions, and how proof systems can be calculationally derived using fixpoint induction. Nevertheless, it is interesting to compare the two taxonomies and see how Cousot's cube can be expressed in Total Outcome Logic.

\newcommand{\cornernode}[2]{\scalebox{.5}{$#1 \stackrel{?}\leftrightarrow #2$}}

\begin{figure}
\begin{tikzcd}[
  3d cd=0em, row sep=.6em, column sep=-1.5em,
  arrows={dash,thick}]
&&
\mynode{\triple{\ov{\lnot P}}C{\sometimes{\lnot Q}}}{} \ar[dr]
&& &&
\mynode{\triple{\un{\lnot P}}C{\un{\lnot Q}}}{Necessary Incorrectness}
\\
&&&
\cornernode{P}{\widetilde{\mathsf{pre}}(C, Q)}
  \ar[dr, "\subseteq"] \ar[ul, "\supseteq"{above}]
   \ar[rrrr] \ar[ddll] \ar[ddddd]
&&&&
\cornernode{\widetilde{\mathsf{post}}(C, P)}{Q}
  \ar[dr, "\supseteq"] \ar[ul, "\subseteq"{above}]
   \ar[ddll] \ar[ddddd]
\\
\mynode{\triple{\ov{\lnot P}}C{\always{\lnot Q}}}{Necessary Conditions}
    \arrow[rrrrrrrr, dashed, blue, bend right=10, thin]
&&&&
\mynode{\triple{\ov P}C{\always Q}}{Partial Hoare Logic}
&&&&
\mynode{\triple{\ov{\lnot P}}C{\always{\lnot Q}}}{Necessary Conditions}
\\
&
\cornernode{P}{{\mathsf{pre}}(C, Q)}
  \ar[dr, "\subseteq"] \ar[ul, "\supseteq"{above}] \ar[rrrr, crossing over] \ar[ddddd]
&&&&
\cornernode{{\mathsf{post}}(C, P)}{Q}
  \ar[dr, "\supseteq"] \ar[ul, "\subseteq"{above}]
\\
&&
\mynode{\triple{\ov P}C{\sometimes Q}}{Lisbon Logic}
&&&&
\mynode{\triple{\un P}C{\un Q}}{Incorrectness Logic}
\\
&&
\mynode{\triple{\ov{\lnot P}}C{\sometimespart{\lnot Q}}}{}
&& &&
\mynode{\triple{\un{\lnot P}}C{\un{\lnot Q} \oplus \Uparrow}}{}
\\
&&&
\cornernode{P}{\widetilde{\mathsf{pre}}_\bot(C, Q)} \ar[rrrr] \ar[ddll]
  \ar[dr, "\subseteq"] \ar[ul, "\supseteq"{above}] 
&&&&
\cornernode{\widetilde{\mathsf{post}}_\bot(C, P)}{Q} \ar[ddll]
  \ar[dr, "\supseteq"] \ar[ul, "\subseteq"{above}] 
\\
\mynode{\triple{\ov{\lnot P}}C{\alwaystot{\lnot Q}}}{}
    \arrow[rrrrrrrr, dashed, blue, bend right=10, thin]
&&&&
\mynode{\triple{\ov P}C{\alwaystot Q}}{Total Hoare Logic}
&&&&
\mynode{\triple{\ov{\lnot P}}C{\alwaystot{\lnot Q}}}{Total NC}
\\
&
\cornernode{P}{{\mathsf{pre}}_\bot(C, Q)}
  \ar[dr, "\subseteq"] \ar[ul, "\supseteq"{above}] \ar[rrrr]
&&&&
\cornernode{{\mathsf{post}}_\bot(C, P)}{Q}
  \ar[dr, "\supseteq"] \ar[ul, "\subseteq"{above}]  \ar[uuuuu, crossing over]
\\
&&
\mynode{\triple{\ov P}C{\sometimespart\fls}}{UNTer (BUA)}
&&&&
\mynode{\triple{\un P}C{\sometimespart\fls}}{UNTer (FUA)}
\end{tikzcd}
\caption{The Cousot Cube, with logics expressed as Total Outcome Logic triples. The blue dashed lines represent Galois Connections.}
\label{fig:cousot-cube}
\end{figure}

The corners of the cube in \Cref{fig:cousot-cube} contain predicate transformers, similar to those that we saw in \Cref{sec:vk-taxonomy}. We write, \eg $P \stackrel{?}\leftrightarrow \wpre(C, Q)$, to indicate that $P$ is either a subset or superset of $\wpre(C, Q)$. Two nodes protrude from each corner, denoting the $\subseteq$ and $\supseteq$ case. This makes for 16 logics, although 4 pairs of them collapse due to Galois connections.
The top face of the cube is comprised of partial correctness logics, where termination and nontermination cannot be precisely quantified. Cousot's predicate transformers are a subset of those of \citet{verscht2025taxonomy}, and accordingly, all of the logics on the top face also appear in \Cref{fig:vk-taxonomy}.
\begin{align*}
\spost(C, P) &= \{ \tau \mid \exists \sigma \in P.\ \tau \in \supp(\de{C}(\sigma)) \}  &&= \asp(C, P)
\\
\widetilde{\spost}(C, P) &= \{ \tau \mid \forall \sigma \notin P.\ \tau \notin \supp(\de{C}(\sigma)) \} &&= \dslp(C, P)
\\
\wpre(C, Q) &= \{ \sigma \mid \exists \tau \in \supp(\de{C}(\sigma)).\ \tau \in Q \}
  &&= \awp(C, Q)
\\
\widetilde{\wpre}(C, Q) &= \{ \sigma \mid \forall \tau \in \supp(\de{C}(\sigma)).\ \tau \in Q_\nonterm \}
  &&= \dwlp(C, Q)
\end{align*}
The bottom face of the cube is where the two taxonomies diverge. The $\bot$ subscripted variants of the transformers are defined similarly to above, except that postconditions can range over $\Sigma_\nonterm$, meaning that the postcondition can express properties about (non)termination. As such, the Total Outcome Logic triples on the bottom face are only examples of what can be expressed in these logics, and each logic on the bottom face is strictly more expressive than the corresponding logic on top. For example, Total Hoare Logic appears in the middle of the square, but we could also express partial correctness there as $P \subseteq \widetilde{\wpre}_\bot(C, Q_\nonterm)$ instead of $P \subseteq \widetilde{\wpre}_\bot(C, Q)$.

Similarly, at the front of the bottom face we have the \emph{backward} and \emph{forward} under-approximate (BUA and FUA) variants of the nontermination logic UNTer \cite{raad2024nontermproving}, stating that every state in $P$ has a nonterminating trace and that there exists a nonterminating trace starting in $P$, respectively. However, those nodes also subsume Lisbon Logic and Incorrectness Logic, respectively. By contrast, the FUA variant of UNTer does not appear in \Cref{fig:vk-taxonomy}, since there is no strongest postcondition transformer that includes nontermination information.

Cousot includes one more logic that does not fit into the hierarchy---\emph{Hoare Incorrectness Logic}---which is the negation of a Hoare Triple $\lnot\triple{\ov P}C{\always Q}$. In \TOL, we can express Hoare Incorrectness Logic as
$\triple{\un P}C{\sometimes \lnot Q}$, which is also equivalent to $\exists \sigma\in P.\ \exists \tau \in \supp(\de{C}(\sigma)).\ \tau \notin Q$. This logic identifies a single trace---\ie a counterexample---invalidating the correctness specification.

\section{A New Taxonomy for Correctness and Incorrectness}
\label{sec:new-taxonomy}

We have shown that \TOL has the expressive power to subsume all the logics in both the \citet{verscht2025taxonomy} taxonomy of Hoare-like logics and the \citet{cousot2024calculational} cube. However, although the prior two sections have explored more than 16 different program logics, only a few of them have found practical use. The front face of the upper cube in \Cref{fig:vk-taxonomy}---containing (partial and total) Hoare Logic and Lisbon Logic---is well understood, and serves as the foundation of decades of program analysis research. Incorrectness Logic has served as a semantic justification for real-world static analysis systems \cite{realbugs,isl}, although it was more recently shown that Lisbon Logic also models those systems \cite{zilberstein2024outcome,raad2024nontermproving} and has some advantages in its ability to identify the cause of a bug \cite{outcome,ascari2023sufficient}.

The remaining logics are not readily applicable to real static analysis problems. Their underlying concepts---demonic strongest postconditions (being reachable from \emph{all} start states) and liberal postconditions (being unreachable from any start state)---neatly form Galois connections and complete the cubes, but defy intuitions and do not map to any of the properties about programs that have arisen in the past several decades of static analysis research. Even worse, some of these properties cannot be described using inductive rules \cite{verscht2025taxonomy}.

By contrast, Total Outcome Logic offers more expressiveness to reason about properties in that upper square.
This expressivity stems from lifting reasoning to the level of sets of outcomes, where we can encode reachability and termination requirements as outcome assertions rather than building them into the interpretation of our triples. This lends \TOL far greater versatility than simply instantiating the 16+ different logics, as we are no longer confined to the corners of the cube, but can navigate freely in the volume within.

Importantly, this enables shared program analyses in \TOL for different properties. 
While \Cref{fig:vk-taxonomy} suggests that Total Hoare Logic is the \emph{strongest} logic in that upper front square, this gives only a limited, and perhaps misleading, picture as to the compatibility of the logics in that square.

Total Hoare triples give over-approximate correctness specifications, while Lisbon triples target the existence of a more select set of outcomes---potential bugs, for instance. More precisely, suppose that $Q_\ok$ represents some collection of good states and $Q_\er$ represents bugs states. It may be the case that the bug described by $Q_\er$ \emph{sometimes} occurs and so the Lisbon triple $\triple{\ov P}C{\sometimes Q_\er}$ is valid. But since the bug does not \emph{always} occur, total Hoare Logic can only use a weaker postcondition $\triple{\ov P}C{\alwaystot(Q_\ok \vee Q_\er)}$. These two specification are incompatible; there is no weakening relation between them. Similar issues arise between the other logics when postconditions differ. Some logics are wholly incompatible---\ie one type of triple cannot be used in a subderivation of the other whatsoever, such as between partial correctness $\triple{\ov{P}}{C}{\always Q}$ and Lisbon Logic $\triple{\ov{P}}{C}{\sometimes Q}$---rendering them isolated and disconnected for the purpose of analysis. 

In \TOL, we can instead perform the bulk of our program analysis on finer-grained intermediate specifications, where we track \emph{reachable} outcomes of interest with $\oplus$ in the form of assertions: $\ov{Q_1} \oplus \cdots \oplus \ov{Q_n}$. Then, letting $Q = Q_1 \vee \cdots \vee Q_n$ and $1 \le k \le n$, we can weaken these to our desired property with an application of \sruleref{Consequence}, as shown below.
\begin{equation}\label{eq:taxonomy}
\begin{tikzcd}[
    3d cd=0em, row sep=.75em, column sep=3em]
  &
  \mynode{\triple{\ov{P}}C{\sometimes{Q_k}}}{Lisbon Logic} \ar[d, rightarrow] \ar[r, rightarrow]
  &
  \mynode{\triple{\ov{P}}C{\sometimespart{Q_k}}}{Angelic Partial Correctness} \ar[d, rightarrow]
  \\
  \triple{\ov{P}}C{\ov{Q_1} \oplus \cdots \oplus \ov{Q_n}} \ar[ur, rightarrow,bend left=15] \ar[dr, rightarrow,bend right=15]
  &
  \mynode{\triple{\ov{P}}C{\sometimes{Q}}}{Lisbon Logic} \ar[r, rightarrow]
  &
  \mynode{\triple{\ov{P}}C{\sometimespart{Q}}}{Angelic Partial Correctness}
  \\
  &
  \mynode{\triple{\ov{P}}C{\alwaystot Q}}{Total Hoare Logic} \ar[r, rightarrow]\ar[u, rightarrow]
  &
  \mynode{\triple{\ov{P}}C{\always Q}}{Partial Hoare Logic} \ar[u, rightarrow]
\end{tikzcd}
\end{equation}
The commuting square from the front of the upper cube in \Cref{fig:vk-taxonomy} still appears, but above it there is also a Lisbon triple with a stronger postcondition $Q_k \Rightarrow Q$, which is incomparable to the total Hoare triple below it. Yet, both specifications stem from a common \TOL triple, shown on the left of the figure. This is significant because helper functions in a codebase can be given specifications in \TOL and then be repurposed for total correctness, partial correctness, and/or incorrectness.

With this in mind, we present our full taxonomy in \Cref{fig:new-taxonomy}, which is a superset of the one above. The full taxonomy includes more of the recently introduced logics for (in)correctness, and (non)termination, including Incorrectness Logic \cite{il}. All of these logics are arranged on the front face of the diagram, and common \TOL specifications are placed behind. 

\newcommand{\smoplss}{\oplus\cdotp\!\cdotp\!\cdotp\!\oplus}
\begin{figure}
\begin{tikzcd}[
  3d cd=-4.5em, row sep=.5em, column sep=-2em]
\mynode{\triple{\un P}C{\sometimes Q_k}}{Hoare Incorrectness Logic}
\\
\phantom{XXX}
&
\scalebox{.75}{$\triple{\ex{P}}C{\ex{Q}}$} \ar[d]
&
\scalebox{.6}{$\triple{\ov{P}}C{\ov{Q_1} \smoplss\ov{Q_n}}$} \ar[d] \ar[rrd]
&
\phantom{XXX}
\\
\mynode{\triple{\un{P}}C{\un{Q_k}}}{Incorrectness Logic} \ar[uu]
&
\mynode{\triple{\ex{P}}C{\ex{Q}}}{Total Exact Logic} \ar[r]
&
\mynode{\triple{\ov{P}}C{\alwaystot{Q}}}{Total Hoare Logic}  \ar[r]
&
\mynode{\triple{\ov{P}}C{\sometimes{Q}}}{Lisbon Logic}
&
\mynode{\triple{\ov{P}}C{\sometimes{Q_k}}}{Lisbon Logic} \ar[dd] \ar[l] \ar[lllluu, bend right=12]
\\
&
\scalebox{.75}{$\triple{\ex{P}}C{\ex{Q}\oplus\mathord\Uparrow}$} \ar[d] \ar[dddl, bend right=29]
&
\scalebox{.6}{$\triple{\ov{P}}C{\ov{Q_1} \smoplss\ov{Q_n}\oplus\mathord\Uparrow}$} \ar[d] \ar[rrddd, bend left=27] \ar[urr, bend right=10]
\\
\mynode{\triple{\un{P}}C{\un{Q}}}{Incorrectness Logic} \ar[uu]
&
\mynode{\triple{\ex{P}}C{\ex{Q} \oplus \always\fls}}{(Partial) Exact Logic} \ar[l, crossing over] \ar[r] \ar[uu, crossing over, leftarrow]
&
\mynode{\triple{\ov{P}}C{\always Q}}{Partial Hoare Logic} \ar[r] \ar[uu, crossing over, leftarrow]
&
\mynode{\triple{\ov{P}}C{\sometimespart{Q}}}{Angelic Part. Cor.} \ar[uu, crossing over, leftarrow]
&
\mynode{\triple{\ov{P}}C{\sometimespart{Q_k}}}{Angelic Part. Cor.} \ar[l, crossing over]
\\
&
\scalebox{.75}{$\triple{\ex{P}}C{\Uparrow}$} \ar[d]
&
\scalebox{.75}{$\triple{\ov{P}}C{\Uparrow}$} \ar[d] \ar[l] 
\\
\mynode{\triple{\un{P}}C{\sometimespart\fls}}{UNTer (FUA)}
&
\mynode{\triple{\ex{P}}C{\always\fls}}{Exact Logic} \ar[l] \ar[r]
&
\mynode{\triple{\ov{P}}C{\always\fls}}{Partial Hoare Logic} \ar[rr] \ar[uu, crossing over]
&&
\mynode{\triple{\ov{P}}C{\sometimespart{\fls}}}{UNTer (BUA)} \ar[uu] \ar[llll, bend left=8]
\end{tikzcd}
\caption{A new taxonomy of correctness and incorrectness logics, where $Q = Q_1 \vee \cdots \vee Q_n$ and $1 \le k \le n$.}
\label{fig:new-taxonomy}
\end{figure}

Recall that neither \Cref{fig:vk-taxonomy} nor \Cref{fig:cousot-cube} provide ways to compare Hoare Logic and Incorrectness Logic. The key to doing so is Exact Logic, which expresses both kinds of triples together \cite{maksimovi_c2023exact}. Exact Logic is encoded in \TOL using exact predicates $\ex P$, stating that the reachable states are exactly $P$. Whereas, the original Exact Logic was based on partial correctness, we also include a new Total Exact Logic, which additionally guarantees termination.

But despite our ability to encode Incorrectness and Exact Logics in \TOL, there are advantages to forgoing them and instead using variations of Lisbon Logic for bug-finding \cite{outcome,ascari2023sufficient,raad2024nontermproving}. Whereas Incorrectness Logic witnesses \emph{some} erroneous trace starting from the precondition, Lisbon Logic ensures that the bug can be witnessed from \emph{all} start states. This is an important property in program analysis, making it easier to find \emph{manifest errors}---bugs that can occur regardless of context \cite{realbugs}. Exact Logic was developed for symbolic execution \cite{loow2024compositional}, but is not an ideal foundation for other kinds static analysis as it has no weakening rule, and therefore cannot perform reasoning in abstract domains to make use of, \eg loop invariants \cite{zilberstein2024outcome,ascari2022limits}. Case in point: there is no weakening of the exact triple $\triple{\ex P}C{\ex Q}$ to use the more precise postcondition $\ex{Q_k}$.

By contrast, \TOL offers more versatility to navigate between correctness and incorrectness specifications.
With this in mind, we turn our attention to the portion of the taxonomy containing Partial and Total Hoare Logic, along with everything to the right. This is an expanded version of (\ref{eq:taxonomy}) above, where the unifying power of \TOL shines. In the back there are three forms of \TOL triples. Programs in the first row always terminate, programs in the second row sometimes terminate, and programs in the third row never terminate. A single \TOL triple can subsume partial correctness, total correctness, Lisbon Logic, and nontermination, all of which have been successfully deployed in industrial static analysis tools.
\TOL unifies all these kinds of correctness and incorrectness, making it an ideal logical foundation. In next sections, we will develop reasoning principles for \TOL and illustrate their versatility through a few simple case studies.

\section{Derived Rules}
\label{sec:derived}

Throughout \Cref{sec:vk-taxonomy,sec:cousot,sec:new-taxonomy}, we have seen how \TOL is able to express the specifications of many types of logics. However, it is also important that \TOL provides convenient reasoning principles for those types of specifications. In this section, we derive new rules for common patterns arising from the logics that we have previously seen.

\subsubsection*{If Statements, While Loops, and Divergence.} Recall from \Cref{sec:syntax} that we gave syntactic sugar to define if statements and while loops in terms of the branching, iteration, and $\assume b$ constructs of our language. Whereas we previously only gave rules for standard branching and iteration, we now give rules for if and while. The if rule is the same as the one from standard Outcome Logic, and requires the precondition to be split into two parts to satisfy the guard and its negation. The two branches can then be analyzed compositionally.
\begin{mathpar} \label{rule:if}
    \inferrule*[Right=If]
        {\varphi_1 \vDash b \\ \triple{\varphi_1}{C_1}{\psi_1} \\
        \varphi_2 \vDash \neg b \\ \triple{\varphi_2}{C_2}{\psi_2}
        }
        {\triple{\varphi_1 \oplus \varphi_2}{\iftf{b}{C_1}{C_2}}{\psi_1 \oplus \psi_2}}
\end{mathpar}
The rule for while loops differs from standard Outcome Logic, as it includes nontermination information. We use three families of assertions: $\varphi_n$ represents the outcomes where $b$ remains true, $\psi_n$ represents the outcomes where $b$ is false, and $\zeta_n$ represents nontermination in nested loops.
Compared to the \textsc{Iter} rule that we saw previously, the \textsc{While} rule only has a single premise, as the entailments for $\varphi_n$ and $\psi_n$ remove the need to derive the behaviors of the \textsc{assume} commands.
\begin{mathpar} \label{rule:while}
    \inferrule*[Right=\textsc{While}]
        {(\psi_n)_{n \in \N} \rightsquigarrow \psi_\infty \\
        (\varphi_n \oplus \psi_n \oplus \zeta_n)_{n \in \N} \Uparrow \zeta_\infty \\
        \forall n \in \N. \\
            \varphi_n \vDash b \\ \psi_n \vDash \neg b \\ \zeta_n \vDash \div \\ 
            \triple{\varphi_n \oplus \zeta_n}{C}{\varphi_{n+1} \oplus \psi_{n+1} \oplus \zeta_{n+1}}
        }
        {\triple{\varphi_0 \oplus \psi_0 \oplus \zeta_0}{\whl{b}{C}}{\psi_\infty \oplus \zeta_\infty}}
\end{mathpar}
Standard Outcome Logic can only reason about terminating outcomes, and is thus only able to detect nontermination when no terminating executions are witnessed, \ie via triples of the form $\tripleol{\varphi}{C}{\0 \odot \top}$. Whereas this restricts proofs to \textit{possible termination} and \textit{guaranteed nontermination}, \TOL now allows us to demonstrate \textit{sure-termination} or \textit{possible nontermination}. 

In terms of triples, the former is often bundled with a safety property as a total correctness specification. We can verify total correctness using \textit{variants}, which track the number of iterations a loop makes until reaching a zero-variant -- a point of sure-termination where the loop guard fails. In addition, a program cannot interfere with the existing weight on nontermination $\nonterm$. This means that any outcome assertion that only describes nontermination should be preserved across triples. It is useful to expand the \sruleref{Div} rule to reflect this.
\begin{mathpar} \label{rule:variant}
    \inferrule*[Right=Variant]
        {\forall n < N. \\
            \varphi_{n+1} \vDash b \\ \varphi_0 \vDash \neg b \\
            \triple{\varphi_{n+1}}{C}{\varphi_n}    
        }
        {\triple{\exists N : \N.\: \varphi_N}{\whl{b}{C}}{\varphi_0}}
 
     \inferrule*[Right=Div*]
        {\zeta \vDash \div}
        {\triple{\zeta}{C}{\zeta}}
\end{mathpar}

\subsubsection*{Hoare Logic and Lisbon Logic.} We will now derive the rules of Lisbon Logic and partial and total correctness Hoare Logic. We will present most of the rules as total correctness triples (of the form $\triple{\ov P}C{\alwaystot Q}$), since $\alwaystot Q \Rightarrow \always Q$, and therefore a total correctness triple can be weakened to be a partial correctness one. We first give rules for sequencing total specifications, since the modalities used in pre- and postconditions do not exactly match.
\begin{mathpar}
    \inferrule*[right={Seq-Total-Hoare}]
        {\triple{\ov P}{C_1}{\alwaystot Q} \quad
        \triple{\ov Q}{C_2}{\alwaystot R}
        }
        {\triple{\ov P}{C_1 \seq C_2}{\alwaystot R}} 
    \qquad
    \label{rule:seq}
    \inferrule*[right={Seq-Lisbon}]
        {\triple{\ov P}{C_1}{\sometimes Q} \quad
        \triple{\ov Q}{C_2}{\sometimes R}
        }
        {\triple{\ov P}{C_1 \seq C_2}{\sometimes R}}
\end{mathpar}
Similarly, we give a rule for if statements in Hoare and Lisbon Logic, which do not require entailments since the premises of the rules directly specify whether the guard is true or not.
\begin{mathpar} 
    \label{rule:if-hoare}
    \inferrule*[right={If-Hoare}]{
      \triple{\ov{P \land b}}{C_1}{\alwaystot Q}
      \\
      \triple{\ov{P \land \neg b}}{C_2}{\alwaystot Q}
    }{
      \triple{\ov P}{\iftf{b}{C_1}{C_2}}{\alwaystot Q}
    }

    \label{rule:if-lisbon}
    \inferrule*[right={If-Lisbon}]
        {\triple{\ov{P \land b}}{C_1}{\sometimes Q} \\ \triple{\ov{P \land \neg b}}{C_2}{\sometimes Q}}
        {\triple{\ov P}{\iftf{b}{C_1}{C_2}}{\sometimes Q}}
\end{mathpar}
The inference rules for different variants of these logics differs the most when it comes to loops. It is well known that the invariant rule is complete for reasoning about loops in partial correctness Hoare Logic \cite{cook1978soundness}. We derive the invariant rule in our embedded logic.
\begin{mathpar} \label{rule:invariant}
    \inferrule*[Right=\textsc{Invariant}]
        {\triple{\ov{P \land b}}{C}{\always P}}
        {\triple{\ov P}{\whl{b}{C}}{\always(P \land \neg b)}}
\end{mathpar}
We can also derive the standard variant rule employed in total Hoare Logic. The rule uses an integer-valued \emph{ranking} expression $R$, which represents how close the program is to termination. In the premise, we are given that each iteration of the loop strictly decreases $R$, and the loop guard fails once $R$ reaches 0. This allows us to conclude that the loop surely terminates:
\begin{mathpar} \label{rule:hoare-variant}
\hspace{-1.7cm}   \inferrule*[Right=Hoare-Variant]
        {P \land b \implies R > 0 \\
        \forall n \in \N.\:\: \triple{\ov{P \land b \land R = n}}{C}{\alwaystot (P \land R < n)}}
        {\triple{\ov{P \land R \leq N}}{\whl{b}{C}}{\alwaystot (P \land \neg b)}}
\end{mathpar}

\subsubsection*{Nontermination.}
Bug-finding is often a more practical goal than correctness verification \cite{il}. We can thus shift our sight from proving total correctness to proving the presence of nontermination by establishing a \textit{quasi-invariant} \cite{larraz2014maxsmt}. This is a property preserved by the loop body which ensures that the loop guard holds; the quasi-invariant forces indefinite reiteration, and any computation that satisfies the quasi-invariant once is trapped within the loop. Taken semantically, this aligns closely with the notion of a \textit{recurrent set} \cite{gupta2008nonterm} of program states, visited infinitely often by the loop. 
In a nondeterministic setting, we can define this either angelically or demonically to show that the program \textit{sometimes} or \textit{always} diverges, respectively. Using the modalities from \Cref{sec:vk-taxonomy}, we get:
\[    \sometimespart\fls = \exists u : U\setminus\{\0\}.\: \diverge{u} \oplus \top 
        \qquad\qquad
     \always\fls = \exists u : U.\: \diverge{u}
\]
Then, the following two rules can be derived to prove nontermination.
\begin{mathpar} 
  \hspace*{-2.1cm}   \label{rule:quasi-invariant}
      \inferrule*[Right=QInv-Angel]
        {\ov P \vDash b \\ \triple{\ov P}{C}{\sometimes P}}
        {\triple{\ov P}{\whl{b}{C}}{\sometimespart \fls}}
        
  \hspace*{1cm}
    \inferrule*[Right=QInv-Demon]
        {\ov P \vDash b \\ \triple{\ov P}{C}{\always P}}
        {\triple{\ov P}{\whl{b}{C}}{\always \fls}}
\end{mathpar}
\section{Case Studies}
\label{sec:examples}

We want to emphasize that the power afforded by Total Outcome Logic lies in its subsumption of different types of reasoning, especially pertaining to termination and nontermination. To demonstrate the breadth of program properties that can be proven in \TOL, we deploy it on a few examples. In these examples, we use variable assignments $x \coloneqq E$ as basic actions. These actions are defined in the usual way, as in \cite{zilberstein:completeOL}.

\subsection{Total Correctness: Quicksort Partition}

In the classical formulation of Hoare logic for total correctness, \cite{manna1974axiomatic} derive the correctness of the Quicksort partition algorithm, where $A$ is an integer array of size $n+1$ and $p$ is the pre-computed pivot value:
\begin{align*}
    \textsc{Partition} \triangleq \begin{cases}
        & i := 0 \seq j := n \seq \\
        & \code{while $i \leq j$ do } \\
        & \qquad \code{if $A[i] \leq p$ then $i := i + 1 \seq$} \\
        & \qquad \code{else if $A[j] \geq p$ then $j := j - 1 \seq$} \\
        & \qquad \code{else } 
        \code{swap}(A, i, j) \seq 
        i := i + 1 \seq 
        j := j - 1 
    \end{cases}
\end{align*}
Total correctness of \textsc{Partition} can similarly be established in \TOL. Again, we can encode a total correctness triple $\total{P}{C}{Q}$ as outcome triple $\triple{\ov P}{C}{\alwaystot Q}$. As the same inference rules available in total Hoare logic can be derived in \TOL, our desired proof follows essentially the same steps. First, for a given pivot value $p$, define predicates $\alpha$ and $\beta$ over integers:
\[
    \alpha(i) \triangleq \bigwedge_{0 \leq k < i} A[k] \leq p
    \qquad 
    \beta(j) \triangleq \bigwedge_{j < k \leq n} A[k] \geq p
\]
We want to derive $\triple{\ov\tru}{\textsc{Partition}}{\alwaystot (\alpha(i) \land \beta(j) \land i > j)}$, which states that any run of the program must terminate having divided the array into two sections: elements $\leq p$ and those $\geq p$. We decorate \textsc{Partition} in \Aref{fig:partition} in \Aref{app:partition} to sketch the full proof.

The last step of the proof uses \textsc{Hoare-Variant}. Our \textit{invariant} $\alpha(i) \land \beta(j)$ enforces two boundaries in the array; every element behind index $i$ must be $\leq p$, while every element after $j$ must be $\geq p$. A natural choice for the \textit{variant} is then the distance between the boundaries $j - i$, which we show must approach one another after each iteration. That is, we need to derive the following across the loop body $C$: 
\[
    \triple{\ov{\alpha(i) \land \beta(i) \land i \leq j \land j - i = m}}{C}{\alwaystot(\alpha(i) \land \beta(i) \land j - i < m)}
\]
Clearly, the loop condition $i \leq j$ implies $j - i > 0$. This provides us with all the premises needed to conclude $\triple{\ov{\alpha(i) \land \beta(j)}}{\textsc{Partition}}{\alwaystot(\alpha(i) \land \beta(j) \land i > j)}$.

\subsection{Probabilistic Nontermination}

Whereas much of the work on probabilistic verification focuses on \emph{almost sure termination} \cite{kaminski,mciver2005abstraction,mciver2018new}---programs that terminate with probability 1---there are many programs simulating physical phenomena, which do not almost surely terminate, but are still important to reason about \cite{icard2017almost-sure}. An example of such programs is a scenario where a tortoise marches forward at a constant speed, and an erratic hare attempts to catch up.

The program below models such a scenario, where $t$ represents the position of the tortoise, $h$ represents the position of the hare, and $k$ is the number of steps that have been taken. The tortoise starts one step ahead of the hare. Each step, the hare either takes a step forward, or leaps forward. However, the hare gets tired over time, so it can only leap a distance of $1+\frac1{2^k}$. The program terminates if the hare catches up to the tortoise. Note that $C_1 +_p C_2$ is syntactic sugar for $(\assume p\fatsemi C_1) + (\assume{1-p}\fatsemi C_2)$.
\[
\begin{array}{ll}
\qquad\qquad\;\;  t \coloneqq 1 \fatsemi h \coloneqq 0 \fatsemi k \coloneqq 0 \fatsemi\phantom{x} \\
\qquad\qquad\;\;  \whl{h < t}{} \\
  C_\mathsf{body} \triangleq \left\{ \begin{array}{l}
  \quad t \coloneqq t+1 \fatsemi \phantom{x} \\
  \quad (h \coloneqq h + 1) +_\frac12 (h \coloneqq h+ 1 +\frac{1}{2^k}) \fatsemi \phantom{x} \\
  \quad k \coloneqq k+1
  \end{array}\right.
\end{array}
\]
The hare catches the tortoise with probability $\frac12$, when its initial move is a leap forward. If the hare's first move is not a leap, then it catches the tortoise with probability 0. We can see this informally, since in the best case, the hare catches up each round by $\frac12+\frac14+\frac18+\cdots$, which converges to 1, but that occurs only if the hare \emph{always} leaps forward, an event with probability $\frac12\cdot\frac12\cdot\frac12\cdot\cdots=0$.
We will make this formal using the \textsc{rule:while}{While} rule. For this, we need to define the three families of assertions $\varphi_n$, $\psi_n$ and $\zeta_n$. We start with $\varphi_n$, which describes the outcomes where execution continues.
If $n=0$, then $\varphi_0$ simply describes the start state of the program. If $n=1$, then $\varphi_n$ describes the outcome where the hare's first move was a regular step forward. For $n\ge 2$, we only describe two outcomes. With probability $\frac1{2^n}$, the hare leaps forward on every step after the first one, and in that case, the difference in position is exactly $\frac1{2^n}$. If not, then the hare did not leap on at least one step, in which case the difference in position is strictly greater than $\frac1{2^n}$.
\begin{align*}
  \varphi_0 &\triangleq \ov{t = 1 \land h = 0 \land k = 0}
  \qquad\qquad
  \varphi_1 \triangleq \wg{\ov{t = 2 \land h = 1 \land k = 1}}{\frac12}
  \\
  \varphi_n &\triangleq \wg{\ov{t - h = \frac1{2^{n-1}} \land k = n}}{\frac1{2^n}} \oplus \wg{\ov{t - h > \frac1{2^{n-1}} \land k = n}}{\frac{2^{n-1}-1}{2^n}}
  \quad\text{if}~ n\ge 2
\end{align*}
\begin{figure}
\[\small
  \begin{array}{lc|cl}
    \ob{\varphi_n} \Leftrightarrow
    \ob{\wg{\ov{t - h = \frac1{2^{n-1}} \land k = n}}{\frac1{2^n}} \oplus \wg{\ov{t - h > \frac1{2^{n-1}} \land k = n}}{\frac{2^{n-1}-1}{2^n}}}
    \\
    \;\; t\coloneqq t+1 \fatsemi\phantom{x}
    \\
    \lrob{\wg{\ov{t - h = 1+\frac1{2^{n-1}} \land k = n}}{\frac1{2^n}} \oplus \wg{\ov{t - h > 1+\frac1{2^{n-1}} \land k = n}}{\frac{2^{n-1}-1}{2^n}}}
    \\
    \;\; (h \coloneqq h + 1) +_\frac12 (h \coloneqq h+ 1 +\frac{1}{2^k}) \fatsemi \phantom{x}
    \\
    \lrob{\begin{array}{l}
        \wg{\ov{t - h = \frac1{2^{n-1}} \land k = n}}{\frac1{2^{n+1}}} \oplus
        \wg{\ov{t - h = \frac1{2^{n-1}} - \frac1{2^n} \land k = n}}{\frac1{2^{n+1}}} \oplus\\
        \wg{\ov{t - h > \frac1{2^{n-1}} \land k = n}}{\frac{2^{n-1}-1}{2^{n+1}}} \oplus
        \wg{\ov{t - h > \frac1{2^{n-1}} - \frac{1}{2^n} \land k = n}}{\frac{2^{n-1}-1}{2^{n+1}}}
      \end{array}}
    \\
    \Rightarrow
    \lrob{\begin{array}{l}
        \wg{\ov{t - h > \frac1{2^{n}} \land k = n}}{\frac1{2^{n+1}}} \oplus
        \wg{\ov{t - h = \frac1{2^n} \land k = n}}{\frac1{2^{n+1}}} \oplus\\
        \wg{\ov{t - h > \frac1{2^{n}} \land k = n}}{\frac{2^{n-1}-1}{2^{n+1}}} \oplus
        \wg{\ov{t - h > \frac{1}{2^n} \land k = n}}{\frac{2^{n-1}-1}{2^{n+1}}}
      \end{array}}
    \\
    \Rightarrow
    \ob{\wg{\ov{t - h = \frac1{2^n} \land k = n}}{\frac1{2^{n+1}}} \oplus \wg{\ov{t - h > \frac1{2^{n}} \land k = n}}{\frac{2^n - 1}{2^{n+1}}}}
    \\
    \;\; k \coloneqq k+1
    \\
    \ob{\wg{\ov{t - h = \frac1{2^n} \land k = n+1}}{\frac1{2^{n+1}}} \oplus \wg{\ov{t - h > \frac1{2^{n}} \land k = n+1}}{\frac{2^n - 1}{2^{n+1}}}}
    \Leftrightarrow \ob{\varphi_{n+1}\oplus\psi_{n+1}}
  \end{array}
\]
\caption{Derivation of the $n\ge 2$ case.}
\label{fig:prob-ex}
\end{figure}
The $\psi_n$ assertions describe the terminating outcomes. The program can terminate after the first step, so $\psi_1 \triangleq \wg{\ov{t = 2 \land h = 2 \land k = 1}}{\frac12}$, but $\psi_n \triangleq \wg{\ov\tru}0$ for every other $n\neq 1$. Finally, $(\psi_n)_{n \in \N} \rightsquigarrow\psi_\infty \triangleq \wg{\ov{h=t}}{\frac12}$, the only terminating outcome. The $\zeta_n$ assertions describe nonterminating outcomes. The program only diverges in the limit, so we get $\zeta_n \triangleq \wg{\ov\tru}0$ for all $n\in\mathbb N$ and $\zeta_\infty \triangleq \mathord{\wg{\Uparrow}{\frac12}}$. Clearly $(\varphi_n \oplus \zeta_n)_{n \in \N} \Uparrow \zeta_\infty$, since according to $\varphi_n$, the probability that $t>h$ after $n$ iterations is $\frac{2^{n-1}-1}{2^n}$, which converges to $\frac12$.
We now establish the premises of the \ruleref{rule:while}{While} rule, ignoring the $\zeta_n$ terms since they are vacuous. We show the cases where $n=0$ and $n=1$ below on the left and right, respectively.
\[\small
\begin{array}{l|l}
\arraycolsep=0pt
  \begin{array}{l}
    \ob{\varphi_0} \Leftrightarrow
    \ob{\ov{t=1\land h=0 \land k=0}}
    \\
    \;\; t\coloneqq t+1 \fatsemi\phantom{x}
    \\
    \ob{\ov{t = 2\land h=0\land k=0}}
    \\
    \;\; (h \coloneqq h + 1) +_\frac12 (h \coloneqq h+ 1 +\frac{1}{2^k}) \fatsemi \phantom{x}
    \\
    \ob{\bigoplus_{i \in \{1, 2\}} \wg{\ov{t=2 \land h=i\land k=0}}{\frac12} }
    \\
    \;\; k \coloneqq k+1
    \\
    \ob{\bigoplus_{i \in \{1, 2\}} \wg{\ov{t=2 \land h=i\land k=1}}{\frac12} }
    \\
    \Leftrightarrow \ob{\varphi_1\oplus\psi_1}
  \end{array}
  &
  \begin{array}{l}
    \ob{\varphi_1} \Leftrightarrow
    \ob{\wg{\ov{t=2\land h=1 \land k=1}}{\frac12}}
    \\
    \;\; t\coloneqq t+1 \fatsemi\phantom{x}
    \\
    \ob{\wg{\ov{t=3\land h=1 \land k=1}}{\frac12}}
    \\
    \;\; (h \coloneqq h + 1) +_\frac12 (h \coloneqq h+ 1 +\frac{1}{2^k}) \fatsemi \phantom{x}
    \\
    \ob{\wg{\ov{t=3 \land h=2\land k=1}}{\frac14} \oplus \wg{\ov{t=3\land h=2+\frac12 \land k=1}}{\frac14}}
    \\
    \;\; k \coloneqq k+1
    \\
    \ob{\wg{\ov{t=3 \land h=2\land k=2}}{\frac14} \oplus \wg{\ov{t=3\land h=2+\frac12 \land k=2}}{\frac14}}
    \\
    \Rightarrow \ob{\wg{\ov{t-h = \frac12 \land k=2}}{\frac14} \oplus \wg{\ov{t-h > \frac12 \land k=2}}{\frac14}}
    \\
    \Leftrightarrow \ob{\varphi_2\oplus\psi_2}
  \end{array}
\end{array}  
\]
Finally, we show the case where $n\ge 2$ in \Cref{fig:prob-ex}.
We complete the proof with a straightforward application of the \ruleref{rule:while}{While} rule.
\[
  \inferrule*[right=\ruleref{rule:while}{While}]{
    \forall n\in\mathbb N.\quad
    \triple{\varphi_n}{C_\mathsf{body}}{\varphi_{n+1}\oplus{\psi_{n+1}}}
  }{
    \triple{\ov{t=1 \land h=0\land k=0}}{\whl{h< t}{C_\mathsf{body}}}{ \wg{\ov{h=t}}{\frac12}\oplus \wg{\Uparrow}{\frac12}}
  }
\]
This tells us that the program terminates in a state where $h=t$ with probability $\frac12$ and it diverges (the hare never catches the tortoise) with probability $\frac12$.

\section{Related Work}

\subsubsection*{Correctness, Incorrectness, and Taxonomies of Program Logics.}

Recent years have seen increased interests in program logics with abilities to reason about nondeterminism in varying ways. This was spurred by Incorrectness Logic \cite{il}, which advocated for using angelic nondeterminism to identify programs that \emph{sometimes} have undesirable behavior. Subsequently, new logics such as Outcome Logic \cite{outcome,zilberstein2024outcome,zilberstein2025demonic,zilberstein2024probabilistic}, Hyper Hoare Logic \cite{hyperhoare}, Exact Separation Logic \cite{maksimovi_c2023exact}, Local Completeness Logic \cite{Bruni2021ALF,brunijacm}, and Quantitative Hyper Weakest Pre \cite{zhang2024quantitative} emerged with the ability to reason about both angelic and and demonic nondeterminism in a single framework.

However, none of the above logics handle nontermination; while they can be used to prove that programs \emph{sometimes} terminate or \emph{always} diverge, they cannot be used to prove that programs \emph{always} terminate, or \emph{sometimes} diverge. The former is the well-known total correctness property \cite{manna1974axiomatic}. The latter is useful in program analysis, as many codebases have nontermination bugs that can be found with new static analysis tools \cite{raad2024nontermproving}.
To our knowledge, Total Outcome Logic is the only logic that subsumes all of the aforementioned abilities.

This proliferation of logics sparked efforts to {\em taxonomize} them in terms of weakening relations, Galois connections, and contrapositives.
For instance, \citet{zhang:quantitative-sp} began to organize logics by their characterizations in terms of predicate transformers. Through this lens, they discovered new logics such as \emph{partial incorrectness logic}. The taxonomy of  \citet{zhang:quantitative-sp}  was extended by \citet{verscht2025taxonomy}, who characterize 16+ Hoare-like logics spanning three dimensions of program analysis: (1) correctness vs. incorrectness, (2) totality vs. partiality, and (3) angelic vs. demonic nondeterminism.
 \citet{cousot2024calculational} described a framework for generating program logics calculationally by composing abstractions, resulting in a similar taxonomy of 12 logics, arranged in a cube, where (non)termination reasoning forms a major axis. 

While both taxonomies support partial and total correctness as well as incorrectness, their focus is to show how these emerge separately as specialized capabilities of disjoint logics. \TOL subsumes both taxonomies, uniting these different forms of reasoning under a single metatheoretic foundation. Moreover, intermediate specifications in \TOL enable reusable proof fragments, as they can be weakened to prove the various triples featured in correctness and incorrectness logics.

\subsubsection*{Weighted Programming.}

 \citet{batz2022weighted}  proposed weighted programming as a paradigm in which each outcome of a program is weighted by an element of a semiring. The initial formulation defined a weakest precondition style calculus for deriving the weight of a single outcome. It used an operational semantics based on total semirings to encode the bounded models that we have in this paper, such as probabilistic programs. While the operational semantics could identify infinite traces, the semantics was not derived from a fixed point.

Outcome Logic extended the model of weighted programming to partial semirings, to encode more models \cite{zilberstein:completeOL}. Outcome Logic also has the ability to reason about multiple executions in a single derivation, so that it can characterize the weights of many outcomes at once. Quantitative Weakest Hyper Pre \cite{zhang2024quantitative} is the predicate transformer analogue of Outcome Logic, as weakest preconditions is to Hoare Logic. However, neither Outcome Logic nor Quantitative Weakest Hyper Pre can identify the weight of nonterminating traces.

\subsubsection*{Powerdomains and Unbounded Nondeterminism.}

Powerdomains \cite{plotkin1976powerdomain,smyth1978power} provide a mechanism to lift a dcpo $\langle X, \le\rangle$ to a dcpo on $\mathcal P(X)$ (the powerset of $X$).
There are three standard powerdomains with different properties and tradeoffs, built from the Hoare and \citet{smyth1978power} orders below, and the Egli-Milner order, which combines the two: $S \sqsubseteq_{\mathsf{EM}} T$ iff $S\sqsubseteq_{\mathsf{H}} T$ and $S\sqsubseteq_{\mathsf{S}} T$.
\[
S\sqsubseteq_{\mathsf{H}} T
\tiff
\forall x\in S.\ \exists y\in T.\ x\le y
\qquad\qquad
S\sqsubseteq_{\mathsf{S}} T
\tiff
\forall y\in T.\ \exists x\in S.\ x\le y
\]
When applied to a flat poset $\langle \Sigma_\nonterm, \le\rangle$\footnote{The flat poset has $\nonterm \le \sigma$ and $\sigma\le\sigma$ for all $\sigma\in\Sigma_\nonterm$, but $\sigma\not\le\tau$ for all $\sigma\neq\tau$.}, we obtain the dcpos below \cite{s_ondergaard1992non}.
\[
\arraycolsep=0pt
\begin{array}{rllll}
  \mathsf{Hoare} \quad &\triangleq\quad & \langle  \{ S \cup \{\nonterm\} \mid S \subseteq \Sigma \}, &  \sqsubseteq_{\mathsf{H}} &\rangle
  \cong \langle \mathcal P(\Sigma), \subseteq \rangle
  \\
  \mathsf{Smyth} \quad&\triangleq& \langle  \mathcal P_{\mathsf{fin}}(\Sigma) \cup \{\{ \nonterm \}\}, &  \sqsubseteq_{\mathsf{S}} &\rangle
  \\
  \mathsf{Plotkin} \quad& \triangleq& \langle \mathcal P_{\mathsf{fin}}(\Sigma) \cup \{ S \subseteq \Sigma \mid \nonterm \in S \}, ~
  &\sqsubseteq_{\mathsf{EM}} &\rangle
\end{array}
\]
Only the Hoare powerdomain supports unbounded nondeterminism, however it cannot identify when nontermination may occur (except when there are no terminating outcomes at all), making it a good semantic basis for partial correctness.
Since $\nonterm$ is always present, the Hoare powerdomain can equivalently be defined as the dcpo on $\mathcal P(\Sigma)$ ordered by subset inclusion.
The semantic model of Outcome Logic \cite{zilberstein:completeOL} can be seen as a generalization of the Hoare powerdomain.

The other two powerdomains can identify nontermination, but do not allow unbounded nondeterminism, stemming from the impossibility result of \citet{plotkin:countable-nondeterminism}.
The \citet{smyth1978power} powerdomain consists of all \emph{finite} sets of terminating outcomes and $\{\nonterm\}$. If nontermination is possible at all, then the semantics is $\{\nonterm\}$, so we cannot distinguish a program that sometimes diverges from one that always does. This makes it a good semantic basis for total correctness.

Finally, the \citet{plotkin1976powerdomain} Powerdomain---derived from the Egli-Milner order---contains all finite subsets of $\Sigma$ and infinite subsets containing $\nonterm$. This is reminiscent of the \TOL semantic model, which maximizes the weight of nontermination for infinite sets. Indeed, in the powerset interpretation, the maximum possible weight is 1, and $m(\nonterm) = 1$ means that $\nonterm$ is in the set of outcomes. As such, our dcpo $\langle \WA^\infty(\Sigma_\nonterm), \sqsubseteq\rangle$ can be viewed as a generalization of the Plotkin powerdomain.

\section{Conclusion}

Recent interest in new program logics for correctness and incorrectness has led to the development of unified reasoning frameworks.  Despite being quite extensive, one dimension that was missing in existing frameworks was the ability to reason about both termination and nontermination, leaving the picture incomplete. In this paper, we proposed Total Outcome Logic as a logic to derive precise specifications about both terminating and nonterminating traces. Total Outcome Logic is generalized over a weighted programming model, meaning that reasoning about (non)termination orthogonally extends to programs with other computational effects, such as probabilistic programs.

Outcome Logic also provides an adequate theory to unify correctness and incorrectness analyses \cite{outcome} and Total Outcome Logic extends this theory further to more types of correctness and incorrectness properties. We showed that \TOL subsumes and extends entire taxonomies of program logics \cite{verscht2025taxonomy,cousot2024calculational}, rendering it a powerful metatheoretic foundation for shared program analysis.

As a next step, it would be interesting to build static analysis algorithms that take advantage of \TOL's expressive power. Inspired by techniques such as bi-abduction \cite{biab,zilberstein2024outcome,realbugs,isl}, our new tool would use formulae of the form shown in \Cref{sec:new-taxonomy} to increase the amount of specification sharing over prior work, while supporting partial correctness, total correctness, incorrectness, nontermination, and more all within a single tool. We believe that this extensible approach will help to scale the implementation and execution of static analysis for many different types of programs.

\bibliographystyle{ACM-Reference-Format}
\bibliography{refs}

\ifx\extended\undefined\else
\allowdisplaybreaks
\appendix
\clearpage

{\noindent \huge\bfseries\sffamily Appendix}

\section{Semantics of Tests}
\label{app:defs}

This definition for the semantics of tests $\detest{b} \colon \Sigma\ \to \{\0,\1\}$ was omitted from \Cref{sec:sem-weight}.
\begin{align*}
    \detest{\tru}(\sigma) &\triangleq \1 
    \\
    \detest{\fls}(\sigma) &\triangleq \0
    \\
    \detest{b_1 \lor b_2}(\sigma) &\triangleq \begin{cases}
        \1 & \text{if}~ \detest{b_1}(\sigma) = \1 ~\text{or}~ \detest{b_2}(\sigma) = \1 \\
        \0 & \text{otherwise}
    \end{cases}
    \\
    \detest{b_1 \land b_2}(\sigma) &\triangleq \begin{cases}
        \1 & \text{if}~ \detest{b_1}(\sigma) =  \detest{b_2}(\sigma) = \1 \\
        \0 & \text{otherwise}
    \end{cases}
    \\
    \detest{\lnot b}(\sigma) &\triangleq \begin{cases}
        \1 & \text{if}~ \detest{b}(\sigma) = \0 \\
        \0 & \text{if}~ \detest{b}(\sigma) = \1
    \end{cases}
    \\
    \detest{t}(\sigma) &\triangleq \begin{cases}
        \1 & \text{if } \sigma \in t \\
        \0 & \text{if } \sigma \not\in t
    \end{cases}
\end{align*}

\section{Domain Properties}
Here, we prove a few preliminary results on our semantic domain $\WA(\Sigma_{\nonterm})$. The first of these describes how the Kleisli extension $(-)^{\dagger}: (X \rightarrow \WA(Y)) \rightarrow \WA(X) \rightarrow \WA(Y)$ interacts with $+$, scalar multiplication $\cdot$, and divergence:

\begin{lemma} \label{lemma:bind-effects}
For any $f, g \in (X \rightarrow \WA(Y))$, $m, m' \in \WA(X)$, and $u \in U$:
    \begin{enumerate}
        \item $f^{\dagger}(m + m') = f^{\dagger}(m) + f^{\dagger}(m')$
        \item If $\supp(m) \subseteq \Sigma$, $(f + g)^{\dagger}(m) = f^{\dagger}(m) + g^{\dagger}(m)$
        \item $f^{\dagger}(u \cdot m) = u \cdot f^{\dagger}(m)$
        \item If $\supp(m) \subseteq \Sigma$, $f^{\dagger}(m \cdot u) = f^{\dagger}(m) \cdot u$
        \item $f^{\dagger}(\eta(\nonterm)) = \eta(\nonterm)$
    \end{enumerate}
\end{lemma}

\begin{proof}
(1), (3), and (4) are left to the reader. We show:
\begin{enumerate}
    \item[(2)] If $\supp(m) \subseteq \Sigma$, then $m(\nonterm) = \0$. We have:
    \begin{align*}
        (f + g)^{\dagger}(m)
        &= \bigg(\sum_{\sigma \in \supp(m) \cap \Sigma} m(\sigma) \cdot (f + g)(\sigma)\bigg) + m(\nonterm) \cdot \eta(\nonterm) \\
        &= \sum_{\sigma \in \supp(m) \cap \Sigma} m(\sigma) \cdot (f(\sigma) + g(\sigma)) \\
        &= \bigg(\sum_{\sigma \in \supp(m) \cap \Sigma} m(\sigma) \cdot f(\sigma)\bigg) + \bigg(\sum_{\sigma \in \supp(m) \cap \Sigma} m(\sigma) \cdot g(\sigma)\bigg) \\
        &= f^{\dagger}(m) + g^{\dagger}(m)
    \end{align*}

    \item[(5)] Note that $\supp(\eta(\nonterm)) = \{\nonterm\}$, so $\supp(\eta(\nonterm)) \cap \Sigma = \emptyset$. Then:
    \begin{align*}
        f^{\dagger}(\eta(\nonterm)) 
        &= \bigg(\sum_{\sigma \in \supp(\eta(\nonterm)) \cap \Sigma} m(\sigma) \cdot f(\sigma) \bigg) + \eta(\nonterm)(\nonterm) \cdot \eta(\nonterm) \\
        &= \eta(\nonterm)(\nonterm) \cdot \eta(\nonterm) \\
        &= \eta(\nonterm)
    \end{align*}
\end{enumerate}
\end{proof}

\section{Totality of Language Semantics}
\todo{Removed multiplication with $\top$ in the definition of $(-)^{\dagger}$. Need to update here.}

\subsection{$\mathcal{D}$-Closure}
To accommodate both a continuous semantics and a (countably) infinite state set, it was necessary to restrict our domain to a subset $\mathcal{D}$ of weighting functions. In particular, we identified two possible classes of semantics with corresponding restrictions on the semiring of weights $\mathcal{A} = \langle X, +, \cdot, \0, \1 \rangle$ used:
\begin{itemize}
    \item \textbf{Conservative}: $\mathcal{A}$ is partial and bounded. Then, we take $\mathcal{D} \triangleq \{m \in \WA^{\infty}(\Sigma) \mid |m| = \top \}$.
    
    \item \textbf{Indicative}: $\mathcal{A}$ is bounded with strongly-infinite top element $\top$. For any other semiring, we can adjoin an element $w \not\in X$ and extend $+$ and $\cdot$ such that $w$ is a strongly-infinite upper bound:
    \begin{itemize}
        \item $w + x = x + w = w$ for all $x \in X$;
        \item $w \cdot x = x \cdot w = w$ for all $x \in X\setminus \{\0\}$;
        \item $w \cdot \0 = \0 \cdot w = \0$.
    \end{itemize}
    In this case, take $\mathcal{D} \triangleq \WA^{\infty}(\Sigma)$. \\
\end{itemize}

In this section we demonstrate that the program semantics is closed in $\mathcal{D}$ for each of the classes described above. 
\subsubsection{Conservative }

\begin{lemma}\label{lem:conservation}
    For any collection $(x_i \in X)_{i \in I}$ such that $\sum_{i \in I} x_i = \top$, if $(m_i)_{i \in I}$ is a family of weighting functions with $|m_i| = \top$ for each $i \in I$, then $\left|\sum_{i \in I} x_i \cdot m_i \right| = \top$.
\end{lemma}
\begin{proof}
    It holds that
    \begin{align*}
        \left|\sum_{i \in I} x_i \cdot m_i \right|
        = \sum_{i \in I} x_i \cdot |m_i|
        = \sum_{i \in I} x_i \cdot T
        = \left(\sum_{i \in I} x_i\right) \cdot \top 
        = \top \cdot \top 
        = \top
    \end{align*}
\end{proof}

\begin{lemma}
    For any $m \in \mathcal{D}$ and $f \in \Sigma \rightarrow \mathcal{D}$, $\bind(m, f) \in \mathcal{D}$.  
\end{lemma}
\begin{proof}
    Recall that $\mathcal{D} \subseteq \WA^{\infty}(\Sigma) = \WA^+(\Sigma) \cup \WA^{\omega}(\Sigma)$. First, suppose that $m, f(\sigma) \in \WA^+(\Sigma)$ (i.e. $m$ and $f(\sigma)$ have finite support) for all $\sigma \in \supp(m) \cap \Sigma$. Recall:
    \[
        \bind(m, f) = m(\nonterm) \cdot \top \cdot \mloop + \sum_{\sigma \in \supp(m) \cap \Sigma} m(\sigma) \cdot f(\sigma)
    \]
    So, 
    \[
        \supp(\bind(m, f)) = (\supp(m) \cap \{\nonterm\}) \cup \bigcup_{\sigma \in \supp(m) \cap \Sigma} \supp(f(\sigma))
    \]
    Since $\supp(m)$ is finite, and $\supp(f(\sigma))$ is finite for each $\sigma \in \supp(m) \cap \Sigma$, the above expression consists of a finite union of finite sets, which is itself finite. Thus, $\bind(m, f) \in \WA^+(\Sigma) \subseteq \WA^{\infty}(\Sigma)$. \\

    Next, we consider the case in which $m \in \WA^{\omega}(\Sigma)$ or $f(\sigma) \in \WA^{\omega}(\Sigma)$ for some $\tau \in \supp(m)\cap\Sigma$. We split our analysis across the two weighting schemes defined (and, accordingly, the two definitions for $\mathcal{D}$ provided for each):
    \begin{itemize}
        \item \textbf{Conservative Weighting}: $\mathcal{D} = \{m \in \WA^{\infty} \mid |m| = \top\}$. \\
        
        Since $m, f(\sigma) \in \mathcal{D}$ for every $\sigma \in \Sigma$, it holds that 
        \begin{itemize} 
            \item $(m(\sigma) \in X)_{\sigma \in \supp(m)}$ is a family of weights such that $\left| \sum_{\sigma \in \supp(m)} m(\sigma) \right| = |m| = \top$;
            \item $|f(\sigma)| = \top$ for each $\sigma \in \supp(m)\cap\Sigma$;
            \item $|\top \cdot \mloop| = \top$. 
        \end{itemize}
        Thus, by \cref{lem:conservation}, $|\bind(m, f)| = \top$. \\
        
        \item \textbf{Indicative Weighting}: $\mathcal{D} = \WA^{\infty}(\Sigma)$. \\
        
        If $m \in \WA^{\omega}(\Sigma)$, then $m(\nonterm) = \sup_{m' \in E(m)} m'(\nonterm) = \top$, indicating the presence of nontermination. As a strong infinity, $\top$ is absorbing and thus propagates through $\bind(m, f)$:
        \begingroup
        \addtolength{\jot}{1em}
        \begin{align*}
            \bind(m, f)(\nonterm) 
            &= m(\nonterm) \cdot \top \cdot \mloop(\nonterm) +\sum_{\sigma \in \supp(m) \cap \Sigma} m(\sigma) \cdot f(\sigma)(\nonterm)  \\
            &= \top \cdot \top \cdot \1 + \sum_{\sigma \in \supp(m) \cap \Sigma} m(\sigma) \cdot f(\sigma)(\nonterm) \\
            &= \top + \sum_{\sigma \in \supp(m) \cap \Sigma} m(\sigma) \cdot f(\sigma)(\nonterm) \\
            &= \top = \sup_{m' \in E(\bind(m, f))} m'(\nonterm)
        \end{align*}    
        \endgroup
        Similarly, if $f(\tau) \in \WA^{\omega}(\Sigma)$ for some $\tau \in \supp(m) \cap \Sigma$:
        \begingroup
        \addtolength{\jot}{1em}
        \begin{align*}
            \bind(m, f)(\nonterm) 
            &= m(\nonterm) \cdot \top \cdot \mloop(\nonterm) + \sum_{\sigma \in \supp(m) \cap \Sigma} m(\sigma) \cdot f(\sigma)(\nonterm) \\
            &=  \big(m(\nonterm) \cdot \top \cdot \1\big) + \big(m(\tau) \cdot f(\tau)(\nonterm)\big) + \sum_{\sigma \in \supp(m) \cap \Sigma/\{\tau\}} m(\sigma) \cdot f(\sigma)(\nonterm) \\
            &= \big(m(\nonterm) \cdot \top \cdot \1\big) + \big(m(\tau) \cdot \top\big) + \sum_{\sigma \in \supp(m) \cap \Sigma/\{\tau\}} m(\sigma) \cdot f(\sigma)(\nonterm) \\
            \intertext{Since $\tau \in \supp(m)$, we know $m(\tau) \neq 0$, so $m(\tau) \cdot \top = \top$:}
            &= \big(m(\nonterm) \cdot \top \cdot \1\big) + \top + \sum_{\sigma \in \supp(m) \cap \Sigma/\{\tau\}} m(\sigma) \cdot f(\sigma)(\nonterm) \\
            &= \top = \sup_{m' \in E(\bind(m, f))} m'(\nonterm)
        \end{align*}
        \endgroup
        Therefore, $\bind(m, f) \in \WA^{\omega}(\Sigma) \subseteq \WA^{\infty}(\Sigma) = \mathcal{D}$. 
    \end{itemize}

\end{proof}

\subsection{Fixpoint Existence}

\begin{theorem}[Scott-Continuity of $\WA(\Sigma_{\nonterm})$]
    Suppose $\mathcal{A}$ is a partial semiring which is bounded, finitary, Scott-continuous, and lower Scott-continuous. Then, $+: \WA(\Sigma_{\nonterm})^2 \rightarrow \WA(\Sigma_{\nonterm})$ on weighting functions is Scott-continuous with respect to the \textbf{fusion order} $\sqsubseteq$. That is, for any directed subset $D \subseteq \WA(\Sigma_{\nonterm})$,
    \[
        \sup_{m_1 \in D}(m_1 + m_2) = \sup D + m_2
    \]
\end{theorem}
\vspace{1em}
\begin{proof}
    Let $\sigma \in \Sigma$. Then,
    \begin{align*}
        \sup_{m_1 \in D}((m_1 + m_2)(\sigma)) 
        &= \sup_{m_1 \in D}(m_1(\sigma) + m_2(\sigma)) \\
        &= \sup_{m_1 \in D}(m_1(\sigma)) + m_2(\sigma) 
        & \text{(Scott-continuity of $\mathcal{A}$)} \\
        &= \big(\sup_{m_1 \in D} m_1(\sigma)\big) + m_2(\sigma) 
    \intertext{Dually,}
        \inf_{m_1 \in D}((m_1 + m_2)(\nonterm)) 
        &= \inf_{m_1 \in D}(m_1(\nonterm) + m_2(\nonterm)) \\
        &= \inf_{m_1 \in D}(m_1(\nonterm) + m_2(\nonterm)
        & \text{(Lower Scott-continuity of $\mathcal{A}$)} \\
        &= \big(\inf_{m_1 \in D} m_1(\nonterm)\big) + m_2(\nonterm) 
    \end{align*}
    Now, since fusion order extends $\leq$ pointwise on $\Sigma$ and $\geq$ on $\nonterm$, we have that $\sup_{m_1 \in D}(m_1 + m_2) = \sup D + m_2$. \\
\end{proof}

\begin{lemma}[Joint Continuity]\label{joint-continuity}
    Suppose $X$ is a dcpo. For any function $f: X \times X \rightarrow X$ that is Scott-continuous in the variables separately, $f$ must be continuous in the variables jointly as well. I.e. for directed subsets $D_1, D_2 \subseteq X$,
    \[
        \sup_{x \in D_1, y \in D_2} f(x, y) = f(\sup D_1, \sup D_2)
    \]
\end{lemma}

\begin{proof}
    Since $f$ is continuous separately in the first and second variables:\\
    \begin{align*}
        f(\sup D_1, \sup D_2)
        = \sup_{x \in D_1} f(x, \sup D_2) 
        = \sup_{x \in D_1} \sup_{y \in D_2} f(x, y) 
        = \sup_{x \in D_1, y \in D_2} f(x, y)
    \end{align*}
\end{proof}

\begin{lemma}\label{lemma:sum-continuity}
    Let $\langle X, +, \cdot, 0, 1 \rangle$ be a finitary (complete), continuous, partial semiring. For any family of Scott-continuous functions $(f_i: X \rightarrow X)_{i \in I}$ and directed set $D \subseteq X$:
    \[
        \sup_{x \in D} \sum_{i \in I} f_i(x) = \sum_{i \in I} f_i(\sup D)
    \]
\end{lemma}
\begin{proof}
    Since each $f_i$ is Scott-continuous, we have $\{f_i(x) \mid x \in D\}$ is a directed set. We now proceed by induction on $I$. 
    \begin{itemize}
        \item Base Case: $I = \{i_1\}$. Then, we only need to show $\sup_{x \in D} f_{i_1}(x) = f_{i_1} (\sup D)$, which follows from the Scott-continuity of $f_{i_1}$.
        
        \item Successor Case: Suppose the claim holds for all set smaller than $I$, for $I$ finite. We can partition $I$ into disjoint non-empty parts $I_1$ and $I_2$, with $I = I_1 \cup I_2$, $I_1 \cap I_2 = \emptyset$, and $I_1, I_2 \ne \emptyset$. So, 
        \begin{align*}
            \sup_{x \in D} \sum_{i \in I} f_i(x) 
            &= \sup_{x \in D}\bigg(\sum_{i \in I_1} f_i(x) + \sum_{i \in I_2} f_i(x) \bigg) \\
            \intertext{By the induction hypothesis, each of $\lambda x.\: \sum_{i \in I_1} f_i(x)$ and $\lambda x.\: \sum_{i \in I_2} f_i(x)$ are continuous, so $\{\sum_{i \in I_1} f_i(x) \mid x \in D\}$ and $\{\sum_{i \in I_2} f_i(x) \mid x \in D\}$ are directed. We apply joint continuity of $+: X^2 \rightarrow X$ to get:}
            \vspace{1em}
            &= \sup_{x \in D} \sum_{i \in I_1} f_i(x) + \sup_{x \in D} \sum_{i \in I_2} f_i(x) \\
            &= \sum_{i \in I_1} f_i(\sup D) + \sum_{i \in I_2} f_i(\sup D) 
            & \text{(IH)} \\
            &= \sum_{i \in I} f_i(\sup D)
        \end{align*}
        
        \item Limit Case: Suppose the claim holds for all finite indexing sets. Since the semiring is finitary:
        \begingroup
        \addtolength{\jot}{0.5em}
        \begin{align*}
            \sup_{x \in D} \sum_{i \in I} f_i(x) 
            &= \sup_{x \in D} \bigg(\sup_{J \subseteq I, \text{ $J$ finite}} \sum_{i \in J} f_i(x)\bigg) \\
            &= \sup_{J \subseteq I, \text{ $J$ finite}} \bigg(\sup_{x \in D} \sum_{i \in J} f_i(x)\bigg) \\
            &= \sup_{J \subseteq I, \text{ $J$ finite}} \bigg(\sum_{i \in J} f_i(\sup D)\bigg) \\
            &= \sum_{i \in I} f_i(\sup D) 
        \end{align*}
        \endgroup
    \end{itemize}
\end{proof}

\begin{corollary}
    For any \textbf{finite} family of Scott-continuous functions $\big(f_i: \WA(\Sigma_{\nonterm})\rightarrow \WA(\Sigma_{\nonterm})\big)_{i \in I}$ and directed set $D \subseteq \WA(\Sigma_{\nonterm})$:
    \[
        \sup_{m \in D} \sum_{i \in I} f_i(m) = \sum_{i \in I} f_i(\sup D)
    \]
\end{corollary}
\begin{proof}
    See the proof of $\ref{lemma:sum-continuity}$ up to the successor case.
\end{proof}

\begin{lemma}
    For any $\sigma \in \Sigma$, the function $g(f) = m(\sigma) \cdot f(\sigma)(\tau)$ is Scott-continuous in the semiring $\mathcal{A}$. 
\end{lemma}
\begin{proof}
    \begin{align*}
        \sup_{f \in D} g(f) 
        &= \sup_{f \in D} m(\sigma) \cdot f(\sigma)(\tau) \\
        \intertext{By Scott-continuity of the $\cdot$ operator in the semiring:}
        &= m(\sigma) \cdot \sup_{f \in D} f(\sigma)(\tau) \\
        \intertext{Since we are using the pointwise* ordering:}
        &= m(\sigma) \cdot (\sup D)(\sigma)(\tau) = g(\sup D) 
    \end{align*}
\end{proof}

\begin{corollary}
    For $D \subseteq (\Sigma \rightarrow \WA(\Sigma_{\nonterm}))$ directed, 
    \[
        \sup_{f \in D} \sum_{\sigma \in \supp(m) \cap \Sigma} m(\sigma) \cdot (\sup D)(\sigma)(\tau)
        = \sum_{\sigma \in \supp(m) \cap \Sigma} m(\sigma) \cdot (\sup D)(\sigma)(\tau)
    \]
\end{corollary}
\begin{proof}
    By Lemma \ref{lemma:sum-continuity}.
\end{proof}

\begin{lemma}
    If $m \in \mathcal{D}$ and $f \in (\Sigma \rightarrow \mathcal{D})$, $\bind(\de{C}(\sigma), f)$ is Scott-continuous with respect to its second argument $f$. 
\end{lemma}
\begin{proof}
    Take $D \subseteq \Sigma \rightarrow \WA(\Sigma_{\nonterm})$ directed. Let $\sigma \in \Sigma$ and $\de{C}(\sigma) \in \mathcal{D}$. Recall that $\mathcal{D} \subseteq \WA^{\infty}(\Sigma) = \WA^+(\Sigma) \cup \WA^{\omega}(\Sigma)$, so we can split our analysis across two cases:
    \begin{itemize}
        \item Suppose $\de{C}(\sigma) \in \WA^+(\Sigma)$. Then, $\supp(\de{C}(\sigma))$ is finite, as is $\supp(\de{C}(\sigma) \cap \Sigma$.
        \begingroup
        \addtolength{\jot}{0.5em}
        \begin{align*}
            &\sup_{f \in D} \bind(\de{C}(\sigma), f) \\
            &= \sup_{f \in D} \bigg(\de{C}(\sigma)(\nonterm) \cdot \top \cdot \mloop + \sum_{\tau \in \supp(\de{C}(\sigma) \cap \Sigma)} \de{C}(\sigma)(\tau) \cdot f(\tau) \bigg) \\
            &= \de{C}(\sigma)(\nonterm) \cdot \top \cdot \mloop + \sup_{f \in D} \bigg(\sum_{\tau \in \supp(\de{C}(\sigma) \cap \Sigma)} \de{C}(\sigma)(\tau) \cdot f(\tau) \bigg) \\
            \intertext{Since the above summation is finite (by assumption), we can apply (Corollary):}
            &= \de{C}(\sigma)(\nonterm) \cdot \top \cdot \mloop + \sum_{\tau \in \supp(\de{C}(\sigma)) \cap \Sigma} \de{C}(\sigma)(\tau) \cdot (\sup D)(\tau) \\
            &= \bind(\de{C}(\sigma), \sup D)
        \end{align*}
        \endgroup

        \item Otherwise, $\de{C}(\sigma) \in \WA^{\omega}(\Sigma)$. We know that for any program state $\tau' \in \Sigma$,
        \begingroup 
        \addtolength{\jot}{0.5em}
        \begin{align*}
            & \sup_{f \in D} \bind(\de{C}(\sigma), f)(\tau') \\
            &= \sup_{f \in D} \bigg(\de{C}(\sigma)(\nonterm) \cdot \top \cdot \mloop(\tau') + \sum_{\tau \in \supp(\de{C}(\sigma)) \cap \Sigma} \de{C}(\sigma)(\tau) \cdot f(\tau)(\tau')\bigg) \\
            &= \de{C}(\sigma)(\nonterm) \cdot \top \cdot \mloop(\tau') + \sup_{f \in D} \bigg(\sum_{\tau \in \supp(\de{C}(\sigma)) \cap \Sigma} \de{C}(\sigma)(\tau) \cdot f(\tau)(\tau')\bigg) \\
            &= \de{C}(\sigma)(\nonterm) \cdot \top \cdot \mloop(\tau') + \sum_{\tau \in \supp(\de{C}(\sigma)) \cap \Sigma} \de{C}(\sigma)(\tau) \cdot (\sup D)(\tau)(\tau') \\
            &= \bind(\de{C}(\sigma), \sup D)(\tau')
        \end{align*}
        \endgroup
        In other words, $\sup_{f \in D} \bind(\de{C}(\sigma), f) \in E(\bind(\de{C}(\sigma), \sup D))$, so by the definition of $\WA^{\omega}(\Sigma)$,
        \begin{align*}
            \sup_{f \in D} \bind(\de{C}(\sigma), f)(\nonterm) \:&\sqsubseteq\: \bind(\de{C}(\sigma), \sup D)(\nonterm)
        \end{align*}
        Conversely, $\bind(\de{C}(\sigma), \sup D) \in E(\sup_{f \in D} \bind(\de{C}(\sigma), f))$ implies 
        \begin{align*}
            \sup_{f \in D} \bind(\de{C}(\sigma), f)(\nonterm) \:&\sqsupseteq\: \bind(\de{C}(\sigma), \sup D)(\nonterm)
        \end{align*}
        Thus, $\sup_{f \in D} \bind(\de{C}(\sigma), f)(\nonterm) = \bind(\de{C}(\sigma), \sup D)(\nonterm)$. 
    \end{itemize}
\end{proof}
\vspace{1em}

\begin{lemma}
    Let $\Phi_{\langle C, e_1, e_2 \rangle}(f)(\sigma) = \de{e_1}(\sigma) \cdot \bindnt(\de{C}(\sigma), f) + \de{e_2}(\sigma) \cdot \eta(\sigma)$ and suppose that it is a total function, then $\Phi_{\iter{C}{e_1}{e_2}}$ is Scott continuous with respect to the pointwise order: $f_1 \sqsubseteq f_2$ iff $f_1(\sigma) \sqsubseteq f_2(\sigma)$ for all $\sigma \in \Sigma$. 
\end{lemma}
\begin{proof}
    For all directed sets $D \subseteq (\Sigma \rightarrow \mathcal{W}(\Sigma \cup \{\nonterm\})$ and $\sigma \in \Sigma$, we have 
    \begin{align*}
        \sup_{f \in D} \Phi_{\iter{C}{e_1}{e_2}}(f)(\sigma) 
        &= \sup_{f \in D} \big(\de{e_1}(\sigma) \cdot \bindnt(\de{C}(\sigma), f) + \de{e_2}(\sigma) \cdot \eta(\sigma)\big)
        \intertext{By the continuity of $+$ and $\cdot$ in $\WA$:}
        &= \de{e_1}(\sigma) \cdot \bigg(\sup_{f \in D} \bindnt(\de{C}(\sigma), f)\bigg) + \de{e_2}(\sigma) \cdot \eta(\sigma)
        \intertext{By Lemma 6 (the previous Lemma):}
        &= \de{e_1}(\sigma) \cdot \bindnt(\de{C}(\sigma), \sup D) + \de{e_2}(\sigma) \cdot \eta(\sigma) \\ 
        &= \Phi_{\iter{C}{e_1}{e_2}}(\sup D)(\sigma) 
    \end{align*}
    As this holds for all $\sigma \in \Sigma$, we also have that
    \[
        \sup_{f \in D} \Phi_{\iter{C}{e_1}{e_2}}(f) = \Phi_{\langle C, e_1, e_2 \rangle}(\sup D)
    \]
\end{proof}
 
\section{Subsumption of Program Logics}
Here, we prove the results of \Cref{sec:vk-taxonomy} on the subsumption of classical program logics. For this, recall that we limit ourselves to a nondeterministic interpretation of TOL, instantiated on the Boolean semiring.
First, we show that the encodings of modalities $\always$ and $\sometimes$ of Dynamic Logic in TOL are DeMorgan duals.

\begin{lemma}[Modal Duality] \label{lemma:modal-duals}
    \[
        \sometimes P = \neg \always \neg P 
        \quad\text{and}\quad 
        \always P = \neg \sometimes \neg P
    \]
\end{lemma}
\begin{proof}
\begin{align*}
    \neg \always \neg P
    &= \neg \{m \mid \supp(m) \subseteq (\neg P)_{\nonterm}\} \\
    &= \WA(\Sigma_{\nonterm}) \setminus \{m \mid \supp(m) \subseteq (\neg P)_{\nonterm}\} \\
    &= \{m \mid \supp(m) \not\subseteq (\neg P)_{\nonterm}\} \\
    &= \{m \mid \supp(m) \cap P \neq \emptyset\} 
    \quad= \sometimes P \\
    \vspace{1em} \\
    \neg \sometimes \neg P 
    &= \neg \{m \mid \supp(m) \cap \neg P \neq \emptyset\} \\
    &= \WA(\Sigma_{\nonterm}) \setminus \{m \mid \supp(m) \cap \neg P \neq \emptyset\} \\
    &= \{m \mid \supp(m) \cap \neg P = \emptyset\} \\
    &= \{m \mid \supp(m) \subseteq P_{\nonterm}\} 
    \quad= \always P 
\end{align*}
\end{proof}
Moreover, the alternative modalities $\alwaystot$ and $\sometimespart$ exhibit the same duality. Proof of this is nearly identical to the above, and so is omitted.  

Theorems \ref{thm:subsume-dwlp-awp} and \ref{thm:subsume-dwp-awlp} state that Hoare, Lisbon, Total Hoare, as well as angelic partial correctness triples can be encoded using the above modalities. These correspond to the logics in the front upper cube of \Cref{fig:vk-taxonomy} which are formulated in terms of weakest-precondition transformers.

\begingroup 
\def\thetheorem{\ref{thm:subsume-dwlp-awp}}
\begin{theorem}[Subsumption of Hoare Logic]
$
    \vDash \triple{\ov P}{C}{\always Q} \tiff P \subseteq \dwlp(C, Q)
$.
\end{theorem}
\addtocounter{theorem}{-1}
\endgroup
\begin{proof}
    $\implies$: Suppose $\sigma \in P$. Then, $\eta(\sigma) \vDash P$. We assume $\triple{\ov P}{C}{\always Q}$, so
    \begin{align*}
        \dem{C}(\eta(\sigma)) 
        &\vDash \always Q 
        \:\:=\:\: \exists (u, v) \in U^2.\: \ov{Q}^{(u)} \oplus \diverge{v} 
        \:\:=\:\: \{m \mid \supp(m) \subseteq Q_{\nonterm} \} 
    \end{align*}
    That is, $\supp(\de{C}(\sigma)) \subseteq Q_{\nonterm}$. By definition $\sigma \in \dwlp(C, Q)$. \\

    \noindent $\impliedby$: We assume $\vDash P \subseteq \dwlp(C, Q)$. Suppose $m \vDash \ov P$. So, $|m| = 1$ and $\supp(m) \subseteq P \subseteq \dwlp(C, Q)$. Note that $\nonterm \not\in \supp(m)$. Then, since $\supp(\de{C}(\sigma)) \subseteq Q_{\nonterm}$ for all $\sigma \in \supp(m)$, we have
    \begin{align*}
        \supp(\dem{C}(m)) 
        &= \bigcup_{\sigma \in \supp(m)} \supp(\de{C}(\sigma))
        \quad\subseteq\quad Q_{\nonterm}
    \end{align*}
    By definition of $\always$, $\dem{C}(m) \vDash \always Q$. 

\end{proof}

\begingroup 
\def\thetheorem{\ref{thm:subsume-dwlp-awp}}
\begin{theorem}[Subsumption of Lisbon Logic]
$
    \vDash \triple{\ov P}{C}{\sometimes Q} \tiff P \subseteq \awp(C, Q)
$.
\end{theorem}
\addtocounter{theorem}{-1}
\endgroup
\begin{proof}
    $\implies$: Suppose $\sigma \in P$. Then, $\eta(\sigma) \vDash P$. Since $\vDash \triple{\ov P}{C}{\sometimes Q}$, we have that 
    \begin{align*}
        \dem{C}(\eta(\sigma)) \vDash \sometimes Q 
        \:\:=\:\: \exists u : U \setminus \{\0\}.\: \ov{Q}^{(u)} \oplus \top 
        \:\:=\:\: \{m \mid \supp(m) \cap Q \neq \emptyset \}
    \end{align*}
    This means $\supp(\de{C}(\sigma)) \cap Q \neq \emptyset$, so $\sigma \in \awp(C, Q)$ by definition. \\

    \noindent $\impliedby$: Assume $P \subseteq \langle Q \rangle$. Suppose $m \vDash \ov P$, i.e. $|m| = 1$, $\supp(m) \subseteq P \subseteq \awp(C, Q)$. Note that $\nonterm \not\in \supp(m)$. So, for any $\sigma \in \supp(m)$, $\supp(\de{C}(\sigma)) \cap Q \neq \emptyset$. It holds that 
    \begin{align*}
        \supp\left(\dem{C}(m)\right) \cap Q
        &= \left( \bigcup_{\sigma \in \supp(m)} \supp(\de{C}(\sigma)) \right) \cap Q 
        = \bigcup_{\sigma \in \supp(m)} \left(\supp(\de{C}(\sigma)) \cap Q \right)
        \neq \emptyset 
    \end{align*}
    By definition, $\dem{C}(m) \vDash \sometimes Q$. 

\end{proof}

\begingroup 
\def\thetheorem{\ref{thm:subsume-dwp-awlp}}
\begin{theorem}[Subsumption of Total Hoare Logic]
$
    \vDash \triple{\ov P}{C}{\alwaystot Q} \tiff P \subseteq \dwp(C, Q)
$.
\end{theorem}
\addtocounter{theorem}{-1}
\endgroup
\begin{proof}
    $\implies$: Suppose $\sigma \in P$. Then, $\eta(\sigma) \vDash \ov P$ and since $\vDash \triple{\ov P}{C}{\alwaystot Q}$, 
    \begin{align*}
        \dem{C}(\eta(\sigma)) \vDash \alwaystot Q 
        = \exists u : U\setminus\{\0\}.\: \ov{Q}^{(u)}
        = \{m \mid \supp(m) \subseteq Q\}
    \end{align*}
    That is, $\supp(\de{C}(\sigma)) \subseteq Q$, so we have $\sigma \in \dwp(C, Q)$. \\

    \noindent $\impliedby$: Suppose $m \vDash \ov P$, i.e. $|m| = 1$ and $\supp(m) \subseteq P \subseteq \dwp(C, Q)$. Note that $\nonterm \not\in \supp(m)$, and 
    \[
        \supp\left(\dem{C}(m)\right)
        = \bigcup_{\sigma \in \supp(m)} \supp(\de{C}(\sigma))
    \]
    For each $\sigma \in \supp(m)$, $\supp(\de{C}(\sigma)) \subseteq Q$. This gives us that $\dem{C}(m) \vDash \alwaystot Q$, as desired. 

\end{proof}

\begingroup 
\def\thetheorem{\ref{thm:subsume-dwp-awlp}}
\begin{theorem}[Subsumption of Angelic Partial Correctness]
\[
    \vDash \triple{\ov P}{C}{\sometimespart Q} \tiff P \subseteq \awlp(P, Q)
\]
\end{theorem}
\addtocounter{theorem}{-1}
\endgroup
\begin{proof}
    $\implies$: Suppose $\sigma \in P$. Then, $\eta(\sigma) \vDash \ov P$ and since $\vDash \triple{\ov P}{C}{\sometimespart Q}$, 
    \begin{align*}
        \dem{C}(\eta(\sigma)) \vDash \sometimespart Q 
        &= \exists (u, v): U^2\setminus\{(\0, \0)\}.\: \ov{Q}^{(u)} \oplus \Uparrow^{(v)} \oplus \top \\
        &= \{m \mid \supp(m) \cap Q_{\nonterm} \neq \emptyset\}
    \end{align*}
    That is, $\supp(\de{C}(\sigma)) \cap Q_\nonterm \neq \emptyset$, so we have $\sigma \in \awlp(C, Q)$. \\

    \noindent $\impliedby$: Suppose $m \vDash \ov P$, i.e. $|m| = 1$ and $\supp(m) \subseteq P \subseteq \awlp(C, Q)$. By definition, for any $\sigma \in \supp(m)$, $\supp(\de{C}(\sigma)) \cap Q_\nonterm \neq \emptyset$. It holds then that
    \[
        \supp\left(\dem{C}(m)\right) \cap Q_\nonterm
        = \left(\bigcup_{\sigma \in \supp(m)} \supp(\de{C}(\sigma))\right) \cap Q_\nonterm
        = \bigcup_{\sigma \in \supp(m)} \left(\supp(\de{C}(\sigma)) \cap Q_\nonterm\right)
        \neq \emptyset
    \]
    This gives us $\dem{C}(m) \vDash \sometimespart Q$, as desired. 

\end{proof}

\section{Soundness and Relative Completeness}
Here, we prove the soundness and (relative) completeness of \TOL. First, some preliminary definitions:
\begin{definition} \label{def:projection}
    For any test $b \in \bb{2}^{\Sigma}$ and $m \in \WA(\Sigma_{\nonterm})$, we define the \textit{projection} of $m$ onto $b$ to be the following weighting function:
    \[
        (\?{b}{m})(\sigma) \triangleq \begin{cases}
            m(\sigma) & \text{if } \detest{b}(\sigma) = \1 \\
            0 & \text{if } \detest{b}(\sigma) = \0 \text{ or } \sigma = \nonterm
        \end{cases}
    \]
\end{definition}
In our proofs, it will be useful to decompose a program configuration $m \in \WA(\Sigma_\nonterm)$ into two -- one part which solely describes weights on program states $\Sigma$, and another which solely describes the degree of nontermination. We introduce the following notational devices:
\begin{definition}
Suppose $m \in \WA(\Sigma_{\nonterm})$. We define $\projsig{m} \triangleq \?{\tru}{m}$ and $\projdiv{m} \triangleq m(\nonterm) \cdot \eta(\nonterm)$ to be the \textup{stateful} and \textup{nonterminating} components of $m$, respectively. These are weighting functions: 
    \begin{align*}
        (\projsig{m})(\sigma) &= \begin{cases}
            m(\sigma) & \text{if } \sigma \in \Sigma \\
            \0 & \text{if } \sigma =\: \nonterm
        \end{cases} & 
        (\projdiv{m})(\sigma) &= \begin{cases}
            \0 & \text{if } \sigma \in \Sigma \\
            m(\sigma) & \text{if } \sigma =\: \nonterm
        \end{cases}
    \end{align*}
\end{definition}
We give some properties for the projection operator: 
\begin{lemma} \label{lemma:projection-properties}
    The projection operator $(\?{b}{-})$ commutes with the following operations:
    \begin{enumerate}
        \item[(i)] $\sum_{i \in I} (\?{b}{m_i}) = \?{b}{\sum_{i \in I} m_i}$
        \item[(ii)] $u \cdot (\?{b}{m}) = \?{b}{u \cdot m}$ and $(\?{b}{m}) \cdot u = \?{b}{m \cdot u}$
    \end{enumerate}
    Moreoever, for any outcome assertion $\varphi$:
    \begin{enumerate}
        \item[(ii)] if for all $m \in \varphi$, $m = \?{b}{m'}$ for some $m'$, then $\varphi \vDash b$
        \item[(iv)] if for all $m \in \varphi$, $m = \projdiv{m'}$ for some $m'$, then $\varphi \vDash \div$. 
    \end{enumerate}
\end{lemma}
\begin{proof}
    (i) and (ii): We show that $\left(u \cdot \sum_{i \in I} (\?{b}{m_i}) \cdot v\right)(\sigma) = \left(\?{b}{u \cdot \sum_{i \in I} m_i} \cdot v\right)(\sigma)$. From the definition of \hyperref[def:projection]{projection}, it is straightforward to see that these are equal to 
    \[
        \begin{cases}
            u \cdot \sum_{i \in I} m_i(\sigma) \cdot v & \text{if } \detest{b}(\sigma) = \1 \\
            \0 & \text{if } \detest{b}(\sigma) = \0 \text{ or } \sigma = \nonterm
        \end{cases}
    \]
    (iii): Suppose that for all $m \in \varphi$, $m = \?{b}{m'}$ for some $m'$. Then, it should be clear that for all $\sigma \in \supp(m) = \supp(\?{b}{m'})$, $\detest{b}(\sigma) = 1$. By definition, $\varphi \vDash b$. (iv) holds by a similar argument.
\end{proof}
\begin{corollary}
    \quad
    $
        \dem{\assume e}(\?{b}{m}) = \?{b}{\left(\dem{\assume e}(m)\right)}
    $
\end{corollary}
\begin{proof}
    \begin{align*}
        &\?{b}{\dem{\assume e}(m)}(\tau) \\
        &\quad = \?{b}{\left[\sum_{\sigma \in \supp(m) \cap \Sigma} m(\sigma) \cdot \de{\assume e}(\sigma) + m(\nonterm) \cdot \eta(\nonterm) \right]} \\
        &\quad = \sum_{\sigma \in \supp(m) \cap \Sigma} \?{b}{\left[m(\sigma) \cdot \de{\assume e}(\sigma)(\tau)\right]} + \?{b}{m(\nonterm) \cdot \eta(\nonterm)} \\
        &\quad =  \sum_{\sigma \in \supp(m) \cap \Sigma} \?{b}{\left[m(\sigma) \cdot \de{e}(\sigma) \cdot \eta(\sigma) \right]} \\
        &\quad = \sum_{\sigma \in \supp(m) \cap \Sigma} \?{b}{\left[m(\sigma)\right] \cdot \de{e}(\sigma) \cdot \eta(\sigma)} \\
        &\quad = \sum_{\sigma \in \supp(\?{b}{m}) \cap \Sigma} \?{b}{\left[m(\sigma)\right] \cdot \de{\assume e}(\sigma)}
    \end{align*}
    Note that this is precisely $\dem{\assume e}(\?{b}{m})$.
\end{proof}

We proceed by proving some results on the semantics of iteration. Recall that the semantics for a command $\iter{C}{e}{e'}$ is given in terms of the characteristic function $\cPhi{C}{e}{e'}: (\Sigma \rightarrow \WA(\Sigma_{\nonterm})) \rightarrow \Sigma \rightarrow \WA(\Sigma_{\nonterm})$:
\[
    \cPhi{C}{e}{e'}(f)(\sigma) = \de{e}(\sigma) \cdot f^{\dagger}\big(\de{C}(\sigma)\big) + \de{e'}(\sigma) \cdot \eta(\sigma)
\]
The following Lemma states that applications of $\cPhi{C}{e}{e'}$ starting on $\bot$ are equivalent to \textit{unrolling} $\iter{C}{e}{e'}$ a corresponding number of times.
\begin{lemma} \label{lemma:n-unrolling}
For all $n \in \N$, $\sigma \in \Sigma$, and $\tau \in \Sigma_{\nonterm}$,
\footnote{Our semantics must capture two computational effects -- weighted execution and nontermination -- each of which compose differently when programs are sequenced. As a result, the unrolling is expressed differently depending on whether we want the final weight on a program state in $\Sigma$ or the final weight of the nontermination outcome. In the case of the former, we aggregate the weight collected for traces terminating within the first $n$ iterations. In the latter, we push any weight of traces not yet terminated (including any divergent weight from potential inner loops) onto $\nonterm$. }
\[
    \cPhi{C}{e}{e'}^{n+1}(\bot)(\sigma)(\tau) = \begin{cases}
        \sum_{k = 0}^n \de{(\assume e \seq C)^n \seq \assume e'}(\sigma)(\tau) & \text{if } \tau \in \Sigma \\
        \bot^{\dagger}(\de{(\assume e \seq C)^{n+1}}(\sigma))(\nonterm) & \text{if } \tau = \:\nonterm
    \end{cases}
\]
Note: For a command $C$, we define $C^0 \triangleq \skp$ and $C^{n+1} \triangleq C^n \seq C$. 
\end{lemma}
\begin{proof}
    We proceed by induction on $n$. 
    \begin{itemize}
        \item $n = 0$. We have that
        \begin{align*}
            \cPhi{C}{e}{e'}(\bot)(\sigma)
            &= \de{e}(\sigma) \cdot \bot^{\dagger}\big(\de{C}(\sigma)\big) + \de{e'}(\sigma) \cdot \eta(\sigma) \\
            &= \bot^{\dagger}(\de{\assume e \seq C}(\sigma)) + \de{\assume e'}(\sigma)
        \end{align*}
        Now, if $\tau \in \Sigma$, then $\bot^{\dagger}(\de{\assume e \seq C)}(\sigma))(\tau) = \0$, since $\supp(\bot^{\dagger}(\de{\assume e \seq C}(\sigma)) = \{\nonterm\}$. On the other hand, $\de{\assume e'}(\sigma)(\nonterm) = \0$. This gives us the two cases, as desired.

        \item \textit{Inductive step}. Suppose the claim holds for $n$. First,
        \begin{align*}
            \cPhi{C}{e}{e'}^{n+2}(\bot)(\sigma) 
            &= \de{e}(\sigma) \cdot \big(\cPhi{C}{e}{e'}^{n+1}(\bot)\big)^{\dagger}\big(\de{C}(\sigma)\big) + \de{e'}(\sigma) + \eta(\sigma) \\
            &= \de{e}(\sigma) \cdot \bigg(\sum_{\tau' \in \supp(\de{C}(\sigma))} \de{C}(\sigma)(\tau') \cdot \cPhi{C}{e}{e'}^{n+1}(\bot)(\tau')\bigg) + \de{e'}(\sigma) \cdot \eta(\sigma)
        \end{align*}
        Now, for $\tau \in \Sigma$, the induction hypothesis gives us:
        \begingroup
        \addtolength{\jot}{1em}
        \footnotesize
        \begin{align*}
            & \cPhi{C}{e}{e'}^{n+2}(\bot)(\sigma)(\tau) \\
            &= \de{e}(\sigma) \cdot \left(\sum_{\tau' \in \supp{\de{C}(\sigma)}} \de{C}(\sigma)(\tau') \cdot \sum_{k=0}^n \de{(\assume e \seq C)^k \seq \assume e'}(\tau')(\tau)\right) + \de{e'}(\sigma) \cdot \eta(\sigma)(\tau) \\
            &= \left(\sum_{\tau' \in \supp(\de{C}(\sigma))} \sum_{k=0}^n \de{e}(\sigma) \cdot \de{C}(\sigma)(\tau') \cdot \de{(\assume e \seq C)^k \seq \assume e'}(\tau')(\tau)\right) + \de{e'}(\sigma) \cdot \eta(\sigma)(\tau) \\
            &= \left(\:\sum_{k=0}^n \sum_{\tau' \in \supp(\de{C}(\sigma))} \de{e}(\sigma) \cdot \de{C}(\sigma)(\tau') \cdot \de{(\assume e \seq C)^k \seq \assume e'}(\tau')(\tau)\right) + \de{e'}(\sigma) \cdot \eta(\sigma)(\tau) \\
            &= \left(\:\sum_{k=1}^{n+1} \de{(\assume e \seq C)^k \seq \assume e'}(\sigma)(\tau) \right) + \de{\assume e'}(\sigma)(\tau) \\
            &= \sum_{k=0}^{n+1} \de{(\assume e \seq C)^k \seq \assume e'}(\sigma)(\tau)
        \end{align*}
        \endgroup
        Applying the other case of the induction hypothesis:
        {\footnotesize
        \begingroup
        \addtolength{\jot}{0.5em}
        \begin{align*}
            &\cPhi{C}{e}{e'}^{n+2}(\bot)(\sigma)(\nonterm) \\
            &= \de{e}(\sigma) \cdot \bigg(\sum_{\tau' \in \supp(\de{C}(\sigma))} \de{C}(\sigma)(\tau') \cdot \bot^{\dagger}(\de{(\assume e \seq C)^{n+1}}(\tau'))(\nonterm)\bigg) + \de{e'}(\sigma) \cdot \eta(\sigma)(\nonterm) \\
            &= \bot^{\dagger}(\de{(\assume e \seq C)^{n+2}}(\sigma))(\nonterm) + \de{e'}(\sigma) \cdot \0 \\
            &= \bot^{\dagger}(\de{(\assume e \seq C)^{n+2}}(\sigma))(\nonterm)
        \end{align*}
        \endgroup
        }%
    \end{itemize}
\end{proof}

We can now reformulate the semantics of $\iter{C}{e}{e'}$ by unrolling the command iteration-by-iteration:
\begin{lemma}[Unrolling] \label{lemma:unrolling}
For all $\sigma \in \Sigma$ and $\tau \in \Sigma_{\nonterm}$, 
\[
    \de{\iter{C}{e}{e'}}(\sigma)(\tau) = \begin{cases}
        \sum_{n \in \N} \de{(\assume e \seq C)^n \seq \assume e'}(\sigma)(\tau) & \text{if } \tau \in \Sigma \\
        \inf_{n \in \N} \bot^{\dagger}(\de{(\assume e \seq C)^n}(\sigma))(\tau) & \text{if } \tau = \:\nonterm
    \end{cases}
\]
\end{lemma}
\begin{proof}
By the program semantics (\cref{fig:inference-command}) and Kleene's fixed point theorem,
\[
    \de{\iter{C}{e}{e'}}(\sigma) = \lfp f.\: \cPhi{C}{e}{e'}(f)(\sigma) = \sup_{n \in N} \cPhi{C}{e}{e'}^n(\bot)(\sigma)
\]
Now, suppose $\tau \in \Sigma$. Since $\cPhi{C}{e}{e'}^0(\bot) = \bot$, or the bottom of the order on $(\Sigma \rightarrow \WA(\Sigma_{\nonterm}))$, we can rewrite the supremum as:
\begin{align*}
    \sup_{n \in \N} \cPhi{C}{e}{e'}^n(\bot)(\sigma)(\tau) 
    &= \sup_{n \in \N} \cPhi{C}{e}{e'}^{n+1}(\bot)(\sigma)(\tau)
    \intertext{By \cref{lemma:n-unrolling}:}
    &= \sup_{n \in \N} \sum_{k=0}^n \de{(\assume e \seq C)^k \seq \assume e'}(\sigma)(\tau) \\
    \intertext{By the definition of infinite sums:}
    &= \sum_{n \in \N} \de{(\assume e \seq C)^n \seq \assume e'}(\sigma)(\tau)
\end{align*}
Finally, we can obtain the remaining case using \autoref{lemma:n-unrolling}:
\begin{align*}
    \sup_{n \in \N} \cPhi{C}{e}{e'}^n(\sigma)(\nonterm) 
    &= \sup_{n \in \N} \big(\bot^{\dagger}(\de{(\assume e \seq C)^n)})(\nonterm)\big) \\
    &= \inf_{n \in \N} \bot^{\dagger}(\de{(\assume e \seq C)^n)})(\nonterm)
\end{align*}
\end{proof}
\begin{theorem}[Soundness]
\[
    \Omega \vdash \triple{\varphi}{C}{\psi} \quad\implies\quad \vDash \triple{\varphi}{C}{\psi}
\]
\end{theorem}
\begin{proof}
    The triple $\triple{\varphi}{C}{\psi}$ is proved using inference rules given in \cref{fig:inference-command} and \cref{fig:inference-structural}. If the last step in this proof makes use of an axiom, then we are done since we took all axioms in $\Omega$ to be semantically valid. Otherwise, we proceed by induction on the derivation of $\Omega \vdash \triple{\varphi}{C}{\psi}$.
    \begin{itemize}
        \item \textsc{False}. We need to show $\vDash \triple{\bot}{C}{\varphi}$. Suppose $m \vDash \bot$. But this is impossible, so the claim holds vacuously. 

        \item \textsc{True}. We need to show $\vDash \triple{\varphi}{C}{\top}$. Suppose $m \vDash \varphi$. It is trivial that $\dem{C}(m) \vDash \top$, so we are done. 

        \item \textsc{Div}. We want to demonstrate $\vDash \triple{\Uparrow^{(u)}}{C}{\Uparrow^{(u)}}$. Suppose $m \vDash \:\Uparrow^{(u)}$. Then, $m = u \cdot \eta(\nonterm)$ by definition. We have
        \begin{align*}
            \dem{C}(m) 
            &= \sum_{\sigma \in \supp(m) \cap \Sigma} \big(u \cdot \eta(\nonterm)\big)(\sigma) \cdot \de{C}(\sigma) + \big(u \cdot \eta(\nonterm)\big)(\nonterm) \cdot \eta(\nonterm) \\
            &= \sum_{\sigma \in \supp(m) \cap \Sigma} \0 \cdot \de{C}(\sigma) + u \cdot \eta(\nonterm) 
            = u \cdot \eta(\nonterm) 
        \end{align*}
        So, $\dem{C}(m) \vDash \:\Uparrow^{(u)}$. 

        \item \textsc{Scale}. By the induction hypothesis, we have $\vDash \triple{\varphi}{C}{\psi}$. We need to show $\vDash \triple{u \odot \varphi}{C}{u \odot \psi}$. Suppose $m \vDash u \odot \varphi$. By definition, $ m = u \cdot m'$ where $m' \vDash \varphi$. Using (iii) of \autoref{lemma:bind-effects}, it follows that
        \[
            \dem{C}(m) = \dem{C}(u \cdot m') = u \cdot \dem{C}(m')
        \]
        Since $\dem{C}(m') \vDash \psi$, we can conclude that $\dem{C}(m) \vDash u \odot \psi$. 

        \item \textsc{Disj}. By the induction hypothesis, $\vDash \triple{\varphi_1}{C}{\psi_1}$ and $\vDash \triple{\varphi_2}{C}{\psi_2}$. We want to show $\vDash \triple{\varphi_1 \lor \varphi_2}{C}{\psi_1 \lor \psi_2}$. Suppose $m \vDash \varphi_1 \lor \varphi_2$. Without loss of generality, take $m \vDash \varphi_1$. The induction hypothesis gives us $\dem{C}(m) \vDash \psi_1$, which can be weakened to $\dem{C}(m) \vDash \psi_1 \lor \psi_2$. The case for $m \vDash \varphi_2$ is symmetric.  

        \item \textsc{Conj}. We have $\vDash \triple{\varphi_1}{C}{\psi_1}$ and $\vDash \triple{\varphi_2}{C}{\psi_2}$ by the induction hypothesis and we want to show $\vDash \triple{\varphi_1 \land \varphi_2}{C}{\psi_1 \land \psi_2}$. Supposing $m \vDash \varphi_1 \land \varphi_2$, it holds that $m \vDash \varphi_1$ and $m \vDash \varphi_2$. Then, $\dem{C}(m) \vDash \psi_1$ and $\dem{C}(m) \vDash \psi_2$. By definition, $\dem{C}(m) \vDash \psi_1 \land \psi_2$, as desired.

        \item \textsc{Choice}. By the induction hypothesis, we have that $\vDash \triple{\phi(t)}{C}{\phi'(t)}$ for all $t \in T$. We now show that $\triple{\bigoplus_{t \in T} \phi(t)}{C}{\bigoplus_{t \in T} \phi'(t)}$. Suppose $m \vDash \bigoplus_{t \in T} \phi(t)$. By definition, this means $m = \sum_{t \in T} m_t$ such that for all $t \in T$, $m_t \vDash \phi(t)$. Since $\sum$ commutes with $(-)^{\dagger}$ in the first argument by \autoref{lemma:bind-effects}, 
        \begin{align*}
            \dem{C}(m) 
            &= \dem{C}\left(\sum_{t \in T} m_t\right) 
            = \sum_{t \in T} \dem{C}(m_t)
        \end{align*}
        For each $m_t$, $\dem{C}(m_t) \vDash \phi'(t)$, giving us $\dem{C}(m) \vDash \bigoplus_{t \in T} \phi'(t)$. 

        \item \textsc{Exists}. By the induction hypothesis, we are given $\forall t \in T.\: \vDash \triple{\phi(t)}{C}{\phi'(t)}$. We need to show $\vDash \triple{\exists t : T.\: \phi(t)}{C}{\exists t : T.\: \phi'(t)}$. Suppose $m \vDash \exists t : T.\: \phi(t)$.
        By definition, this means $m \in \bigcup_{t \in T} \phi(t)$, so there must be some $t \in T$ for which $m \in \phi(t)$. Then, by the induction hypothesis, $\dem{C}(m) \vDash \phi'(t)$, we have that $\dem{C}(m) \vDash \exists t : T.\: \phi'(t)$.

        \item \textsc{Consequence}. Given to us are $\vDash \varphi' \implies \varphi$ and $\vDash \psi \implies \psi'$, and we have $\vDash \triple{\varphi}{C}{\psi}$ by the induction hypothesis. We need to show $\vDash \triple{\varphi'}{C}{\psi'}$. Suppose $m \vDash \varphi'$. Then, $m \vDash \varphi$ and $\dem{C}(m) \vDash \psi$. Finally, this gives us $\dem{C}(m) \vDash \psi'$ as desired. 
        
        \item \textsc{Skip}. We need to show that $\vDash \triple{\varphi}{\skp}{\varphi}$. Suppose that $m \vDash \varphi$. We have that $\dem{\skp}(m) = m$, so $\dem{\skp}(m) \vDash \varphi$ as desired. 
            
        \item \textsc{Seq}. Given $\Omega \vdash \triple{\varphi}{C_1}{\vartheta}$ and $\Omega \vdash \triple{\vartheta}{C_2}{\psi}$, we need to show $\vDash \triple{\varphi}{C_1 \seq C_2}{\psi}$. Suppose $m \vDash \varphi$, so by the induction hypothesis, $\dem{C_1}(m) \vDash \vartheta$ and $\dem{C_2}(\dem{C_1}(m)) \vDash \psi$. Since $\dem{C_2}(\dem{C_1}(m)) = (\dem{C_1} \circ \dem{C_2})(m) = \dem{C_1 \seq C_2}(m)$, we are done.

        \item \textsc{Plus}. Given $\varphi \vDash \1$, $\Omega \vdash \triple{\varphi}{C_1}{\psi_1}$, and $\Omega \vdash \triple{\varphi}{C_2}{\psi_2}$, we need to show $\vDash \triple{\varphi}{C_1 + C_2}{\psi_1 \oplus \psi_2}$. First, suppose $m \vDash \varphi$. Since $\varphi \vDash \1$, it must be that $\supp(m) \subseteq \Sigma$. By \autoref{lemma:bind-effects}, we know that the $(-)^{\dagger}$ operator commutes with $+$ in its first argument: 
        \[
            \dem{C_1 + C_2}(m) = \dem{C_1}(m) + \dem{C_2}(m)
        \]
        By the induction hypothesis, $\dem{C_1}(m) \vDash \psi_1$ and $\dem{C_2}(m) \vDash \psi_2$, so we have $\dem{C_1 + C_2}(m) \vDash \psi_1 \oplus \psi_2$.
        
        \item \textsc{Assume}. Suppose $\varphi \vDash e = u$. Recall that this holds iff 
        \[
            \forall m \in \varphi.\: m(\nonterm) = 0 \text{ and } \forall \sigma \in \supp(m) \cap \Sigma.\: \de{e}(\sigma) = u. 
        \]
        We need to show $\vDash \triple{\varphi}{\assume e}{\varphi \ u}$. Take $m \vDash \varphi$. Then, 
        \begin{align*}
            \dem{\assume e}(m) 
            &= \sum_{\sigma \in \supp(m) \cap \Sigma} m(\sigma) \cdot \de{\assume}(\sigma) + m(\nonterm) \cdot \eta(\nonterm) \\
            &= \sum_{\sigma \in \supp(m) \cap \Sigma} m(\sigma) \cdot \de{e}(\sigma) \cdot \eta(\sigma) + m(\nonterm) \cdot \eta(\nonterm) \\
            &= \sum_{\sigma \in \supp(m) \cap \Sigma} m(\sigma) \cdot u \cdot \eta(\sigma)  + 0 \\
            &= \left(\sum_{\sigma \in \supp(m) \cap \Sigma} m(\sigma) \cdot \eta(\sigma) \right) \cdot u \\
            &= m \cdot u
        \end{align*}
        By definition, $m \cdot u \vDash \varphi \odot u$, as desired. 

        \item \textsc{Iter}. We need to show that $\vDash \triple{\varphi_0 \oplus \zeta_0}{\iter{C}{e}{e'}}{\psi_\infty \oplus \zeta_\infty}$, given the premises. Suppose $m \vDash \varphi_0 \oplus \zeta_0$. By the induction hypothesis, we have that for all $n \in \N$,
        \[
            \vDash \triple{\varphi_n \oplus \zeta_n}{\assume e \seq C}{\varphi_{n+1} \oplus \zeta_{n+1}}
            \qquad
            \vDash \triple{\varphi_n}{\assume e'}{\psi_n}
        \]
        By mathematical induction on $n$, it is straightforward to show that
        \begin{align}
            \dem{(\assume e \seq C)^n}(m) &\vDash \varphi_n \oplus \zeta_n \\
            \dem{(\assume e \seq C)^n \seq \assume e'}(m) &\vDash \psi_n \oplus \zeta_n
        \end{align}
        for all $n$. Note that $\vDash \triple{\zeta_n}{\assume e'}{\zeta_n}$ since $\zeta_n \vDash \div$. In particular, we have that for any $m \vDash \zeta_n$, $\supp(m) \subseteq \{\nonterm\}$, so $\dem{\assume e'}(m) = m$.
        
        Consider each of these results separately:
        \begin{enumerate}
            \item For all $n$, $\dem{(\assume e \seq C)^n \seq \assume e'}(m) \vDash \psi_n \oplus \zeta_n$. Then, there exist $p_n, q_n \in \WA(\Sigma_{\nonterm})$ such that 
            \[
                \dem{(\assume e \seq C)^n \seq \assume e'}(m) = p_n + q_n
            \]
            for $p_n \vDash \psi_n$ and $q_n \vDash \zeta_n$. Since $(\psi_n)_{n \in \N} \rightsquigarrow \psi_{\infty}$, it must be that $\sum_{n \in \N} p_n \vDash \psi_\infty$.
            \vspace{0.5em}
            \item For all $n$, $\dem{(\assume e \seq C)^n \seq \assume e'}(m)\vDash \psi_n \oplus \zeta_n$. Since \\ $(\varphi_n \oplus \zeta_n)_{n \in \N} \Uparrow \zeta_\infty$, we have 
            \[
                \left(\inf_{n \in \N} \big|\dem{(\assume e \seq C)^n}(m)\big| \cdot \top\right) \cdot \eta(\nonterm) \vDash \zeta_\infty
            \]
        \end{enumerate}
        We can now combine the two parts. By \autoref{lemma:unrolling}, the unrolling of $\dem{\iter{C}{e}{e'}}(m)$ is \todo{Need to update unrolling!}
        \begin{align*}
            \dem{\iter{C}{e}{e'}}
            &= \sum_{n \in \N} \projsig{\dem{(\assume e \seq C)^n \seq \assume e'}(m)} \\
            &\quad + \left(\inf_{n \in \N} |\dem{(\assume e \seq C)^n}(m)| \cdot \top \right) \cdot \eta(\nonterm) 
        \end{align*}
        Therefore, $\dem{\iter{C}{e}{e'}}(m) \vDash \psi_\infty \oplus \zeta_\infty$. 
        

    \end{itemize}
\end{proof}

Paving a path to relative completeness, we first prove that triples of the form $\triple{\varphi}{C}{\spost(C, \varphi)}$ can be derived, where $\spost(C, \varphi)$ is the \textit{strongest postcondition} that makes $\triple{\varphi}{C}{\psi}$ true. Defined otherwise: 
\begin{definition}[Strongest Postcondition]
$\spost(C, \varphi)$ consists of all program configurations reachable from $\varphi$ by executing command $C$: 
\[
    \spost(C, \varphi) \triangleq \{\dem{C}(m) \mid m \in \varphi\}
\]
\end{definition}
\begin{lemma} \label{lemma:spost-derivation}
\[
    \Omega \vdash \triple{\varphi}{C}{\spost(C, \varphi)}
\]
\end{lemma}
\begin{proof}
    By induction on the structure of the program $C$.
    \begin{itemize}
        \item $C = \skp$. $\dem{\skp}(m) = m$ for all $m \in \WA(\Sigma_{\nonterm})$, so $\spost(\skp, \varphi) = \varphi$. The desired triple is derived by one application of the \hyperref[fig:inference-command]{\textsc{Skip}} rule:
        \vspace{-1em}
        \begin{mathpar}
            \inferrule*[Right=Skip]{\:}{\triple{\varphi}{C}{\varphi}}
        \end{mathpar}

        \item $C = C_1 \seq C_2$. First, observe that
        \begin{align*}
            \spost(C_1 \seq C_2, \varphi)
            &= \{\dem{C_1 \seq C_2}(m) \mid m \in \varphi\} \\
            &= \{\dem{C_2}(\dem{C_1}(m)) \mid m \in \varphi\} \\
            &= \{\dem{C_2}(m') \mid m' \in \{\dem{C_1}(m) \mid m \in \varphi\}\} \\
            &= \spost(C_2, \spost(C_1, \varphi))
        \end{align*}
        By the induction hypothesis, we have
        \begin{align*}
            \Omega &\vdash \triple{\varphi}{C_1}{\spost(C_1, \varphi)} \\
            \Omega &\vdash \triple{\spost(C_1, \varphi)}{C_2}{\spost(C_2, \spost(C_1, \varphi))}
        \end{align*}
        The derivation is completed with the \hyperref[fig:inference-command]{\textsc{Seq}} rule:
        \begin{mathpar}
            \inferrule*[Right=Seq]
                {\inferrule
                    {\Omega}
                    {\triple{\varphi}{C_1}{\spost(C_1, \varphi)}} \\
                \inferrule
                    {\Omega}
                    {\triple{\spost(C_1, \varphi)}{C_2}{\spost(C_2, \spost(C_1, \varphi))}}
                }
                {\triple{\varphi}{C_1 \seq C_2}{\spost(C_1 \seq C_2, \varphi)}}
        \end{mathpar}

        \item $C = C_1 + C_2$. It holds that 
        \begin{align*}
            \spost(C_1 + C_2, \varphi) 
            &= \big\{\dem{C_1 + C_2}(m) \mid m \in \varphi\big\} \\
            &= \big\{\dem{C_1 + C_2}\big(\projsig{m} + m(\nonterm) \cdot \eta(\nonterm)\big) \mid m \in \varphi\big\} 
        \end{align*}
        By \autoref{lemma:bind-effects} (i), this is equal to:
        \begin{align*}
            &= \big\{\dem{C_1 + C_2}\big(\projsig{m}\big) + \dem{C_1 + C_2}\big(m(\nonterm) \cdot \eta(\nonterm)\big) \mid m \in \varphi\big\} \\
            \intertext{Observe that $\supp(\projsig{m}) \subseteq \Sigma$. Applying (ii) and (iv) of \autoref{lemma:bind-effects}:}
            &= \big\{\dem{C_1}\big(\projsig{m}\big) + \dem{C_2}\big(\projsig{m}\big) + m(\nonterm) \cdot \eta(\nonterm) \mid m \in \varphi \big\} \\
            &= \bigcup_{m \in \varphi} \big\{\dem{C_1}\big(m'\big) + \dem{C_2}\big(m'\big) + m(\nonterm) \cdot \eta(\nonterm) \mid m' \in \bone(\projsig{m})\big\} \\
            &= \exists m : \varphi.\: \spost(C_1, \projsig{m}) \oplus \spost(C_2, \projsig{m}) \:\oplus \Uparrow(m(\nonterm))
        \end{align*}
        We break the derivation into parts. First, note that $\bone(\projsig{m}) \vDash \textsf{term}$ for any $m$. Then, we have the subderivation $(*)$:
        \begin{mathpar}
            \small
            \inferrule*[Right=Plus,sep=-7em]
                {\inferrule*
                    {\Omega}
                    {\triple{\bone(\projsig{m})}{C_1}{\spost(C_1, \bone(\projsig{m}))}} \\
                \inferrule*[vdots=3em]
                    {\Omega}
                    {\triple{\bone(\projsig{m})}{C_2}{\spost(C_2, \bone(\projsig{m}))}}
                }
                {\triple{\bone(\projsig{m})}{C_1 + C_2}{\spost(C_1, \bone(\projsig{m})) \oplus \spost(C_2, \bone(\projsig{m}))}}
        \end{mathpar}
        The full derivation is completed as follows:
        \begin{mathpar}
            \small  
            \inferrule*[Right=Exists]
                {\forall m \in \varphi.\:\:
                \inferrule*[Right=Choice, sep=3em]
                    {\inferrule*[Right=Plus]
                        {(*)}
                        {\triple{\bone(\projsig{m})}{C_1 + C_2}{\ldots}} \\
                    \inferrule*[Right=Div]
                        {\:}
                        {\triple{\Uparrow(m(\nonterm))}{C_1 + C_2}{\Uparrow(m(\nonterm))}}
                    }
                    {\triple{\bone(m)}{C_1 + C_2}{ \spost(C_1, \bone(\projsig{m})) \oplus \spost(C_2, \bone(\projsig{m})) \oplus \Uparrow(m(\nonterm))}}
                }
                {\triple{\varphi}{C_1 + C_2}{\spost(C_1 + C_2, \varphi)}}
        \end{mathpar}

        \item $C = \assume e$, where $e$ must either be a test $b \in \Test$ or weight $u \in U$. First, we see that 
        \begin{align*}
            \spost(\assume e, \varphi) 
            &= \big\{\dem{\assume e}(m) \mid m \in \varphi \big\} \\
            &= \big\{\dem{\assume e}\big(\projsig{m} + m(\nonterm) \cdot \eta(\nonterm)\big) \mid m \in \varphi \big\} \\
            &= \big\{\dem{\assume e}\big(\projsig{m}\big) + m(\nonterm) \cdot \eta(\nonterm) \mid m \in \varphi \big\} 
        \end{align*}
        Recall that, by \autoref{def:projection}, 
        \[
            \projsig{m} 
            = \?{b}{(\projsig{m})} + \?{\neg b}{(\projsig{m})}
            = \?{b}{m} + \?{\neg b}{m}
        \]
        So $\bone(\projsig{m}) = \bone(\?{b}{m}) + \bone(\?{\neg b}{m})$. Then, we have that:
        \begin{align*}
            \spost(\assume b, \varphi)
            &= \big\{\dem{\assume b}(\?{b}{m} + \?{\neg b}{m}) + m(\nonterm) \cdot \eta(\nonterm) \mid m \in \varphi \big\} \\
            &= \big\{\?{b}{m} + m(\nonterm) \cdot \eta(\nonterm) \mid m \in \varphi \big\} \\
            &= \exists m : \varphi.\: \bone(\?{b}{m}) + \Uparrow(m(\nonterm))
        \end{align*}
        It should be clear that $\?{b}{m} \vDash b = \1$ and $\?{\neg b}{m} \vDash b = \0$. Then, we have $(*)$:
        \begin{mathpar}
            \small
            \inferrule*[Right=Choice]
                {\inferrule[Assume]
                    {\bone(\?{b}{m}) \vDash b = \1}
                    {\triple{\bone(\?{b}{m})}{\assume b}{\bone(\?{b}{m}) \odot \1}} \\
                \inferrule[Assume]
                    {\bone(\?{\neg b}{m}) \vDash b = \0}
                    {\triple{\bone(\?{\neg b}{m})}{\assume b}{\bone(\?{b}{m}) \odot \0}}
                }
                {\triple{\bone(\?{b}{m}) \oplus \bone(\?{\neg b}{m})}{\assume b}{\bone(\?{\neg b}{m})}} \\
        \end{mathpar}
        The derivation is completed as follows:
        \begin{mathpar}
            \small
            \inferrule*[Right=Exists]
                {\forall m \in \varphi.\:\:
                \inferrule*[Right=Choice,sep=-2em]
                    {\inferrule*[Right=Choice]
                        {(*)}
                        {\triple{\bone(\?{b}{m}) \oplus \bone(\?{\neg b}{m})}{\assume b}{\bone(\?{\neg b}{m})}} \\
                    \inferrule*[Right=Div,vdots=3em]
                        {\:}
                        {\triple{\Uparrow(m(\nonterm))}{\assume b}{\Uparrow(m(\nonterm))}}
                    }
                    {\triple
                        {\bone(\?{b}{m}) \oplus \bone(\?{\neg b}{m}) \oplus \Uparrow(m(\nonterm))}
                        {\assume b}
                        {\bone(\?{b}{m}) \oplus \Uparrow(m(\nonterm))}
                    }
                }
                {\triple{\varphi}{\assume b}{\spost(\assume b, \varphi)}}
        \end{mathpar}
        Suppose instead that $e$ is a weight $u \in U$. Since $\supp(\projsig{m}) \subseteq \Sigma$, $\dem{\assume u}(\projsig{m}) = (\projsig{m}) \cdot u$ by \autoref{lemma:bind-effects}. So,
        \begin{align*}  
            \spost(\assume u, \varphi)
            &= \big\{(\projsig{m}) \cdot u + m(\nonterm) \cdot \eta(\nonterm) \mid m \in \varphi\big\} \\
            &= \exists m : \varphi.\: \bone((\projsig{m}) \odot u) \:\oplus \Uparrow(m(\nonterm)) 
        \end{align*}
        We know $\bone(\projsig{m}) \vDash u = u$, so the full derivation is given by:
        \begin{mathpar}
            \small
            \inferrule*[Right=Exists]
                {\forall m \in \varphi.\:\:
                \inferrule*[Right=Choice, sep=0em]
                    {\inferrule*[Right=Assume]
                        {\bone(\projsig{m}) \vDash u = u}
                        {\triple{\bone(\projsig{m})}{\assume u}{\bone(\projsig{m}) \odot u}} \\
                    \inferrule*[Right=Div, vdots=3em]
                        {\:}
                        {\triple{\Uparrow(m(\nonterm))}{\assume u}{\Uparrow(m(\nonterm))}}
                    }
                    {\triple{\bone(\projsig{m}) \oplus \Uparrow(m(\nonterm))}{\assume u}{\bone(\projsig{m}) \odot u \:\oplus \Uparrow(m(\nonterm))}}
                }
                {\triple{\varphi}{\assume u}{\spost(\assume u, \varphi)}}
        \end{mathpar}
        
        \item $C = \iter{C}{e}{e'}$. Note that $\varphi = \exists m.\: \varphi.\: \bone(m)$. To use the \cruleref{Iter} rule, we define the corresponding families of assertions, parametric on a weighting function $m$:
        \begin{align*}
            \varphi_n(m) 
            &\triangleq \bone\big(\projsig{\dem{(\assume e \seq\: C)^n}(m)}\big) \\
            \psi_n(m) 
            &\triangleq \bone\big(\projsig{\dem{(\assume e \seq\: C)^n \seq\: \assume e'}(m)}\big) \\
            \zeta_n(m) 
            &\triangleq \:\Uparrow\big(\dem{(\assume e \seq C)^n}(m)(\nonterm)\big) \\
            \psi_{\infty}(m) 
            &\triangleq \bone\big(\projsig{\dem{\iter{C}{e}{e'}}(m)}\big) \\
            \zeta_\infty(m) 
            &\triangleq \:\Uparrow\left(\dem{\iter{C}{e}{e'}}(m)(\nonterm)\right)
        \end{align*}
        Observe that 
        \begin{align*}
            \varphi_n(m) \oplus \zeta_n(m) 
            &= \bone\big(\dem{(\assume e \seq C)^n}(m)\big) 
            = \spost((\assume e \seq C)^n, \bone(m)) \\
            \psi_{\infty}(m) \oplus \zeta_\infty(m) 
            &= \bone\left(\dem{\iter{C}{e}{e'}}(m)\right)
            = \spost(\iter{C}{e}{e'}, \bone(m))
        \end{align*}
        Next, we show that the premises of \cruleref{Iter} hold for these definitions: 
        \begin{itemize}
            \item $(\psi_n(m))_{n \in \N} \rightsquigarrow \psi_\infty(m)$. 
            \vspace{0.5em}

            Consider a family $(m_n)_{n \in \N}$ such that $m_n \vDash \psi_n(m)$ for all $n$. By definition of $\psi_n(m)$, $m_n$ can only be one weighting function: $\dem{(\assume e \seq C)^n \seq \assume e'}(m)$. Then, 
            \begin{align*}
                \sum_{n \in \N} m_n 
                &= \sum_{n \in \N} \projsig{\dem{(\assume e \seq C)^n \seq \assume e'}(m)} 
                \intertext{By \autoref{lemma:unrolling}, this is precisely the stateful component of $\dem{\iter{C}{e}{e'}}(m)$:}
                &= \projsig{\dem{\iter{C}{e}{e'}}(m)} 
                \vDash \bone\left(\dem{\iter{C}{e}{e'}}(m)\right) = \psi_\infty(m)
            \end{align*}

            \item $(\varphi_n(m) \oplus \zeta_n(m))_{n \in \N} \Uparrow \zeta_\infty(m)$. 
            \vspace{0.5em}

            Take a family $(m_n)_{n \in \N}$ such that $m_n \vDash \varphi_n(m) \oplus \zeta_n(m)$ for all $n$. Then, $m_n \vDash \bone\big(\dem{(\assume e \seq C)^n}(m)\big)$, so $m_n = \dem{(\assume e \seq C)^n}(m)$. We have that 
            \begin{align*}
                (\inf_{n \in \N} |m_n| \cdot \top) \cdot \eta(\nonterm)
                &= \left(\inf_{n \in \N} \left|\dem{(\assume e \seq C)^n}(m)\right|\cdot \top \right) \cdot \eta(\nonterm) 
            \end{align*}
            By \autoref{lemma:unrolling}, this is equal to $\projdiv{\dem{\iter{C}{e}{e'}}(m)}$, which certainly satisfies $\Uparrow\left(\dem{\iter{C}{e}{e'}}(m)\right) = \zeta_\infty(m)$. 
            
            \item $\Omega \vdash \triple{\varphi_n(m) \oplus \zeta_n(m)}{\assume e \seq C}{\varphi_{n+1}(m) \oplus \zeta_{n+1}(m)}$.
            \begin{align*}
                \varphi_{n+1}(m) \oplus \zeta_{n+1}(m) 
                &= \{\dem{(\assume e \seq C)^{n+1}}(m)\} \\
                &= \{\dem{\assume e \seq C}(\dem{(\assume e \seq C)^n}(m))\} \\
                &= \{\dem{\assume e \seq C}(m') \mid m' \in \varphi_n(m) \oplus \zeta_n(m)\} \\
                &= \spost(\assume e \seq C, \varphi_n(m) \oplus \zeta_n(m))
            \end{align*}
            By the induction hypothesis, this triple is derivable.  

            \item $\Omega \vdash \triple{\varphi_n(m)}{\assume e'}{\psi_n(m)}$.
            \begin{align*}
                \varphi_n(m) 
                &= \{\projsig{\dem{(\assume e \seq C)^n \seq \assume e'}(m)}\} \\
                &= \{\dem{\assume e'}(\projsig{(\assume e \seq C)^n})\} \\
                &= \{\dem{\assume e'}(m') \mid m' \in \varphi_n(m)\} \\
                &= \spost(\assume e', \varphi_n(m))
            \end{align*}
            By the induction hypothesis, this triple is derivable.

            \item For all $n \in \N$, it is clear that $\varphi_n(m) \vDash \textsf{term}$ and $\zeta_n(m) \vDash \div$. 
        \end{itemize}

        By definition, $\varphi_0(m) = \bone\big(\projsig{m}\big)$ and $\zeta_0(m) = \:\Uparrow(m(\nonterm))$, so 
        \[
            \varphi_0(m) \oplus \zeta_0(m) = \bone(\projsig(m) + \projdiv{m}) = \bone(m)
        \]
        Then, 
        \begin{align*}
            & \varphi = \exists m : \varphi.\: \varphi_0(m) \oplus \zeta_0(m) \\
            & \spost(\iter{C}{e}{e'}, \varphi) = \exists m : \varphi.\: \psi_{\infty}(m) \oplus \zeta_\infty(m)
        \end{align*}
        The derivation $\Omega \vdash \triple{\varphi}{\iter{C}{e}{e'}}{\spost(\iter{C}{e}{e'}, \varphi)}$ proceeds as follows:
        \begin{mathpar}
            \small
            \inferrule*[Right=Exists]
                {\inferrule*[Right=Iter,sep=-4em]
                    {\forall n.\:
                    \inferrule*
                        {\Omega}
                        {\triple{\varphi_n(m) \oplus \zeta_n(m)}{\assume e \seq C}{\varphi_{n+1}(m) \oplus \zeta_{n+1}(m)}} \\
                    \inferrule*[vdots=3em]
                        {\Omega}
                        {\triple{\varphi_n(m)}{\assume e'}{\psi_n(m)}}
                    }
                    {\triple{\varphi_0(m) \oplus \zeta_0(m)}{\iter{C}{e}{e'}}{\psi_{\infty}(m) \oplus \zeta_\infty(m)}}
                }
                {\triple{\varphi}{\iter{C}{e}{e'}}{\spost(\iter{C}{e}{e'}, \varphi)}}
        \end{mathpar}
        
        \item $C = a$. We assumed that $\Omega$ contains all valid triples pertaining to atomic actions $a \in \Act$, so $\Omega \vdash \triple{\varphi}{a}{\spost(a, \varphi)}$ since $\vDash \triple{\varphi}{a}{\spost(a, \varphi)}$. 
    \end{itemize}
\end{proof}

\begin{theorem}[Relative Completeness]
\[
    \vDash \triple{\varphi}{C}{\psi} \implies \vdash \triple{\varphi}{C}{\psi}
\]
\end{theorem}
\begin{proof}
    We show that $\spost(C, \varphi) \implies \psi$. Suppose that $m \vDash \spost(C, \varphi)$. Then by definition, there must exist some $m'$ such that $m = \dem{C}(m')$. Since $\vDash \triple{\varphi}{C}{\psi}$, we know that $m \vDash \psi$ as well. Then, we can apply \sruleref{Consequence} to derive the desired triple:
    \begin{mathpar}
        \inferrule*[Right=Consequence]
            {\inferrule*[Right=\autoref{lemma:spost-derivation}]
                {\Omega}
                {\triple{\varphi}{C}{\spost(C, \varphi)}} \\ \qquad\qquad 
            \spost(C, \varphi) \implies \psi
            }
            {\triple{\varphi}{C}{\psi}}
    \end{mathpar}
\end{proof}

\section{Rule Derivations}

\begin{lemma} \label{lemma:hoare-triple-equiv}
    The following rules are derivable 
    \begin{mathpar}
        \inferrule*[Left=\autoref{lemma:hoare-triple-equiv} - Partial]
            {\triple{\ov P}{C}{\always Q}}
            {\triple{\always P}{C}{\always Q}}
        \qquad
        \inferrule*[Right=\autoref{lemma:hoare-triple-equiv} - Total]
            {\triple{\ov P}{C}{\alwaystot Q}}
            {\triple{\alwaystot P}{C}{\alwaystot Q}}
    \end{mathpar}
\end{lemma}
\begin{proof}
    Observe that $\always P$ is equivalent to:
    \begin{align*}
        \always P 
        &= \exists m : \always P.\: \bone(m) 
        = \exists m : \always P.\: \oplus_{\sigma \in \supp(m)} m(\sigma) \odot \bone(\eta(\sigma))
    \end{align*}
    For $\sigma \in \supp(m)$, $\eta(\sigma) \vDash \ov P$, so $\always P$ implies $\exists m : \always P.\: \bigoplus_{\sigma \in \supp(m)} m(\sigma) \odot \ov P$. Moreover, 
    \begin{align*}
        \always Q 
        &= \exists m : \always P.\: \always Q \\
        &= \exists m: \always P.\: \oplus_{\sigma \in \supp(m)} \always Q 
        \:\:=\:\: \exists m: \always P.\: \oplus_{\sigma \in \supp(m)} m(\sigma) \odot \always Q
    \end{align*}
    The derivation is as follows:
    \begin{mathpar}
        \inferrule*[Right=Consequence]
            {\inferrule*[Right=Exists]
                {\forall m \in \always P. \\
                    \inferrule*[Right=Choice]
                        {\forall \sigma \in \supp(m). \\
                            \inferrule*[Right=Scale]
                                {\triple{\ov P}{C}{\always Q}}
                                {\triple{m(\sigma) \odot \ov P}{C}{m(\sigma) \odot \always Q}}
                        }
                        {\triple{\oplus_{\sigma \in \supp(m)} m(\sigma) \odot \ov P}{C}{\oplus_{\sigma \in \supp(m)} m(\sigma) \odot \always Q}}
                }
                {\triple{\exists m : \always P.\: \oplus_{\sigma \in \supp(m)} m(\sigma) \odot \always P}{C}{\exists m : \ov P.\: \oplus_{\sigma \in \supp(m)} m(\sigma) \odot \always Q}}
            }
            {\triple{\always P}{C}{\always Q}}
    \end{mathpar}
\end{proof}
Note that $\ov P \implies \always P$, so by a simple application of \sruleref{Consequence}, we can derive the converse: 
\begin{mathpar}
    \inferrule*[Right=Consequence]
        {\triple{\always P}{C}{\always Q}}
        {\triple{\ov P}{C}{\always Q}}
\end{mathpar}
So we see that the two triples $\triple{\ov P}{C}{\always Q}$ and $\triple{\always P}{C}{\always Q}$ are really interchangeable in any derivation. Very similar derivations can be obtained to show that $\triple{\ov P}{C}{\alwaystot Q}$ and $\triple{\alwaystot P}{C}{\alwaystot Q}$ equivalent in the logic as well. 

\begin{lemma} \label{lemma:lisbon-triple-equiv}
    The following rule is derivable:
    \begin{mathpar}
        \inferrule*[Right=\autoref{lemma:lisbon-triple-equiv}]
            {\triple{\ov P}{C}{\sometimes Q}}
            {\triple{\sometimes P}{C}{\sometimes Q}}
    \end{mathpar}
\end{lemma}
\begin{proof}
    For each $m \in \sometimes P$, we can fix a state $\sigma_m \in \supp(m)$ such that $\sigma_m \in \ov P$ as well. This must exist by definition of $\sometimes P$. Let $u_m = m(\sigma)$ be its assigned weight. Note that $u_m \neq \0$. Then, we have as consequences:
    \begin{align*}
        \sometimes P
        \iff& \exists m : \sometimes P.\: \bone(m) \\
        \iff& \exists m : \sometimes P.\: (u_m \odot \bone(\eta(\sigma_m))) 
        \oplus \top \\
        \implies& \exists m : \sometimes P.\: (u_m \odot \ov P) \oplus \top \\
        \: \\
        \exists m : \sometimes P.\: (u_m \odot \sometimes Q) \oplus \top 
        \iff& \exists m : \sometimes P.\: \sometimes Q \oplus \top \\
        \iff& \exists m : \sometimes P.\: \sometimes Q \\
        \iff& \sometimes Q 
    \end{align*}
    The derivation is as follows:
    \begin{mathpar}
        \inferrule*[Right=Consequence]
            {\inferrule*[Right=Exists]
                {\forall m \in \ov P. \\
                \inferrule*[Right=Choice, sep=4em]
                    {\inferrule*[Right=Scale]
                        {\triple{\ov P}{C}{\sometimes Q}} 
                        {\triple{u_m \odot \ov P}{C}{u_m \odot \sometimes Q}} \\
                    \inferrule*[Right=True]
                        {\:}
                        {\triple{\top}{C}{\top}}
                    }
                    {\triple{(u_m \odot \ov P) \oplus \top}{C}{(u_m \odot \sometimes Q) \oplus \top}}                
                }
                {\triple{\exists m : \sometimes P.\: (u_m \odot \ov P) \oplus \top}{C}{\exists m : \sometimes P.\: (u_m \odot \ov P) \oplus \top}}
            }
            {\triple{\sometimes P}{C}{\sometimes Q}}
    \end{mathpar}
\end{proof}

\begin{lemma}
    The following rules are derivable:
    \begin{mathpar}
        \inferrule[Seq-Lisbon]
        {\triple{\ov P}{C_1}{\sometimes Q} \\ 
        \triple{\ov Q}{C_2}{\sometimes R}
        }
        {\triple{\ov P}{C_1 \seq C_2}{\sometimes R}}
        \qquad
        \inferrule[Seq-Total-Hoare]
            {\triple{\ov P}{C_1}{\alwaystot Q} \\ 
            \triple{\ov Q}{C_2}{\alwaystot R}
            }
            {\triple{\ov P}{C_1 \seq C_2}{\alwaystot R}} 
    \end{mathpar}
\end{lemma}
\begin{proof}
    We show only the derivation for \ruleref{rule:seq}{Seq-Total-Hoare} -- we derive \ruleref{rule:seq}{Seq-Lisbon} nearly identically using an application of \autoref{lemma:lisbon-triple-equiv}.
    \begin{mathpar}
        \inferrule*[Right=Seq]
            {\triple{\ov P}{C_1}{\alwaystot Q} \\ 
            \inferrule*[Right=\autoref{lemma:hoare-triple-equiv} - Total]
                {\triple{\ov Q}{C_2}{\alwaystot R}}
                {\triple{\alwaystot Q}{C_2}{\alwaystot R}}
            }
            {\triple{\ov P}{C_1 \seq C_2}{\alwaystot R}}
    \end{mathpar}
\end{proof}

\begin{lemma}
    The following rule is derivable:
    \begin{mathpar}
        \inferrule*[Right=If]
            {\varphi_1 \vDash b \\ \triple{\varphi_1}{C_1}{\psi_1} \\
            \varphi_2 \vDash \neg b \\ \triple{\varphi_2}{C_2}{\psi_2}
            }
            {\triple{\varphi_1 \oplus \varphi_2}{\iftf{b}{C_1}{C_2}}{\psi_1 \oplus \psi_2}}
    \end{mathpar}
\end{lemma}
\begin{proof}
    First, we give a subderivation (1):
    \begin{mathpar}
        \inferrule*[Right=Seq]
            {\inferrule*[Right=Choice]
                {\inferrule*[Right=Assume]
                    {\varphi_1 \vDash b}
                    {\triple{\varphi_1}{\assume b}{\varphi_1}} \\
                \qquad
                \inferrule*[Right=Assume]
                    {\varphi_2 \vDash \neg b}
                    {\triple{\varphi_2}{\assume b}{\0 \odot \varphi_2}}
                }
                {\triple{\varphi_1 \oplus \varphi_2}{\assume b}{\varphi_1}} \\
            \qquad
            \triple{\varphi_1}{C}{\psi_1}
            }
            {\triple{\varphi_1 \oplus \varphi_2}{\assume b \seq C_1}{\psi_1}}
    \end{mathpar}
    The proof of (2) is nearly identical. We complete the derivation with:
    \begin{mathpar}
        \inferrule*[Right=Plus]
            {\inferrule*[Right=Seq]
                {(1)}
                {\triple{\varphi_1 \oplus \varphi_2}{\assume b \seq C_1}{\psi_1}} \\
            \inferrule*[Right=Seq]
                {(2)}
                {\triple{\varphi_1 \oplus \varphi_2}{\assume \neg b \seq C_2}{\psi_2}}
            }
            {\triple{\varphi_1 \oplus \varphi_2}{\iftf{b}{C_1}{C_2}}{\psi_1 \oplus \psi_2}}
    \end{mathpar}
\end{proof}

\begin{lemma}
    The following rule is derivable:
    \begin{mathpar}
        \inferrule*[Right=If-Hoare]
            {\triple{\ov{P \land b}}{C_1}{\alwaystot Q} \\ \triple{\ov{P \land \neg b}}{C_2}{\alwaystot Q}}
            {\triple{\ov P}{\iftf{b}{C_1}{C_2}}{\alwaystot Q}} \\
    \end{mathpar}
\end{lemma}
\begin{proof}
    Note that $\alwaystot(P \land b) \vDash b$ and $\alwaystot(P \land \neg b) \vDash \neg b$. The derivation proceeds as follows:
    \begin{mathpar}
        \inferrule*[Right=Consequence]
            {\inferrule*[Right=\ruleref{rule:if}{If}]
                {\inferrule*[right={\autoref{lemma:hoare-triple-equiv}}]
                    {\triple{\ov{P \land b}}{C_1}{\alwaystot Q}}
                    {\triple{\alwaystot(P \land b)}{C_1}{\alwaystot Q}} \\
                \inferrule*[Right={\autoref{lemma:hoare-triple-equiv}}]
                    {\triple{\ov{P \land \neg b}}{C_2}{\alwaystot Q}}
                    {\triple{\alwaystot(P \land \neg b)}{C_2}{\alwaystot Q}}
                }
                {\triple{\alwaystot(P \land b) \oplus \alwaystot(P \land \neg b)}{\iftf{b}{C_1}{C_2}}{\alwaystot Q \oplus \alwaystot Q}}
            }
            {\triple{\ov P}{\iftf{b}{C_1}{C_2}}{\alwaystot Q}}
    \end{mathpar}
\end{proof}

\begin{lemma}
    The following rule is derivable:
    \begin{mathpar}
        \inferrule*[Right=If-Lisbon]
            {\triple{\ov{P \land b}}{C_1}{\sometimes Q} \\ \triple{\ov{P \land \neg b}}{C_2}{\sometimes Q}}
            {\triple{\ov P}{\iftf{b}{C_1}{C_2}}{\sometimes Q}}
    \end{mathpar}
\end{lemma}
\begin{proof}
    Subderivation (1):
    \begin{mathpar}
        \inferrule*[right=\autoref{lemma:lisbon-triple-equiv}]
            {\inferrule*[Right=If]
                {\ov{P \land b} \vDash b \\
                \triple{\ov{P \land b}}{C_1}{\sometimes Q} \\
                \0 \odot \top \vDash \neg b \\
                \inferrule*[Right=Scale]
                    {\inferrule*[Right=True]{\:}{\triple{\top}{C_2}{\top}}}
                    {\triple{\0 \odot \top}{C_2}{\0 \odot \top}}
                }
                {\triple{\ov{P \land b}}{\iftf{b}{C_1}{C_2}}{\sometimes Q}}
            }
            {\triple{\sometimes(P \land b)}{\iftf{b}{C_1}{C_2}}{\sometimes Q}}
    \end{mathpar}
    Part (2) is symmetrical. 
    \begin{mathpar}
        \inferrule*[right=Consequence]
            {\inferrule*[Right=Disj, sep=-5em]
                {\inferrule*[right=\autoref{lemma:lisbon-triple-equiv}]
                    {(1)}
                    {\triple{\sometimes(P \land b)}{\iftf{b}{C_1}{C_2}}{\sometimes Q}} \\
                \inferrule*[Right=\autoref{lemma:lisbon-triple-equiv}, vdots=2em]
                    {(2)}
                    {\triple{\sometimes(P \land \neg b)}{\iftf{b}{C_1}{C_2}}{\sometimes Q}}
                }
                {\triple{\sometimes(P \land b) \lor \sometimes(P \land \neg b)}{\iftf{b}{C_1}{C_2}}{\sometimes Q \lor \sometimes Q}}
            }
            {\triple{\ov P}{\iftf{b}{C_1}{C_2}}{\sometimes Q}}
    \end{mathpar}
\end{proof}

\begin{lemma}
    The following rule is derivable:
    \begin{mathpar}
        \inferrule*[Right=Div*]
            {\zeta \vDash \div}
            {\triple{\zeta}{C}{\zeta}}
    \end{mathpar}
\end{lemma}
\begin{proof}
    Since $\zeta \vDash \div$, we know that for all $m$ such that $m \vDash \zeta$, $\supp(m) \subseteq \{\nonterm\}$. Then, these are equivalent: $\zeta = \exists m : \zeta.\: \diverge{m(\nonterm)}$. The derivation is as follows:
    \begin{mathpar}
        \inferrule*[Right=Consequence]
            {\inferrule*[Right=Exists]
                {\forall m \in \zeta. \\
                \inferrule*[Right=Div]
                    {\:}
                    {\triple{\diverge{m(\nonterm)}}{C}{\diverge{m(\nonterm)}}}
                }
                {\triple{\exists m : \zeta.\: \diverge{m(\nonterm)}}{C}{\exists m : \zeta.\: \diverge{m(\nonterm)}}}
            }
            {\triple{\zeta}{C}{\zeta}}
    \end{mathpar}
\end{proof}

\begin{lemma}
    The following rule is derivable:
    \begin{mathpar}
        \inferrule*[Right=\textsc{While}]
            {(\psi_n)_{n \in \N} \rightsquigarrow \psi_\infty \\
            (\varphi_n \oplus \psi_n \oplus \zeta_n)_{n \in \N} \Uparrow \zeta_\infty \\
            \forall n \in \N. \\
                \varphi_n \vDash b \\ \psi_n \vDash \neg b \\ \zeta_n \vDash \div \\ 
                \triple{\varphi_n \oplus \zeta_n}{C}{\varphi_{n+1} \oplus \psi_{n+1} \oplus \zeta_{n+1}}
            }
            {\triple{\varphi_0 \oplus \psi_0 \oplus \zeta_0}{\whl{b}{C}}{\psi_\infty \oplus \zeta_\infty}}
    \end{mathpar}
\end{lemma}
\begin{proof}
    We will use \ruleref{fig:inference-command}{Iter} to derive this rule. First, we give the following subderivation ($*$):
    \begin{mathpar}
        \mprset {sep=0.5em}
        \inferrule*[Right=Choice]
            {\inferrule*[Left=Assume]
                {\varphi_n \vDash b}
                {\triple{\varphi_n}{\assume b}{\varphi_n}} \\
            \inferrule*[Right=Assume]
                {\psi_n \vDash \neg b}
                {\triple{\psi_n}{\assume b}{\0 \cdot \top}} \\
            \inferrule*[Right=Div*, vdots=3em]
                {\zeta_n \vDash \div}
                {\triple{\zeta_n}{\assume b}{\zeta_n}}
            }
            {\triple{\varphi_n \oplus \psi_n \oplus \zeta_n}{\assume b}{\varphi_n \oplus \zeta_n}}
    \end{mathpar}
    This gives us (1):
    \begin{mathpar}
        \inferrule*[Right=Seq]
            {\inferrule*
                {(*)}
                {\triple{\varphi_n \oplus \psi_n \oplus \zeta_n}{\assume b}{\varphi_n \oplus \zeta_n}} \\
            \triple{\varphi_n \oplus \zeta_n}{C}{\varphi_{n+1} \oplus \psi_{n+1} \oplus \zeta_{n+1}}
            }
            {\triple{\varphi_n \oplus \psi_n \oplus \zeta_n}{\assume b \seq C}{\varphi_{n+1} \oplus \psi_{n+1} \oplus \zeta_{n+1}}}
    \end{mathpar}
    And we derive (2) similarly: 
    \begin{mathpar}
        \mprset{sep=6em}
        \inferrule*[Right=Choice]
            {\inferrule*[Right=Assume]
                {\varphi_n \vDash b}
                {\triple{\varphi_n}{\assume \neg b}{\0 \cdot \varphi_n}} \\
            \inferrule*[Right=Assume]
                {\psi_n \vDash \neg b}
                {\triple{\psi_n}{\assume \neg b}{\psi_n}}
            }
            {\triple{\varphi_n \oplus \psi_n}{\assume \neg b}{\psi_n}}
    \end{mathpar}
    Now, note that since $\varphi_n \vDash b$ and $\psi_n \vDash \neg b$, it must be that $\varphi_n \oplus \psi_n \vDash \tru$. Recall also that $\whl{b}{C}$ is sugar for $\iter{C}{b}{\neg b}$. The full derivation proceeds as follows:
    \begin{mathpar}
        \inferrule*[right=Iter]
            {\forall n \in \N. \quad
            \inferrule*
                {(1)}
                {\triple{\varphi_n \oplus \psi_n \oplus \zeta_n}{\assume b \seq C}{\varphi_{n+1} \oplus \psi_{n+1} \oplus \zeta_{n+1}}}
            \quad
            \inferrule*
                {(2)}
                {\triple{\varphi_n \oplus \psi_n}{\assume \neg b}{\psi_n}}
            }
            {\triple{\varphi_0 \oplus \psi_0 \oplus \zeta_0}{\whl{b}{C}}{\psi_\infty \oplus \zeta_\infty}}
    \end{mathpar}

\end{proof}

\begin{lemma}
    The following rule is derivable:
    \begin{mathpar}
        \inferrule*[Right=\textsc{Invariant}]
            {\triple{\ov{P \land b}}{C}{\always P}}
            {\triple{\ov P}{\whl{b}{C}}{\always(P \land \neg b)}}
    \end{mathpar}
\end{lemma}
\begin{proof}
    We will derive this rule using \ruleref{rule:while}{While}. For all $n \in \N$, let
    \begin{align*}
        \varphi_n &\triangleq \alwaystot(P \land b) &
        \psi_n = \psi_\infty &\triangleq \alwaystot(P \land \neg b) &
        \zeta_n = \zeta_\infty &\triangleq \exists u.\: \diverge{u} 
    \end{align*}
    We show that the necessary premises hold: 
    \begin{itemize}
        \item Take any $(m_n)_{n \in \N}$ such that $m_n \vDash \psi_n$ for all $n$. That is, $m_n \vDash \alwaystot(P \land \neg b)$, so $\emptyset \subset \supp(m_n) \subseteq (P \land \neg b)$. Then,
        \[
            \emptyset \subset \supp\left(\sum_{n \in \N} m_n \right) 
            = \bigcup_{n \in \N} \supp(m_n) \subseteq (P \land \neg b)
        \]
        By definition, $\sum_{n \in \N} m_n \vDash \alwaystot(P \land \neg b) = \psi_{\infty}$. Thus, $(\psi_n)_{n \in \N} \rightsquigarrow \psi_\infty$. 
        \vspace{0.5em}

        \item Again, take any $(m_n)_{n \in \N}$ such that $m_n \vDash \varphi_n \oplus \psi_n \oplus \zeta_n$ for all $n$. Note that $\zeta_\infty = \exists u: U.\: \diverge{u} = \{m \mid \supp(m) \subseteq \{\nonterm\}\}$. Thus, it always holds that
        \[
            \left(\inf_{n \in N} |m_n| \cdot \top \right) \cdot \eta(\nonterm) \vDash \zeta_\infty
        \]
        We have $(\varphi_n \oplus \psi_n \oplus \zeta_n)_{n \in \N} \Uparrow \zeta_\infty$.
        \vspace{0.5em}

        \item Clearly, $\alwaystot(P \land b) \vDash b$, $\alwaystot(P \land \neg b) \vDash \neg b$, and $\exists u : U.\: \diverge{u} \vDash \div$ for all $n$.
    \end{itemize}
    The derivation is as follows:
    \begin{mathpar}
        \inferrule*[right=Consequence]
            {\inferrule*[Right=While]
                {\inferrule*[Right=Consequence]
                    {\inferrule*[Right=\autoref{lemma:hoare-triple-equiv} - Partial]
                        {\triple{\ov{P \land b}}{C}{\always P}}
                        {\triple{\always(P \land b)}{C}{\always P}}
                    }
                    {\triple{\alwaystot(P \land b) \oplus \exists u : U.\: \diverge{u}}{C}{\alwaystot(P \land b) \oplus \alwaystot(P \land \neg b) \oplus \exists u : U.\: \diverge{u}}}
                }
                {\triple{\alwaystot(P \land b) \oplus \alwaystot(P \land \neg b) \oplus \exists u : U.\: \diverge{u}}{\whl{b}{C}}{\alwaystot(P \land \neg b) \oplus \exists u : U.\: \diverge{u}}}
            }
            {\triple{\ov P}{\whl{b}{C}}{\always(P \land \neg b)}}
    \end{mathpar}
\end{proof}

\begin{lemma}
    The following rule is derivable:
    \begin{mathpar}
        \inferrule*[Right=Variant]
            {\forall n < N. \\
                \varphi_{n+1} \vDash b \\ \varphi_0 \vDash \neg b \\
                \triple{\varphi_{n+1}}{C}{\varphi_n}    
            }
            {\triple{\exists N : \N.\: \varphi_N}{\whl{b}{C}}{\varphi_0}}
    \end{mathpar}
\end{lemma}
\begin{proof}
    Once again, we use \ruleref{rule:while}{While}. For all $N$ and $n$, define 
    \begin{align*}
        \varphi_n' &\triangleq \begin{cases}
            \varphi_{N-n} & \text{if } n < N \\
            \0 \odot \top & \text{otherwise}
        \end{cases} & 
        \psi_n &\triangleq \begin{cases}
            \varphi_0 & \text{if } n \in \{N, \infty\} \\
            \0 \odot \top & \text{otherwise}
        \end{cases} &
        \zeta_n = \zeta_\infty &\triangleq  1_0 
    \end{align*}
    We give the necessary premises for the rule:
    \begin{itemize}
        \item Take any $(m_n)_{n \in \N}$ such that $m_n \vDash \psi_n$ for all $n \in \N$. Then, $m_N \vDash \varphi_0$ while $m_n \vDash \0 \odot \top$ for $n \neq N$. So $\sum_{n\in\N} m_n = m_N \vDash \varphi_0$. This gives us $(\psi_n)_{n \in \N} \rightsquigarrow \psi_\infty$.

        \item Take any $(m_n)_{n \in \N}$ such that $m_n \vDash \varphi_n \oplus \psi_n \oplus \zeta_n$. We know $|m_n| = \0$ for all $n \geq N$, so 
        \[
            \left(\inf_{n \in \N} |m_n| \cdot \top\right) \cdot \eta(\nonterm) = \0 \cdot \eta(\nonterm) \vDash \0 \odot \top
        \]
        We have $(\varphi_n \oplus \psi_n \oplus \zeta_n) \vDash \zeta_\infty$. 
        \vspace{0.5em}

        \item We have as premise that for all $n < N$, $\varphi_n, (\0 \odot \top) \vDash b$, so $\varphi_n' \vDash b$. $\varphi_0, (\0 \odot \top) \vDash \neg b$, so $\psi_n \vDash \neg b$. And finally $\zeta_n = \0 \odot \top \vDash \div$. 
    \end{itemize}
    To derive $\triple{\varphi_n'}{C}{\varphi_{n+1}' \oplus \psi_{n+1}}$ for all $n$, we consider two cases: either $n < N$ or $n \geq N$. We take these subderivations as part (1):
    \begin{mathpar}
        \inferrule*
            {\inferrule*
                {\forall m < N.\: \triple{\varphi_{m+1}}{C}{\varphi_m}}
                {\forall n < N.\: \triple{\varphi_{N-n}}{C}{\varphi_{N-(n+1)}}}            
            }
            {\forall n < N.\: \triple{\varphi_n'}{C}{\varphi_{n+1}' \oplus \psi_{n+1}}}

        \inferrule*
            {\inferrule*[Right=Scale]
                {\inferrule*[Right=True]
                    {\:}
                    {\triple{\top}{C}{\top}}
                }
                {\triple{\0 \cdot T}{C}{\0 \odot \top}}
            }
            {\forall n \geq N.\: \triple{\varphi_n'}{C}{\varphi_{n+1}' \oplus \psi_{n+1}}}
    \end{mathpar}
    The derivation is completed as follows:
    \begin{mathpar}
        \inferrule*[Right=Exists]
            {\forall N \in \N.\: \inferrule*[Right=While]
                {\inferrule*
                    {(1)}
                    {\forall n \in N. \triple{\varphi_n'}{C}{\varphi_{n+1}' \oplus \psi_{n+1}}}
                }
                {\triple{\varphi_N}{\whl{b}{C}}{\varphi_0}}
            }
            {\triple{\exists N : \N.\: \varphi_N}{\whl{b}{C}}{\varphi_0}}
    \end{mathpar}
\end{proof}

\begin{lemma}
    The following rule is derivable:
    \begin{mathpar}
        \hspace{-1.7cm} \inferrule[Hoare-Variant]
            {(P \land b) \Rightarrow R > 0 \\
            \forall n \in \N.\:\: \triple{\ov{P \land b \land R = n}}{C}{\alwaystot (P \land R < n)}}
            {\triple{\ov{P \land R \leq N}}{\whl{b}{C}}{\alwaystot (P \land \neg b)}}
    \end{mathpar}
\end{lemma}
\begin{proof}
    We use \ruleref{rule:while}{While} and fix families of assertions $\varphi_n$, $\psi_n$ and $\zeta_n$ for the rule. 
    
    For all $n \in \N$, let $\zeta_n \triangleq \0 \odot \top$ and 
    \vspace{0.5em}
    \begin{align*}
        \varphi_n \triangleq \begin{cases}
            \alwaystot(P \land R \leq (N - n) \land b) & \text{if } n \leq N \\
            \0 \odot \top & \text{if } n > N
        \end{cases} & \qquad & 
        \psi_n \triangleq \begin{cases}
            \alwaystot(P \land R \leq(N - n) \land \neg b) & \text{if } n \leq N \\
            \0 \odot \top & \text{if } n > N
        \end{cases} 
    \end{align*}
    \vspace{0.5em}
    Let $\psi_\infty \triangleq \alwaystot(P \land \neg b)$ and $\zeta_\infty \triangleq \0 \odot \top$. 

    \noindent We clearly have that $\varphi_n \vDash b$ and $\psi_n \vDash \neg b$. Next, we show that the necessary convergence/divergence conditions hold: 
    \begin{itemize}
        \item Take any $(m_n)_{n \in \N}$ such that $m_n \vDash \psi_n$ for all $n$. Then $m_n \vDash \alwaystot(P \land \neg b)$, so $\sum_{n \in \N} m_n \vDash \alwaystot (P \land \neg b)$. This gives us $(\psi_n)_{n \in \N} \rightsquigarrow \psi_\infty$.
        \item Take $(m_n)_{n \in \N}$ such that $m_n \vDash \varphi_n \oplus \psi_n \oplus \zeta_n$ for all $n$. Then for $m > N$, $m_n \vDash \0 \cdot \top$. So, $(\inf_{n \in \N} |m_n| \cdot \top \cdot \eta(\nonterm) = 0 \cdot \eta(\nonterm))$. Thus, $(\varphi_n \oplus \psi_n \oplus \zeta_n) \Uparrow \zeta_\infty$. 
    \end{itemize}
    \vspace{1em}
    Observe that for $n > N$, $\varphi_n \oplus \psi_n \oplus \zeta_n = \ov{P \land R \leq (N - n)}$. In part (1) below, we derive the triple $\triple{\varphi_n \oplus \zeta_n}{C}{\varphi_{n+1} \oplus \psi_{n+1} \oplus \zeta_{n+1}}$:
    \begin{mathpar}
        \inferrule*[right=Exists]
            {\forall m \leq N - n. \\ 
                \inferrule*[Right=Consequence]
                    {\inferrule*[Right=\autoref{lemma:hoare-triple-equiv} - Total]
                        {\triple{\ov{P \land b \land R = m}}{C}{\alwaystot(P \land R < m)}}
                        {\triple{\alwaystot(P \land b \land R = m)}{C}{\alwaystot(P \land R < m)}}
                    }
                    {\triple{\alwaystot(P \land b \land R = m)}{C}{\alwaystot(P \land R \leq (N - (n + 1)))}}
            }
            {\triple{\alwaystot(P \land b \land R \leq (N - n))}{C}{\alwaystot(P \land R \leq (N - (n + 1)))}}
    \end{mathpar}
    We then complete the derivation with an application of \ruleref{rule:while}{While} and \sruleref{consequence}:
    \begin{mathpar}
        \inferrule*[Right=Consequence]
            {\inferrule*[Right=While]
                {\inferrule*
                    {(1)}
                    {\triple{\alwaystot(P \land b \land R \leq (N-n))}{C}{\alwaystot(P \land R \leq (N - (n + 1)))}}                
                }
                {\triple{\alwaystot(P \land R \leq N)}{\whl{b}{C}}{\alwaystot(P \land \neg b)}}
            }
            {\triple{\ov{P \land R \leq N}}{\whl{b}{C}}{\alwaystot(P \land \neg b)}}
    \end{mathpar}
\end{proof}

\begin{lemma}
    The following rules are derivable:
    \begin{mathpar} 
        \inferrule*[right={QInv-Demon}]
            {\ov P \vDash b \\ \triple{\ov P}{C}{\always P}}
            {\triple{\ov P}{\whl{b}{C}}{\always \fls}}
    
        \inferrule*[right={QInv-Angel}]
            {\ov P \vDash b \\ \triple{\ov P}{C}{\sometimes P}}
            {\triple{\ov P}{\whl{b}{C}}{\sometimespart\fls}}
    \end{mathpar}
\end{lemma}
\begin{proof}
    We show the full derivation for \ruleref{rule:quasi-invariant}{QInv-Angel}, omitting that for the demonic counterpart.   
    
    We use the \ruleref{rule:while}{While} rule. For all $n$, take 
    \begin{align*}
        \varphi_n &\triangleq (\exists u : U\setminus\{\0\}.\: {\ov P}^{(u)}) \oplus \alwaystot b &
        \psi_n &\triangleq \alwaystot(\neg b) &
        \zeta_n &= \exists u : U.\: \diverge{u} \\
        && \psi_\infty &\triangleq \top &
        \zeta_\infty &\triangleq \exists u : U\setminus \{\0\}.\: \diverge{u}
    \end{align*}
    Take any family of weighting functions $(m_n)_{n \in \N}$. 
    \begin{itemize}
        \item If $m_n \vDash \psi_n$ for all $n$, it is trivial that $\sum_{n \in \N} m_n \vDash \top$. So $(\varphi_n)_{n \in \N} \rightsquigarrow \psi_\infty$.
        \item Suppose $m_n \vDash \varphi_n \oplus \psi_n \oplus \zeta_n$. Then, $m_n = p + q$ where $p \vDash \exists u : U\setminus\{\0\}.\: {\ov P}^{(u)} = \varphi_n$, so $|m_n| > |p| > 0$. This means that, for weight $v > 0$,
        \[
            \left(\inf_{n \in \N} |m_n \cdot \top|\right) \cdot \eta(\nonterm) \:\:= \:\: v \cdot \eta(\nonterm) \:\:\vDash\:\: \exists u : U\setminus\{\0\}. \diverge{u}
        \]
        Thus, $(\varphi_n \oplus \psi_n \oplus \zeta_n)_{n \in \N} \Uparrow \zeta_\infty$. 
    \end{itemize}
    Since $\ov P \vDash b$, we know $\varphi_n = \exists u : U\setminus\{\0\}.\: {\ov P}^{(u)} \oplus \alwaystot b \vDash b$. Clearly, $\psi_n \vDash \neg b$ and $\zeta_n \vDash \div$.

    Note that we can partition $\top = \alwaystot b \oplus \alwaystot(\neg b) \oplus \exists u : U.\: \diverge{u}$. So, we have
    \begin{align*}
        \varphi_n \oplus \psi_n \oplus \zeta_n 
        &= (\exists u : U\setminus\{\0\}.\: {\ov P}^{(u)} \oplus \alwaystot b) \oplus \alwaystot(\neg b) \oplus \exists u : U.\: \diverge{u} \\
        &= \exists u : U\setminus\{\0\}.\: {\ov P}^{(u)} \oplus \top \\
        &=\sometimes P \\
        \psi_\infty \oplus \zeta_\infty
        &= \exists u : U\setminus\{\0\}.\: \diverge{u} \oplus \top 
        \:\:=\:\: \sometimespart\fls
    \end{align*} 
    Moreover, $\sometimes P \land \always b = \varphi_n \oplus \zeta_n$. The derivation is as follows:
    \begin{mathpar}
        \inferrule*[right=Consequence]
            {\inferrule*[Right=While]
                {\inferrule*[Right=Consequence]
                    {\inferrule*[Right=\autoref{lemma:lisbon-triple-equiv}]
                        {\triple{\ov P}{C}{\sometimes P}}
                        {\triple{\sometimes P}{C}{\sometimes P}}
                    }
                    {\triple{\sometimes P \land \always b}{C}{\sometimes P}}
                }
                {\triple{\sometimes P}{\whl{b}{C}}{\sometimespart\fls}}
            }
            {\triple{\ov P}{\whl{b}{C}}{\sometimespart\fls}}
    \end{mathpar}
\end{proof}

\section{Additional Case Studies}

\subsection{Proving Nontermination}

\begin{figure}[t]

\begin{align*}
      \textsc{Nt} \triangleq \begin{cases}     
      & x := 1 \seq y := 2 \seq \\
    & \code{while $x + y > 1$ do} \\
    & \qquad x := 3 - x \seq \\
    & \qquad y := 3 - y \seq
       \end{cases}
&\qquad 
    \textsc{MallocDiv} \triangleq \begin{cases}
        & \code{byte}* p \seq \\
        & p :=\: (\code{byte}*)\: \code{malloc}(\code{size}+16) \seq \\
        & \code{while $p$ = NULL do} \\
        & \qquad p :=\: (\code{byte}*)\: \code{malloc}(\code{size}+16) \seq\\
    \end{cases}
\end{align*}
\caption{Two non-terminating programs.}\label{fig:ntprograms}
\end{figure}
\cite{raad2024nontermproving} advocated the need for formal methods to prove nontermination, as many industrial codebases have bugs that cause programs to diverge and are difficult to discover with testing. Here, we use some of their examples to demonstrate how \TOL can be used for nontermination proving.
Consider the \textsc{Nt} program in \autoref{fig:ntprograms}. To show that this program indeed always diverges, we make use of the \ruleref{rule:quasi-invariant}{QInv-Demon} rule to derive the triple  $\ob{\ov\tru} \textsc{Nt} \ob{ \always \fls}$.
\begin{align*}
    & \ob{\ov\tru} \\
    & x := 1 \seq y := 2 \seq \\
    & \ob{\ov{x = 1 \land y = 2}} \\
    & \implies \textcolor{blue}{\langle \ov{x + y = 3} \rangle} \\
    & \qquad x := 3 - x \seq \\
    & \qquad \ob{\ov{3 - x + y = 3}} \qquad // \textsc{Assign} \\
    & \qquad y := 3 - y \seq \\
    & \qquad \ob{\ov{(3 - x) + (3 - y) = 3}} \qquad // \textsc{Assign} \\
    & \qquad \implies \ob{\ov{x + y = 3}} \\
    & \qquad \implies \textcolor{blue}{\langle \always(x + y = 3) \rangle} \\
    & \textcolor{red}{\langle \always\fls \rangle}
\end{align*}
As a second example, we will use the rule \ruleref{rule:quasi-invariant}{QInv-Angel} to establish the \textit{possibility} of nontermination in the presence of branching. Consider a loop involving the \textsf{malloc} function in C as in the program   \textsc{MallocDiv} in \autoref{fig:ntprograms}.


  \textsc{MallocDiv} allocates memory in a loop, iterating until a successful call to \textsf{malloc} and $p$ is set to a non-null value. Crucially, \textsf{malloc} \textit{may} fail each time, giving rise to a divergent execution. We can thus see that the assertion $p = \textsf{NULL}$ is a (angelic) quasi-invariant for this loop; we have 
\begin{mathpar}
    \inferrule*[right=\ruleref{rule:quasi-invariant}{QInv-Angel}]
        {\triple{\ov{p = \code{NULL}}}{p :=\: (\code{byte}*)\: \code{malloc}(\code{size}+16)}{\sometimes(p = \textsf{NULL})}}
        {\triple{\ov{p = \code{NULL}}}{\whl{p = \code{NULL}}{\ldots}}{\sometimespart\fls}}
\end{mathpar}
And thus \textsc{MallocDiv} may diverge. Similar proofs of \textit{guaranteed nontermination} can be derived using \ruleref{rule:quasi-invariant}{QInv-Demon}.

\subsection{Full Proof of Quicksort Partition}
\label{app:partition}

\begin{figure}[!t]
\begin{align*}
    & \ob{\ov\tru} \\
    & i := 0 \seq j := n \seq \\
    & \ob{\ov{i = 0 \land j = n}}
    \implies \textcolor{orange}{\langle \ov{\alpha(i) \land \beta(j) \land j - i \geq 0} \rangle} \\
    & \whl{i \leq j}{} \\
    & \qquad \ob{\ov{\alpha(i) \land \beta(j) \land i \leq j \land j - i = m}} \implies \textcolor{blue}{\langle \ov{\alpha(i) \land \beta(j) \land j - i = m} \rangle} \\
    & \qquad \code{if $A[i] \leq p$ then} \\
    & \qquad\qquad \ob{\ov{\alpha(i) \land \beta(j) \land A[i] \leq p \land j - i = m}} 
    \implies \ob{\ov{\alpha(i + 1) \land \beta(j) \land j - i = m}} \\
    & \qquad\qquad i := i + 1 \seq \\
    & \qquad\qquad \ob{\ov{\alpha(i) \land \beta(j) \land j - i = m - 1} }
    \qquad //\:\: \cruleref{Assign} \\
    & \qquad \code{else if $A[j] \geq p$ then} \\
    & \qquad\qquad \ob{\ov{\alpha(i) \land \beta(j) \land A[j] \geq p \land j - i = m}}
    \implies \ob{\ov{\alpha(i) \land \beta(j-1) \land j - i = m}} \\
    & \qquad\qquad j := j - 1 \seq \\
    & \qquad\qquad \ob{\ov{\alpha(i) \land \beta(j) \land j - i = m - 1} }
    \qquad //\:\: \cruleref{Assign} \\
    & \qquad \code{else} \\
    & \qquad\qquad \ob{\ov{\alpha(i) \land \beta(j) \land A[i] > p \land A[j] < p \land j - i = m}} \\
    & \qquad\qquad \code{swap}(A, i, j) \seq \\
    & \qquad\qquad \ob{\ov{\alpha(i) \land \beta(j) \land A[j] > p \land A[i] < p \land j - i = m}} \\
    & \qquad\qquad \implies \ob{\ov{\alpha(i+1) \land \beta(j-1) \land j - i = m}} \\
    & \qquad\qquad i := i + 1 \seq j := j - 1 \\
    & \qquad\qquad \ob{\ov{\alpha(i) \land \beta(j) \land j - i = m - 2} }
    \qquad //\:\: \cruleref{Assign} \\
    & \qquad \ob{\ov{\alpha(i) \land \beta(j) \land j - i < m}} \qquad //\:\: \ruleref{rule:if}{If} \\
    & \qquad \implies \textcolor{blue}{\langle \alwaystot(\alpha(i) \land \beta(j) \land j - i < m) \rangle} \\
    & \textcolor{orange}{\langle \alwaystot(\alpha(i) \land \beta(j) \land i > j) \rangle}
    \qquad //\:\: \textsc{Hoare-Variant}
\end{align*}
\caption{Decorated proof of the \textsc{Partition} program.}
\label{fig:partition}
\end{figure}
\fi
    
\end{document}